%% file: Main.tex
\documentclass[a4paper,cleveref, autoref]{lipics-v2021}
\nolinenumbers

\usepackage{adjustbox}
\usepackage{algorithmicx}
\usepackage{amsmath, bbm}
\usepackage{amsthm}
\usepackage{comment}
\usepackage{color}
\usepackage{enumitem,kantlipsum}
\usepackage{stmaryrd}
\usepackage{xspace}
\usepackage[normalem]{ulem}
\usepackage[page,title]{appendix}
\usepackage[capitalise]{cleveref}
\usepackage{graphicx}
\usepackage{textcase}
\usepackage{microtype}

\usepackage{balance} 
\usepackage[T1]{fontenc}
\usepackage{nimbusmononarrow} 

\usepackage{scalerel} 

\hyphenchar\font=\string"7F
 \setlist{leftmargin=*}

\theoremstyle{plain}

\newtheorem{fact}[lemma]{Fact}

\theoremstyle{remark}

\theoremstyle{definition}

\newenvironment{subproof}[1][\proofname]{%
  \begin{proof}[#1]%
}{%
  \end{proof}%
}

\input{notation.tex}

\title{Finite-Cliquewidth Sets of Existential Rules}
\subtitle{Toward a General Criterion for Decidable yet Highly 
Expressive Querying
}
\titlerunning{Finite-Cliquewidth Sets of Existential Rules}

\author{Thomas Feller}{Technische Universität Dresden, Germany}{thomas.feller@tu-dresden.de}{https://orcid.org/0000-0001-8420-6118}{}

\author{Tim S. Lyon}{Technische Universität Dresden, Germany}{timothy_stephen.lyon@tu-dresden.de}{https://orcid.org/0000-0003-3214-0828}{}

\author{Piotr Ostropolski-Nalewaja}{Technische Universität Dresden, Germany}{piotr.ostropolski-nalewaja@tu-dresden.de}{https://orcid.org/0000-0002-8021-1638}{}

\author{Sebastian Rudolph}{Technische Universität Dresden, Germany}{sebastian.rudolph@tu-dresden.de}{https://orcid.org/0000-0002-1609-2080}{}
\authorrunning{Thomas Feller, Tim S. Lyon, Piotr Ostropolski-Nalewaja, Sebastian Rudolph}
\Copyright{Thomas Feller, Tim S. Lyon, Piotr Ostropolski-Nalewaja, Sebastian Rudolph}

\ccsdesc[500]{Theory of computation~Description logics}
\ccsdesc[500]{Theory of computation~Database query languages (principles)}
\ccsdesc[500]{Theory of computation~Logic and databases}
\ccsdesc[500]{Mathematics of computing~Graph theory}

\keywords{existential rules, 
TGDs,
cliquewidth,
treewidth,
bounded-treewidth sets,
finite-unification sets,
first-order rewritability,
monadic second-order logic,
datalog}

\begin{document}

\maketitle

\begin{abstract}
In our pursuit of generic criteria for decidable ontology-based querying, we introduce \textit{finite-cliquewidth sets} ($\fcs$) of existential rules, a model-theoretically defined class of rule sets, inspired by the \emph{clique\-width} measure from graph theory. By a generic argument, we show that $\bcs$ ensures decidability of entailment for a sizable class of queries (dubbed \emph{DaMSOQs}) subsuming conjunctive queries (CQs). The $\fcs$ class properly generalizes the class of finite-expansion sets ($\fes$), and for signatures of arity ${\leq}2$, the class of bounded-treewidth sets ($\bts$). For higher arities, $\bts$ is only indirectly subsumed by $\fcs$ by means of reification. 
 Despite the generality of $\fcs$, we provide a rule set with decidable CQ entailment (by virtue of first-order-rewritability) that falls outside $\fcs$, thus demonstrating the incomparability of $\fcs$ and the class of finite-unification sets ($\fus$). 
 In spite of this, we show that if we restrict ourselves to single-headed rule sets over signatures of arity ${\leq}2$, then $\fcs$ subsumes $\fus$.
\end{abstract}

\input{long-introduction.tex}

\input{prelims.tex}


\input{decidability.tex}

\input{clique-width.tex}


\input{cw-binary.tex}

\input{binary-case.tex}
\input{higher-arities.tex}

\input{conclusions.tex}


\balance
\bibliographystyle{abbrv}
\bibliography{bibliography}

\clearpage

\noindent{\huge{\scaleobj{1.1}{\bf Appendices}}}
\medskip\nobalance
\appendix
\input{appendices/fo-well-decoration}
\input{appendices/appendix}

\input{appendices/reify}
\input{appendices/appendix51}

\input{appendices/multihead-full}

\end{document}

%% file: notation.tex
\newcommand{\msubsection}[1]{\medskip{\noindent\bf{}#1{}.}}

\newtheorem{innercustomlem}{Lemma}
\newenvironment{customlem}[1]
  {\renewcommand\theinnercustomlem{#1}\innercustomlem}
  {\endinnercustomlem}
  
  \newtheorem{innercustomthm}{Theorem}
\newenvironment{customthm}[1]
  {\renewcommand\theinnercustomthm{#1}\innercustomthm}
  {\endinnercustomthm}

\newcommand{\instance}{\mathcal{I}}
\newcommand{\inst}{\instance}
\newcommand{\database}{\mathcal{D}}
\newcommand{\db}{\database}
\newcommand{\adom}[1]{\fcn{dom}(#1)}
\newcommand{\aadom}[1]{\fcn{dom}(#1)}
\newcommand{\unimod}{\mathcal{U}}
\newcommand{\cnst}{\mathrm{Cnst}}

\newcommand{\relstr}{\mathcal{A}}
\newcommand{\coloring}{\lambda}
\newcommand{\coloringe}{\lambda_\exists}

\newcommand{\cols}{\mathbb{L}}

\newcommand{\term}{T}
\newcommand{\unitn}{\mathbbm{1}_{n}}

\newcommand{\arity}[1]{\fcn{ar}(#1)}
\newcommand{\cw}[1]{\fcn{cw}(#1)}

\newcommand{\disun}{\oplus}
\newcommand{\add}[2]{\pred{Add}_{#1}(#2)}
\newcommand{\addd}[1]{\pred{Add}_{#1}}
\newcommand{\recol}[2]{\pred{Recolor}_{#1}(#2)}

\newcommand{\hism}{h}

\newcommand{\treewidth}[1]{\fcn{tw}(#1)}
\newcommand{\tw}[1]{\fcn{tw}(#1)}

\newcommand{\head}{\fcn{head}}
\newcommand{\body}{\fcn{body}}
\newcommand{\fcomp}{\circ}
\newcommand{\codom}[1]{\fcn{codom}(#1)}
\newcommand{\fcn}[1]{\mathsf{#1}}
\newcommand{\orig}{\fcn{orig}}
\newcommand{\ent}{\fcn{ent}}
\newcommand{\col}{\fcn{col}}

\newcommand{\ruleset}{\calr}
\newcommand{\rs}{\ruleset}

\newcommand{\bts}{\mathbf{bts}}
\newcommand{\bcs}{\mathbf{fcs}}

\newcommand{\fes}{\mathbf{fes}}

\newcommand{\fus}{\mathbf{fus}}

\newcommand{\bcws}{\bcs}
\newcommand{\fcs}{\mathbf{fcs}}
\newcommand{\fc}{\mathbf{fc}}

\newcommand{\chase}{\mathbf{Ch}}
\newcommand{\sa}{\mathbf{Ch}_{\infty}}
\newcommand{\se}{\mathbf{Ch}_{\exists}}
\newcommand{\sep}{\mathbf{Ch}_{\exists+}}
\renewcommand{\so}{\ksat{1}{\sep,\rulesetnormminus}}
\newcommand{\fchain}[1]{\chase(#1)}
\newcommand{\ksat}[2]{\chase_{#1}(#2)}
\newcommand{\sat}[1]{\ksat{\infty}{#1}}

\newcommand{\kbaseg}{(\dbg,\rsg)}

\newcommand{\rew}[1]{\fcn{rew}(#1)}
\newcommand{\rewrs}[2]{\fcn{rew}_{#1}(#2)}

\newcommand{\calr}{\mathcal{R}}

\newcommand{\struc}{\mathcal{A}}

\newcommand{\conset}{\mathrm{con}}
\newcommand{\nullset}{\mathrm{nul}}
\newcommand{\termset}{\fcn{terms}}

\newcommand{\sig}{\Sigma}
\newcommand{\pred}[1]{\mathtt{#1}}
\newcommand{\const}[1]{\mathtt{#1}}
\newcommand{\thephi}{\varphi}

\newcommand{\fr}{\mathit{fr}}
\newcommand{\tree}{T}
\newcommand{\set}[1]{\{#1\}}

\newcommand{\VX}{\vec{X}}

\newcommand{\vk}{\smash{\vec{k}}}
\newcommand{\vs}{\vec{s}}
\newcommand{\ve}{\vec{e}}
\newcommand{\vx}{\vec{x}}
\newcommand{\vX}{\vec{X}}
\newcommand{\vy}{\vec{y}}
\newcommand{\vvv}{\vec{v}}

\newcommand{\vz}{\vec{z}}
\newcommand{\vt}{\vec{t}}

\newcommand{\vl}{\vec{l}}
\newcommand{\vu}{\vec{u}}
\newcommand{\sigi}{\sig_\mathrm{grid}}
\newcommand{\rsg}{\ruleset_\mathrm{grid}}
\newcommand{\dbg}{\database_\mathrm{grid}}
\newcommand{\ginf}{\mathcal{G}_{\infty}}
\newcommand{\sg}{\sat{\db,\rsg}}
\newcommand{\mq}{(q, M)}
\newcommand{\hpred}{\pred{H}}
\newcommand{\vpred}{\pred{V}}
\newcommand{\rpred}{\pred{R}}
\newcommand{\ppred}{\pred{P}}
\newcommand{\epred}{\pred{E}}

\newcommand{\goal}{\mathtt{Goal}}
\newcommand{\rsgrid}{\ruleset_\mathrm{grid}}
\newcommand{\prmark}{M}
\newcommand{\mquery}{Q}
\newcommand{\query}{q}

\newcommand{\thm}{Thm.} 
\newcommand{\ifandonlyif}{\textit{iff} } 
\newcommand{\iffi}{\textit{iff} } 

\definecolor{piotr}{RGB}{250, 100, 100}
\definecolor{thomas}{RGB}{50, 190, 50}
\definecolor{tim}{RGB}{0, 0, 250}
\definecolor{sebastian}{RGB}{153, 50, 204}
\newcommand{\piotr}[1]{\textbf{\textcolor{black}{piotr}}: {\color{piotr}#1}} \newcommand{\thomas}[1]{\textbf{\textcolor{black}{thomas}}: {\color{thomas}#1}} \newcommand{\tim}[1]{\textbf{\textcolor{black}{tim}}: {\color{tim}#1}} \newcommand{\sebastian}[1]{\textbf{\textcolor{black}{sebastian}}: {\color{sebastian}#1}}

\newcommand{\DL}{\scaleobj{0.8}{\mathrm{DL}}}

\newcommand{\rulesetdl}{\ruleset_{\smash{\DL}}}
\newcommand{\rulesetnorm}{\ruleset_{\smash{\DL}}^{\mathrm{rew}}}
\newcommand{\rulesetnormcon}{\ruleset_{\smash{\DL}}^{\mathrm{conn}}}
\newcommand{\rulesetnormminus}{\ruleset_{\smash{\DL}}^{\mathrm{disc}}}

\newcommand{\cutop}{\operatorname{cut}}
\newcommand{\mergeop}{\operatorname{merge}}
\newcommand{\reduceop}{\operatorname{reduce}}
\newcommand{\stepop}{\rightrightarrows}

\newcommand{\neigh}[2]{\mathit{#2^#1}}
\newcommand{\ldist}[2]{\fcn{dist}_{#1}(#2)}

\newcommand{\smallldots}{\ldots}

\newcommand{\defend}{\hfill\ensuremath{\Diamond}}

\newcounter{SHOWNEW}
\newcommand{\ifcmd}[2]{
    \ifthenelse{\value{SHOWNEW}>0}
        {#2}
        {#1}
}

\ifcmd
    {\newcommand{\prm}[1]{\textcolor{red}{#1}}}
    {\newcommand{\prm}[1]{}}
    
\ifcmd
    {}
    {}
    
\ifcmd
    {\newcommand{\pin}[1]{\textcolor{teal}{#1}}}
    {\newcommand{\pin}[1]{#1}}

%% file: long-introduction.tex
\section{Introduction}

The problem of querying under \emph{existential rules}\footnote{Existential rules are also referred to as \emph{tuple-generating dependencies} (TGDs)~\cite{AbiteboulHV95}, \emph{conceptual graph rules}~\cite{SalMug96}, Datalog$^\pm$~\cite{Gottlob09}, and \emph{$\forall \exists$-rules}~\cite{BagLecMugSal11} in the literature.} (henceforth often shortened to \emph{rules}) is a popular topic in the research fields of database theory and knowledge representation. For arbitrary rule sets, query entailment is undecidable~\cite{ChandraLM81}, motivating research into expressive fragments for which decidability can be regained.

A theoretical tool which has not only proven useful for describing querying methods, but also for identifying classes of rule sets with decidable query entailment, is the \emph{chase}~\cite{Chase84}. Given a database $\database$ and a rule set $\ruleset$, the (potentially non-terminating) chase procedure yields a so-called \emph{universal model}~\cite{DeuNasRem08}. Universal models satisfy exactly those queries entailed by $\database$ and $\ruleset$ and thus allow for the reduction of query entailment to query evaluation.

Rule sets admitting finite universal models (a property equivalent to being a \emph{finite-expansion set}~\cite{BagLecMugSal11} or \emph{core-chase terminating}~\cite{DeuNasRem08}) are particularly well-behaved. 
%
%
%
But even in cases where universal models are necessarily infinite, they may still be sufficiently ``tame'' to allow for decidable query entailment, as is the case with the class of \emph{bounded-treewidth sets} ($\bts$)~\cite{BagLecMugSal11}. A rule set $\ruleset$ qualifies as $\bts$ \iffi for every database~$\database$, there exists a universal model of $\database$ and $\ruleset$ whose \emph{treewidth} (a structural measure originating from graph theory) is bounded by some $n\in \mathbb{N}$. The $\bts$ class subsumes class of finite expansion sets ($\fes$), and gives rise to decidable query entailment for a sizable family of concrete, syntactically defined classes of rule sets deploying more or less refined versions of guardedness~\cite{BagLecMugSal11,CaliGK13,KrotzschR11}.     

While $\bts$ is a fairly general class, it still fails to contain rather simple rule sets; e.g., the rule set $\ruleset^\infty_\mathrm{tran}$, consisting of the two rules
$$
\epred(x,y) \to \exists z \epred(y,z) \quad\text{and}\quad \epred(x,y) \wedge \epred(y,z) \to \epred(x,z),
$$
falls outside the $\bts$ class. To give an idea as to why this holds, consider any database $\db$ containing an $\epred$-fact (e.g., $\epred(\const{a},\const{b})$): the chase will yield a universal model resembling the transitive closure of an infinite $\epred$-path, a clique-like structure of infinite treewidth. What is more, not only does $\ruleset^\infty_\mathrm{tran}$ fail to be $\bts$, it also fails to fall under the other generic decidability criteria: it is neither $\fus$ (described below) nor \emph{finitely controllable} (which would guarantee the existence of a finite ``countermodel'' for each non-entailed query). In spite of this, results for description logics confirm the decidability of (conjunctive) query entailment over $\ruleset^\infty_\mathrm{tran}$ (see~\cite{GlimmLHS08,OrtizRS11}), in\-cen\-ti\-vi\-zing a generalization of the decidability criteria mentioned above.


On a separate, but related note: the $\bts$ class is incomparable with the class of \emph{finite-unification sets} ($\fus$)~\cite{BagLecMugSal11}, another class giving rise to decidable query entailment thanks to \emph{first-order rewritability}. Such rule sets may feature non-guarded rules \cite{CaliGP12}; for instance, they may include \emph{concept products} \cite{BourhisMP17,RudolphKH08}, which create a biclique linking instances of two unary predicates, as in 
$$
\pred{Elephant}(x) \wedge \pred{Mouse}(y) \to \pred{BiggerThan}(x,y).
$$

The above examples demonstrate a crucial weakness of the $\bts$ class, namely, its inability to tolerate universal models exhibiting ``harmless'' unbounded clique-like structures. Opportunely, the graph-theoretic notion of \emph{cliquewidth}~\cite{CourEngBook} overcomes this problem, while retaining most of the desirable properties associated with the notion of treewidth. Inspired by this less mainstream, but more powerful concept, we set out to introduce \emph{finite-cliquewidth sets}~($\bcws$) of rules. As the original cliquewidth notion is tailored to finite undirected graphs, some (not entirely straightforward) generalizations are necessary to adapt it to countable instances, while at the same time, preserving its advantageous properties.

Our contributions can be summarized as follows:
\begin{itemize}
    \item We introduce an abstract framework showing how to utilize specific types of model-theoretic measures to establish the decidability of query entailment for a comprehensive class of queries (dubbed \emph{datalog / monadic second-order queries}, or \emph{DaMSOQs} for short) that significantly extends the class of conjunctive queries. 
    
    \item Generalizing the eponymous notion from finite graph theory, we introduce the model-theoretic measure of \emph{cliquewidth} for countable instances over arbitrary signatures. Based on our framework, we show that the derived notion of \emph{finite-cliquewidth sets} 
    guarantees decidability of DaMSOQ entailment. In particular, we demonstrate that $\ruleset^\infty_\mathrm{tran}$ is indeed $\fcs$, thus showing that $\fcs$ incorporates rule sets outside $\bts$ and $\fus$.
    
    \item We compare the $\fcs$ and $\bts$ classes, obtaining: (good news) for binary signatures, $\fcs$ subsumes $\bts$, which (bad news) does not hold for higher-arity signatures, but (relieving news) $\fcs$ still ``indirectly subsumes'' higher-arity $\bts$ through reification.
    \item We compare the $\fcs$ and $\fus$ classes, obtaining: (good news) for sets of single-headed rules over signatures of arity ${\leq}2$, $\fcs$ subsumes $\fus$, which (bad news) does not hold for multi-headed rules, but (relieving news) this could not be any different as there are $\fus$ rule sets for which DaMSOQ entailment is undecidable.
    
\end{itemize}


\noindent\textbf{For space reasons, we defer technical details and most proofs to the appendix.}




%% file: prelims.tex
\section{Preliminaries}\label{sec:preliminaries}

\medskip \noindent {\bf Syntax and formulae.}
We let $\mathfrak{T}$ denote a set of {\em terms}, defined as the union of three countably infinite, 
mutually disjoint sets of \emph{constants}~$\mathfrak{C}$, \emph{nulls}~$\mathfrak{N}$, and \emph{variables}~$\mathfrak{V}$. We use $\const{a}$, $\const{b}$, $\const{c}$, $\ldots$ (occasionally annotated) to denote constants, and use $x$, $y$, $z$, $\ldots$ (occasionally annotated) to denote both nulls and variables.
A \emph{signature}~$\sig$ is a finite set of \emph{predicate symbols} (called \emph{predicates}), which are capitalized ($\epred$, $\rpred$, $\ldots$).
Throughout the paper, we assume a fixed signature $\sig$, unless stated otherwise. For each predicate~\mbox{$\rpred \in \sig$}, we denote its 
\emph{arity} with~$\arity{\rpred}$. \prm{We speak of \emph{binary signatures} whenever $\sig$}\pin{We say $\sig$ is \emph{binary} if it} contains only predicates of arity ${\leq}2$. We assume $\sig$ to contain a special, ``universal'' unary predicate $\top$, assumed to hold of every term.\footnote{Assuming the presence of such a built-in \emph{domain predicate} $\top$ does not affect our results, but it allows for a simpler and more concise presentation.} 
An \emph{atom} over~$\sig$ is an expression  of the form~$\rpred(\vt)$, where~$\vt$ is an~$\arity{\rpred}$-tuple of terms. \prm{A \emph{ground atom} is an atom~$\rpred(\vt)$, where~$\vt$ consists only of constants}\pin{If $\vt$ consists only of constants, then $\rpred(\vt)$ is a \emph{ground atom}}.
An \emph{instance}~$\instance$ over~$\sig$ is a (possibly countably infinite)
set of atoms over constants and nulls, whereas a~$\emph{database}$~$\database$ is
a finite set of ground atoms. The \emph{active domain}~$\adom{\mathcal{X}}$ 
of a 
 set of atoms $\mathcal{X}$ is the set of terms \prm{(i.e., constants and nulls) }appearing in the atoms of $\mathcal{X}$.\prm{Clearly, databases and instances can be seen as model-theoretic structures, so we use $\models$ to indicate satisfaction as usual.} Moreover, as instances over binary signatures can be viewed as directed edge-labelled graphs\prm{ whose edges are typed with predicates}, we often use graph-theoretic terminology when discussing such objects. 

\renewcommand{\smallldots}{...\,}

\medskip \noindent {\bf Homomorphisms\prm{ and substitutions}.}\prm{A \emph{substitution} is a partial function over~$\mathfrak{T}$.}
Given sets $\mathcal{X}$, $\mathcal{Y}$ of atoms, a~\emph{homomorphism} from $\mathcal{X}$ to $\mathcal{Y}$ is a \prm{substitution}\pin{mapping} \prm{$h$ from the terms of~$\mathcal{X}$ to the terms of~$\mathcal{Y}$}\pin{\mbox{$h:\adom{\mathcal{X}}\to\adom{\mathcal{Y}}$}} that satisfies:
(i)~\prm{$\rpred(h(t_1), \smallldots, h(t_n)) \in \mathcal{Y}$, for all~$\rpred(t_1, \smallldots, t_n) \in \mathcal{X}$}\pin{$\rpred(h(\vt)) \in \mathcal{Y}$, for all~$\rpred(\vt) \in \mathcal{X}$}%
, 
and 
(ii)~$h(\const{a}) = \const{a}$, for each~$\const{a} \in \mathfrak{C}$%
. 
$\mathcal{X}$ and $\mathcal{Y}$ are \emph{homomorphically equivalent} (written $\mathcal{X} \equiv \mathcal{Y}$) \iffi homomorphisms exist from $\mathcal{X}$ to $\mathcal{Y}$ and from $\mathcal{Y}$ to $\mathcal{X}$.
A~homomorphism $h$ is an \emph{isomorphism} \iffi it is bijective and $h^{-1}$ is also a homomorphism. 
An instance~$\instance'$ is an \emph{induced sub-instance} of an instance $\instance$ \iffi (i) $\instance' \subseteq \instance$ and (ii) \prm{if $\rpred(t_1,\smallldots,t_n) \in \instance$ and $t_1,\smallldots,t_n \in \adom{\instance'}$, then $\rpred(t_1,\smallldots,t_n) \in \instance'$}\pin{if $\rpred(\vt) \in \instance$ and $\vt \subseteq \adom{\instance'}$, then $\rpred(\vt) \in \instance'$}.
\renewcommand{\smallldots}{\ldots}

\medskip \noindent {\bf Existential rules.}
An \emph{(existential) rule} $\rho$ is a first-order sentence $\forall{\vx\vy} \; \phi(\vx, \vy)
\rightarrow \exists{\vz} \; \psi(\vy, \vz)$,
where~$\vx$, $\vy$, and $\vz$ are mutually disjoint tuples of variables, and both 
the \emph{body}~$\phi(\vx, \vy)$ and the \emph{head}~$\psi(\vy, \vz)$ of~$\rho$  (denoted with~$\body(\rho)$ and~$\head(\rho)$, respectively) are conjunctions (possibly empty, sometimes \prm{denoted }\pin{seen }as sets) of atoms \prm{containing}\pin{over} 
 the indicated variables. The \emph{frontier} $\fr(\rho)$ of $\rho$ is the set of variables $\vy$ shared between the body and the head. \prm{For better readability, we}\pin{We o}ften omit the universal quantifiers prefixing existential rules. 
\prm{We say that a}\pin{A} rule~$\rho$ is \mbox{(i)~\emph{$n$-ary}} \iffi all predicates appearing in~$\rho$ are of arity at most~$n$, \mbox{(ii) \emph{single-headed}} \iffi $\head(\rho)$ contains \prm{at most one atom}\pin{a single atom}, (iii) \emph{datalog} \iffi $\rho$ does not contain an existential quantifier\prm{, and (iv) \emph{non-datalog} \iffi $\rho$ contains an existential quantifier}\pin{ (otherwise \emph{non-datalog})}. We call a finite set of existential rules $\ruleset$ a {\em rule set}. Satisfaction of a rule $\rho$  (a rule set $\ruleset$) by an instance $\instance$ is defined as usual and is written $\instance \models \rho$ ($\instance \models \ruleset$, respectively). Given a database $\database$ and a rule set $\ruleset$, we define the pair $(\db,\ruleset)$ to be a \emph{knowledge base}, and define an instance $\instance$ to be a \emph{model} of \prm{the knowledge base (i.e., of $\database$ and $\ruleset$)}\pin{$(\database, \ruleset)$}, written $\instance\models (\database,\ruleset)$, \iffi $\database \subseteq \instance$ and $\instance \models \ruleset$. A model $\instance$ of $(\database,\ruleset)$ is called \emph{universal} \iffi there is a homomorphism from $\instance$ into every model of $(\database,\ruleset)$. 

\medskip \noindent {\bf Rule application and Skolem chase.} A rule $\rho\,\, = \,\, \phi(\vx,\vy) \rightarrow \exists \vz \psi(\vy,\vz)$ is \emph{appli\-ca\-ble} to an instance $\instance$ \ifandonlyif there is a homomorphism $\hism$ \prm{mapping}\pin{from} $\phi(\vec{x},\vec{y})$ to $\instance$. We then call $(\rho, \hism)$ a \emph{trigger} of $\instance$.
The \emph{application} of a trigger $(\rho, \hism)$ in $\instance$ yields the instance $\fchain{\inst,\rho,\hism} = \instance \cup {\bar\hism}(\psi(\vec{y},\vec{z}))$, where $\bar\hism$ extends $\hism$, mapping each variable $z$ from $\vz$ to a null denoted \smash{$z_{\rho,\hism(\vy)}$.} 
Note that this entails
$\fchain{\fchain{\inst,\rho,\hism},\rho,\hism'} = \fchain{\inst,\rho,\hism}$ whenever $\hism(\vy) = \hism'(\vy)$. 
%
Moreover, applications of different rules or the same rule with different frontier-mappings are independent, so their order is irrelevant. Hence we can define the parallel one-step application of all applicable rules as
$$\ksat{1}{\inst,\ruleset} = 
\hspace{-9ex}\bigcup_{\hspace{9ex}\rho\in \ruleset,\,(\rho, \hism) \text{\,trigger of\,} \instance}\hspace{-9ex} 
\fchain{\inst,\rho,\hism}.$$     
Then, we define the \emph{(breadth-first) Skolem chase sequence} by letting 
$\ksat{0}{\inst,\ruleset} = \inst$ and $\ksat{i+1}{\inst,\ruleset} = \ksat{1}{\ksat{i}{\inst,\ruleset}, \ruleset}$, ultimately obtaining the \emph{Skolem chase} $\sat{\inst,\ruleset} = \bigcup_{i\in \mathbb{N}} \ksat{i}{\inst,\ruleset}$. We note that the Skolem chase of a countable instance is countable, as  is the number of overall rule applications performed to obtain $\sat{\inst,\ruleset}$. 

\medskip \noindent {\bf (Unions of) conjunctive queries and their entailment.}
A {\em conjunctive query} (CQ) is a  formula $\query(\vy) = \exists \vx \, \phi(\vx, \vy)$ with $\phi(\vx, \vy)$ a conjunction (sometimes written as a set) of atoms over the variables from $\vx,\vy$ and constants.
 The variables from $\vy$ are called {\em free}. A \emph{Boolean} CQ (or BCQ) is a CQ with no free variables.
 A \emph{union of conjunctive queries} (UCQ) $\psi(\vy)$ is a disjunction of CQs with free variables $\vy$. 
 We will treat UCQs as sets of CQs. 
\prm{Satisfaction of a BCQ $\query = \exists \vx \, \phi(\vx)$ in an instance $\instance$ is defined as usual and coincides with the existence of a homomorphism from $\phi(\vx)$ to $\instance$.}\pin{A BCQ $\query = \exists \vx \, \phi(\vx)$ is satisfied in an instance $\instance$ if there exists a homomorphism from $\phi(\vx)$ to~$\instance$.} $\instance$ satisfies a union of BCQs if it satisfies at least one of its disjuncts.
An instance $\inst$ and rule set $\ruleset$ \emph{entail} a BCQ $\query = \exists \vx \, \phi(\vx)$, written $(\inst, \ruleset) \models \query$  \iffi $\phi(\vx)$ maps homomorphically into every model of $\inst$ and $\ruleset$. This coincides with  the existence of a homomorphism from $\phi(\vx)$ into any universal model of $\inst$ and $\ruleset$ (e.g., the Skolem chase $\sat{\inst, \ruleset}$).

\medskip \noindent {\bf Rewritings and finite-unification sets.}
Given a rule set $\ruleset$ and a CQ $\query(\vy)$, we say that a UCQ $\psi(\vy)$ is a {\em rewriting} of $\query(\vy)$ under the rule set $\ruleset$ \ifandonlyif for any database $\database$ and any tuple of its elements $\vec{\const{a}}$ the following holds: 
$\sat{\database,\ruleset} \models \query(\vec{\const{a}}) \ \ifandonlyif \ \database \models \psi(\vec{\const{a}}).$
A rule set $\ruleset$ is a {\em finite-unification set} ($\fus$) \iffi for every CQ, there exists a UCQ rewriting \cite{BagLecMugSal11}. This property is also referred to as \emph{first-order rewritability}.
If a rule set $\ruleset$ is $\fus$, then for any given CQ~$\query(\vy)$, we fix one of its rewritings under $\ruleset$ and denote it with $\rewrs{\ruleset}{\query(\vy)}$.

\medskip \noindent {\bf Treewidth and bounded-treewidth sets.}
Let $\inst$ be an instance. A \emph{tree decomposition} of $\inst$ is a (potentially infinite) tree $\tree = (V,E)$, where:
\begin{itemize}

\item  $V \subseteq 2^{\adom{\inst}}$, that is, each node $X \in V$ is a set of terms of $\inst$, and $\bigcup_{X \in V} X = \adom{\inst}$,

\item for each $\rpred(t_{1}, ...\,,t_{n}) \in \inst$, there is an $X \in V$ with $\{t_{1}, ...\,,t_{n}\} \subseteq X$,

\item for each term $t$ in $\inst$, the subgraph of $T$ induced by the nodes $X$ with $t \in X$ is connected.
\end{itemize}

The \emph{width} of a tree decomposition is set to be the maximum over the sizes of all its nodes minus $1$, if such a \prm{finite} maximum exists; otherwise, it is set to $\infty$. Last, we define the \emph{treewidth} of an instance $\inst$ to be the minimal width among all of its tree decompositions, and denote the treewidth of $\inst$ as $\treewidth{\inst}$.
A set of rules $\ruleset$ is a \emph{bounded-treewidth set} ($\bts$) \ifandonlyif for any database $\database$, there is a universal model for $(\database,\ruleset)$ of finite treewidth.\footnote{The term ``\emph{finite}-treewidth set'' would be more fitting and in line with our terminology, but we stick to the established name. Also, the $\bts$ notion is not used entirely consistently in the literature; it sometimes refers to structural properties of a specific type of chase. The ``semantic $\bts$'' notion adopted here subsumes all the others.}

%% file: decidability.tex
\section{A Generic Decidability Argument}

In this section, we provide an abstract framework for establishing de\-ci\-da\-bi\-li\-ty of entailment for a wide range of queries, 
based on 
cer\-tain model-theoretic criteria being met 
by the considered rule set. 
We first recall the classical notion of a (Boolean) datalog query, and after, we specify the class of queries considered in our frame\-work.

\begin{definition}
Given a signature $\sig$, a \emph{(Boolean) datalog query} $\mathfrak{q}$ over $\sig$ is represented by a  
finite set $\ruleset_\mathfrak{q}$ of datalog rules with predicates from $\sig_{\mathrm{EDB}}\uplus\sig_{\mathrm{IDB}}$ where  $\sig_{\mathrm{EDB}}{\,\subseteq\,} \sig$ and $\sig_{\mathrm{IDB}}{\,\cap\,} \sig=\emptyset$ 
such that (i) $\sig_{\mathrm{EDB}}$-atoms do not occur in rule heads of $\ruleset_\mathfrak{q}$%
, and (ii) $\sig_{\mathrm{IDB}}$ contains a distinguished nullary predicate $\mathtt{Goal}$. 
Given an instance $\instance$ and a datalog query $\mathfrak{q}$\prm{ over $\Sigma$}, we say that $\mathfrak{q}$ \emph{holds in} $\instance$ (\prm{or, }$\instance$ \emph{satisfies} $\mathfrak{q}$), written $\instance\models\mathfrak{q}$, \iffi $\mathtt{Goal}\in\sat{\instance,\ruleset_\mathfrak{q}}$. Query entailment is\prm{ then} defined via satisfaction as usual.
\defend
\end{definition}

Datalog queries can be equivalently expressed as sentences in second-order logic (with the $\sig_{\mathrm{IDB}}$ predicates quantified over) or in least fixed-point logic (LFP). For our purposes, the given formulation is the most convenient; e.g., it makes clear that the second-order entailment problem $(\db, \ruleset) \models \mathfrak{q}$ reduces to the first-order entailment problem $(\db,\ruleset\cup\ruleset_\mathfrak{q}) \models \mathtt{Goal}$. 
\begin{definition}
A \emph{datalog/MSO query (DaMSOQ)} over a signature $\sig$
is a pair $(\mathfrak{q},\Xi)$, where~$\mathfrak{q}$ is a datalog query over $\sig$ and $\Xi$ is a monadic second-order (MSO)\footnote{For an introduction to monadic second-order logic, see~\cite[Section 1.3]{CourEngBook}.} sentence equivalent to~$\mathfrak{q}$. 
Satisfaction and entailment of DaMSOQs is defined via any of their constituents. 
\defend
\end{definition}
Consequently, DaMSOQs are the (semantic) intersection of datalog and MSO queries. While representing DaMSOQs as a pair $(\mathfrak{q},\Xi)$ is logically redundant, it is purposeful and necessary: the below decision procedure requires both constituents as input and it is not generally possible to compute one from the other. Arguably, the most comprehensive, known DaMSOQ fragment that is well-investigated and has a syntactic definition is that of \emph{nested monadically defined queries}, a very expressive yet computationally manageable formalism \cite{RudKro13}, subsuming (unions of) Boolean CQs,
 \emph{monadic datalog queries}~\cite{Cosmadakis88}, \emph{conjunctive 2-way regular path queries}~\cite{FlorescuLS98}, and nested versions thereof (e.g. \emph{regular queries}~\cite{ReutterRV17} and others~\cite{BourhisKR14}).


\begin{definition}\label{def:MSOfriendly}
Let $\Sigma$ be a finite signature.
A \emph{width measure} (for $\Sigma$) is a function $\fcn{w}$ mapping every countable instance over $\Sigma$ to a value from $\mathbb{N} \cup \{\infty\}$.
We call $\fcn{w}$ \emph{MSO-friendly} \iffi there exists an algorithm that, taking a number $n \in \mathbb{N}$ and an MSO sentence $\Xi$ as input,
\vspace{-0.5ex}
\begin{itemize}
\item never terminates if $\Xi$ is unsatisfiable, and
\item always terminates if $\Xi$ has a model $\instance$ with
$\fcn{w}(\instance)\leq n$.\defend
\end{itemize} 
\end{definition}

As an unsophisticated example, note that the \emph{expansion} function $\fcn{expansion}: \inst \mapsto |\adom{\inst}|$, mapping every countable instance to the size of its domain, is an MSO-friendly width measure: there are up to isomorphism only finitely many instances with $n$ elements, which can be computed and checked. As a less trivial example, the notion of treewidth has also been reported to fall into this category~\cite{BagLecMugSal11}.

\begin{definition}
Let $\fcn{w}$ be a width measure.
A rule set $\ruleset$ is called a \emph{finite-$\fcn{w}$ set} \iffi for every database $\database$, there exists a universal model $\instance^*$ of $(\db,\ruleset)$  satisfying $\fcn{w}(\instance^*) \in \mathbb{N}$.\defend
\end{definition}

Note that the finite width required by this definition does not need to be uniformly bounded: it may depend on the database and thus grow beyond any finite bound. This is already the case when using the $\fcn{expansion}$ measure from above, giving rise to the class of \emph{finite-expansion sets} ($\fes$), coinciding with the notion of \emph{core-chase terminating} rule sets~\cite{DeuNasRem08}.    

\begin{theorem}\label{thm:decidablequerying}
Let $\fcn{w}$ be an MSO-friendly width measure and let $\ruleset$ be a \emph{finite-$\fcn{w}$ set}.
Then, the entailment problem $(\db, \ruleset) \models (\mathfrak{q},\Xi)$ for any \prm{given }database $\database$ and DaMSOQ $(\mathfrak{q},\Xi)$ is decidable.
\end{theorem}

\begin{proof}
We prove decidability by providing two semi-deci\-sion procedures: one terminating whenever  
$(\db,\ruleset)\models (\mathfrak{q},\Xi)$, the other terminating whenever $(\db,\ruleset) \not\models (\mathfrak{q},\Xi)$. Then, these two procedures, run in parallel, constitute a decision procedure.
\begin{itemize}
    \item Detecting $(\db,\ruleset) \models (\mathfrak{q},\Xi)$. We note that $(\db,\ruleset) \models (\mathfrak{q},\Xi)$ \iffi $(\db,\ruleset) \models \mathfrak{q}$ \iffi $(\db,\ruleset\cup\ruleset_\mathfrak{q})\models \mathtt{Goal}$. The latter is a first-order entailment problem. Thanks to the completeness of first-order logic \cite{God29FOLcomp}, we can recursively enumerate all the consequences of $(\db,\ruleset\cup\ruleset_\mathfrak{q})$ and terminate as soon as we find $\mathtt{Goal}$ among those, witnessing entailment of the query.
    \item Detecting $(\db,\ruleset) \not\models (\mathfrak{q},\Xi)$. 
    %
    %
    In that case, there must exist some model $\instance$ of $(\db,\ruleset)$ with $\instance\not\models (\mathfrak{q},\Xi)$. Note that such ``countermodels'' can be characterized by the MSO formula $\bigwedge\!\db \wedge \bigwedge\! \ruleset \wedge \neg\Xi$.
    Moreover, any universal model $\instance^*$ of $(\db,\ruleset)$ must satisfy $\instance^*\not\models (\mathfrak{q},\Xi)$, which can be shown (by contradiction) as follows: Let $\instance^*$ be a universal model of $(\db,\ruleset)$ and suppose $\instance^*\models (\mathfrak{q},\Xi)$, i.e., $\instance^*\models \mathfrak{q}$. Since satisfaction of datalog queries is preserved under homomorphisms and $\instance^*$ is universal, we know there exists a homomorphism from $\instance^*$ to $\instance$, implying $\instance\models \mathfrak{q}$, and thus $\instance\models (\mathfrak{q},\Xi)$, contradicting our assumption. As $\rs$ is a finite-$\fcn{w}$ set, there exists a universal model $\instance^*$ of $(\db,\ruleset)$ for which $\fcn{w}(\instance^*)$ is finite. Hence, the following procedure will terminate, witnessing non-entailment: enumerate all natural numbers in their natural order and for each number $n$ initiate a parallel thread with the algorithm from \Cref{def:MSOfriendly} with input $n$ and $\bigwedge\!\db \wedge \bigwedge\! \ruleset \wedge \neg\Xi$ (the algorithm is guaranteed to exist due to MSO-friendliness of $\fcn{w}$). Terminate as soon as one thread does.\qedhere
\end{itemize}
\end{proof}

%% file: clique-width.tex
\section{Cliquewidth and its Properties}\label{sec:cliquewidth}

In this section, we will propose a width measure, for which we use the term \emph{cliquewidth}.
Our definition of this measure works for arbitrary countable instances 
and thus properly generalizes Courcelle's earlier eponymous notions for finite directed edge-labelled graphs \cite{CourEngBook} and countable un\-label\-led un\-di\-rec\-ted graphs \cite{Cou04}, as well as
Grohe and Turán's cliquewidth notion for finite instances of arbitrary arity \cite{GroheT04}.

\subsection{Cliquewidth of Countable Instances}\label{sec:cwdef}

Intuitively, the notion of cliquewidth is based on the idea of assembling the considered structure (e.g., instance or graph) from its singleton elements (e.g., terms or nodes). To better distinguish these elements during assembly, each may be assigned an initial color (from some finite set $\cols$). The ``assembly process'' consists of successively applying the following operations to 
previously assembled node-colored structures:
\begin{itemize}
    \item take the disjoint union of two structures ($\oplus$),
    \item uniformly assign the color $k'$ to all hitherto $k$-colored elements ($\fcn{Recolor}_{k\to k'}$),
    \item given a predicate $\rpred$ and a color sequence $\vk$ of length $\arity{\rpred}$, add $\rpred$-atoms for all tuples of appropriately colored elements ($\fcn{Add}_{\rpred,\vk}$).
\end{itemize}
Then, the cliquewidth of a structure is the minimal number of colors needed to assemble it \prm{(or rather, some \prm{arbitrarily }colored version of it)} through successive applications of the above operations. 
For finite structures (e.g., graphs and instances), this is a straightforward, conceivable notion. 
In order to generalize it to the countably infinite case, one has to find a way to describe infinite assembly processes. 
In the finite case, an ``assembly plan'' can be described by an algebraic expression using the above operators, which in turn can be represented by its corresponding ``syntax tree.'' The more elusive idea of an ``infinite assembly plan'' is then implemented by allowing for infinite, ``unfounded'' syntax trees.
 We formalize this idea of ``assembly-plan-encoding syntax trees'' by representing them as countably infinite instances of a very particular shape.

\newcommand{\szero}{\pred{Succ}_0}
\newcommand{\sone}{\pred{Succ}_1}
\newcommand{\ibtree}{\mathcal{T}_\mathrm{bin}}
\newcommand{\decorum}{\mathrm{Dec}(\cols, \mathrm{Cnst},\Sigma)}

\begin{definition}
We define the \emph{infinite binary tree} to be the instance
$$
\ibtree = \big\{\szero(s,s0) \ \big| \ s {\,\in\,} \{0,1\}^*\big\} \cup 
          \big\{ \sone(s,s1) \ \big| \ s {\,\in\,} \{0,1\}^*\big\}
$$
with binary predicates $\szero$ and $\sone$.
That is, the\prm{infinitely many} nulls of $\ibtree$ are denoted by finite sequences of $0$ and $1$. The \emph{root} of $\ibtree$ is the null identified by the empty sequence, denoted $\varepsilon$.

Given a finite set $\cols$ of \emph{colors}, a finite set $\mathrm{Cnst}\subseteq \mathfrak{C}$ of constants, and a finite signature $\Sigma$, the set $\decorum$ of \emph{decorators} consists of the following unary predicate symbols: 

%
\smallskip
\noindent\begin{minipage}[t]{0.57\linewidth}
\begin{itemize}
\item $\const{c}_k$ for any $\const{c} \in\mathrm{Cnst} \cup\{\const{*}\}$ and $k\in\cols$,
\item $\pred{Add}_{\rpred,\vk}$ for any $\rpred\in\Sigma$ and $\vk\in\cols^{\arity{\rpred}}$,
\end{itemize}
\end{minipage}
\begin{minipage}[t]{0.42\linewidth}
\begin{itemize}
\item $\pred{Recolor}_{k\to k'}$ for $k,k'\in \cols$,
\item $\oplus$, and
$\pred{Void}$.
\end{itemize}
\end{minipage}\\[1ex]

A \emph{$(\cols,\mathrm{Cnst},\Sigma)$-decorated infinite binary tree} (or, \emph{decorated tree} for short) is $\ibtree$ extended with facts over $\decorum$ that only use nulls from the original domain of $\ibtree$, i.e., from $\{0,1\}^*$.
A decorated tree $\mathcal{T}$ is called a  \emph{well-decorated tree} \iffi
\begin{itemize}
\item for every null $s \in \{0,1\}^*$, $\mathcal{T}$ contains exactly one fact $\mathrm{Dec}(s)$ with  $\mathrm{Dec} \in \decorum$,
\item for every $\const{c}\in \mathrm{Cnst}$, $\mathcal{T}$ contains at most one fact of the form $\pred{c}_k(s)$,
\item if $\pred{Add}_{\rpred,\vk}(s) \in \mathcal{T}$ or $\pred{Recolor}_{k\to k'}(s) \in \mathcal{T}$,
       then $\pred{Void}(s0)\not\in \mathcal{T}$ and $\pred{Void}(s1) \in \mathcal{T}$,
\item if $\oplus(s) \in \mathcal{T}$, then $\pred{Void}(s0),\pred{Void}(s1) \not\in \mathcal{T}$,
\item if $\pred{Void}(s) \in \mathcal{T}$ or $\pred{c}_k(s) \in \mathcal{T}$, then  $\pred{Void}(s0),\pred{Void}(s1) \in \mathcal{T}$.\defend

\end{itemize}
\end{definition}

Recall that, due to Rabin's famous Tree Theorem \cite{Rabin69}, 
the validity of a given MSO sentence $\Xi$ in $\ibtree$ is decidable.
Also, it should be obvious that, given a decorated tree, checking well-decoratedness can be done in first-order logic. More precisely, fixing $\cols$, $\Sigma$, and $\mathrm{Cnst}$,  there is a first-order sentence $\Phi_\mathrm{well}$ such that for any $(\cols,\mathrm{Cnst},\Sigma)$-decorated tree $\mathcal{T}$, $\mathcal{T}$ is well-decorated \ifandonlyif $\mathcal{T} \models \Phi_\mathrm{well}$. 

\begin{definition}\label{def:entities-tree-to-instance}
Let $\mathcal{T}$ be a $(\cols,\mathrm{Cnst},\Sigma)$-well-decorated tree.
We define the function $\fcn{ent}^\mathcal{T}\!:\{0,1\}^* \to 2^{\mathrm{Cnst} \cup \{0,1\}^*}$ mapping each null $s \in \{0,1\}^*$ to its \emph{entities} (a set of nulls and constants) as follows:
$$\fcn{ent}^\mathcal{T}\!(s) = \big\{ ss' \mid \pred{*}_k(ss') \in \mathcal{T}, s' \in \{0,1\}^*\big\} \cup \big\{ \const{c} \mid \pred{c}_k(ss') \in \mathcal{T}, \const{c}\in \mathrm{Cnst}, s' \in \{0,1\}^*\big\}.$$
Every tree node $s$ also endows each of its entities with a \emph{color} from $\cols$ through the function $\fcn{col}^\mathcal{T}_s: \fcn{ent}^\mathcal{T}\!(s) \to \cols$ in the following way:
$$ \fcn{col}^\mathcal{T}_s\!(e) = \begin{cases}
k &\!\! \text{if } e=\const{c}\in \mathrm{Cnst} \text{ and } \pred{c}_k(s)\in \mathcal{T}, \text{or if } e=s\in \{0,1\}^* \text{ and } \pred{*}_k(s)\in \mathcal{T},\\
k' &\!\! \text{if } \pred{Recolor}_{k\to k'}(s)\in \mathcal{T} \text{ and } \fcn{col}^\mathcal{T}_{s0}(e) = k,\\
\fcn{col}^\mathcal{T}_{s0}(e) &\!\! \text{if } \pred{Recolor}_{k\to k'}(s)\in \mathcal{T} \text{ and } \fcn{col}^\mathcal{T}_{s0}(e) \neq k, \text{or if } \pred{Add}_{\rpred,\vk}(s)\in \mathcal{T},\\
\fcn{col}^\mathcal{T}_{sb}(e) &\!\! \text{if } {\oplus}(s) \in \mathcal{T}, e \in \fcn{ent}^\mathcal{T}\!(sb), \text{ and } b \in \{0,1\}.\\
\end{cases}$$
\prm{Furthermore, e}\pin{E}very node $s$ is assigned a set of $\Sigma$-atoms over its entities as indicated by the sets $\mathrm{Atoms}_s$\pin{:}\prm{ defined by}
$$ \mathrm{Atoms}_s = 
\begin{cases}
\{\top(\const{c})\} & \!\!\text{if } \pred{c}_k(s)\in \mathcal{T}\!, \const{c} \in \mathrm{Cnst}, \\
\{\top(s)\} & \!\!\text{if } {\pred{*}_k(s)}\in \mathcal{T}, \\
\{\rpred(\ve) \mid \fcn{col}^\mathcal{T}_{s}\!(\ve)=\vk\} & \!\!\text{if } \pred{Add}_{\rpred,\vk}(s)\in \mathcal{T},\\
\emptyset & \!\!\text{otherwise.}\\
\end{cases}$$
Defining a \emph{colored instance} as a pair $(\instance,\coloring)$ of an instance $\instance$ and a function $\coloring$ mapping elements of $\adom{\instance}$ to colors in a set $\cols$, we now associate each node $s$ in $\mathcal{T}$ with the 
colored $\Sigma$-instance $(\instance^\mathcal{T}_s,\coloring^\mathcal{T}_s)$, 
with
$\instance^\mathcal{T}_s = {\bigcup}_{s'\in \{0,1\}^*}\mathrm{Atoms}_{ss'}$
and
$
\coloring^\mathcal{T}_s = \fcn{col}^\mathcal{T}_{s}.
$
Finally, we define the \emph{colored instance $(\instance^\mathcal{T}\!,\coloring^\mathcal{T})$ represented by $\mathcal{T}$} as $(\instance^\mathcal{T}_\varepsilon\!,\coloring^\mathcal{T}_\varepsilon)$ where $\varepsilon$ is the root of $\mathcal{T}$.
\end{definition}

\begin{definition}
Given a colored instance $(\instance,\coloring)$ over a finite set $\mathrm{Cnst}$ of constants and a countable set of nulls as well as a finite signature $\Sigma$,
the \emph{cliquewidth} of $(\instance,\coloring)$, written $\cw{\instance,\coloring}$, is defined to be the smallest natural number $n$ such that $(\instance,\coloring)$ is isomorphic to a colored instance represented by some $(\cols,\mathrm{Cnst},\Sigma)$-well-decorated tree with $|\cols|=n$.
If no such number exists, we let \mbox{$\cw{\instance,\coloring}= \infty$}. The cliquewidth $\cw{\instance}$ of an instance $\instance$ is defined to be the minimum cliquewidth over all of its colored versions.\defend
\end{definition}

\begin{example}\label{ex:cw-transitivity}
The instance \mbox{$\inst_{<} =  \{\rpred(n,m) \mid n,m \in \mathbf{N}, n < m\}$} has a cliquewidth of $2$, witnessed by the well-decorated tree corresponding to the (non-well-founded) expression $E$ impli\-cit\-ly defined by $E = \fcn{Add}_{\rpred,1,2}({\const{*}_1} \oplus \fcn{Recolor}_{1\to 2}(E))$. 
\end{example}
%
%
%
We will later use the following operation on decorated trees.
\begin{definition}
Let $(\inst,\coloring)$ be a colored instance and $\mathcal{T}$ be a well-decorated tree witnessing that $\cw{\inst, \coloring} \leq n$. Then for any $\rpred \in \Sigma$
we let $\fcn{Add}_{\rpred,\vk}(\inst, \coloring)$ denote the instance $\inst^{\mathcal{T'}}$ represented by the well-decorated tree $\mathcal{T'}$ defined as follows:
\vspace{-0.3ex}
\begin{itemize}
    \item The root $\varepsilon$ of $\mathcal{T'}$ is decorated by $\addd{\rpred,\vk}$,
    \item the left sub-tree of $\varepsilon$ is isomorphic to $\mathcal{T}$,
    \item the right sub-tree of $\varepsilon$ is wholly decorated with $\pred{Void}$.\defend
    \end{itemize}
\end{definition}


\subsection{Finite-Cliquewidth Sets and Decidability}\label{sec:fcw-decidability}

We now identify a new class of rule sets for which DaMSOQ query entailment is decidable, that is, the class of \emph{finite-cliquewidth sets}. 

\begin{theorem}\label{thm:cliquewidth-k-MSO-decidable}
For a fixed $n\in \mathbb{N}$, determining if a given MSO formula~$\Xi$ 
has a model $\inst$ with $\cw{\inst} \leq n$ is decidable. Thus, cliquewidth is MSO-friendly.
\end{theorem}
\begin{proof}\!(Sketch)
We use the classical idea of MSO interpretations.
Picking $\cols=\{1,\ldots,n\}$, one shows that for every given MSO sentence $\Xi$, one can compute an MSO sentence $\Xi'$, 
such that for every $(\cols,\mathrm{Cnst},\Sigma)$-well-deco\-rated tree $\mathcal{T}$, $\instance^\mathcal{T} \models \Xi$ \iffi $\mathcal{T} \models \Xi'$. 
Thus, checking if $\Xi$ holds in some $\mathrm{Cnst},\Sigma$-instance of cliquewidth ${\leq}{n}$ can be done by checking if $\Xi'$ holds in some $(\cols,\mathrm{Cnst},\Sigma)$-well-decorated tree, which in turn is equivalent to the existence of a $(\cols,\mathrm{Cnst},\Sigma)$-decorated tree that is a model of $\Xi'\wedge \Phi_\mathrm{well}$. 
Obtain~$\Xi''$ from $\Xi'\wedge \Phi_\mathrm{well}$ by re\-in\-ter\-pre\-ting all unary predicates as MSO set variables that are quantified over existentially. $\Xi''$ is an MSO formula over the signature $\{\szero,\sone\}$ which is valid in $\ibtree$ \iffi some decoration exists that makes $\Xi'\wedge \Phi_\mathrm{well}$ true. Thus, we have reduced our problem to checking the validity of a MSO sentence in $\ibtree$, which is decidable by Rabin's Tree Theorem \cite{Rabin69}. 
\end{proof}

With this insight in place and the appropriate rule set notion defined, we can leverage \Cref{thm:decidablequerying} for our decidability result. 

\begin{definition}
A rule set $\ruleset$ 
is called a \emph{finite-cliquewidth set} ($\fcs$) \ifandonlyif for any database $\database$, there exists a universal model for $(\database,\ruleset)$ of finite cliquewidth.
\defend
\end{definition}

\begin{corollary}\label{cor:fcsdecidable}
For every $\fcs$ rule set $\ruleset$, 
the query entailment problem $(\database,\ruleset) \models (\mathfrak{q},\Xi)$ for databases $\database$, and DaMSOQs $(\mathfrak{q},\Xi)$ is decidable.
\end{corollary}

In view of \cref{ex:cw-transitivity}, it is not hard to verify that the rule set $\ruleset^\infty_\mathrm{tran}$ from the introduction is $\fcs$.
Yet, it is neither $\bts$ (as argued before), nor $\fus$, which can be observed from the fact that it does not admit a finite rewriting of the BCQ $\epred(\const{a},\const{b})$.
Notably, it also does not exhibit \emph{finite controllability} ($\mathbf{fc}$), another generic property that guarantees decidability of query entailment~\cite{Ros06}. 
A rule set $\ruleset$ is $\mathbf{fc}$ \iffi for every $\db$ and CQ $q$ with $(\db,\ruleset) \not\models q$ there exists a \emph{finite} \mbox{(possibly non-}\mbox{universal) model} $\inst \models (\db,\ruleset)$ with $\inst \not\models q$.  
Picking $\db=\{\epred(\const{a},\const{b})\}$ and $q=\exists x \epred(x,x)$ reveals that $\ruleset^\infty_\mathrm{tran}$ is not $\fc$. Therefore, $\fcs$ encompasses rule sets not captured by any of the popular general decidability classes (namely, $\bts$, $\fus$, and $\fc$). On another note, is no surprise that, akin to $\bts$, $\fus$, and $\fc$, the membership of a rule set in $\fcs$ is undecidable, which can be argued exactly in the same way as for $\bts$ and $\fus$ \cite{BagLecMugSal11}.

\subsection{Cliquewidth and~Treewidth}\label{sec:cw-vs-tw}

We now show that for binary signatures, the class of instances with finite cliquewidth subsumes the class of instances with finite tree\-width,\prm{implying that $\fcs$ generalizes $\bts$}
\pin{implying $\bts \subseteq \fcs$.} 

\begin{theorem}\label{thm:tree-width-implies-clique-width}
Let $\inst$ be a countable instance over a binary signature. If $\inst$ has finite treewidth, then $\inst$ has finite cliquewidth.
\end{theorem}
\begin{proof}\!(Sketch) We convert a tree decomposition $T$ of $\inst$ with a width $n$ into a well-decorated tree: (1) By copying nodes, transform $T$ into an infinite binary tree $T'$.  (2) For each term $t$ from $\inst$, let the \emph{pivotal} node of $t$ in $T'$ be the node closest to the root containing $t$. For any two terms $t$ and $t'$ from $\inst$ co-occurring in an atom, their pivotal nodes are in an ancestor relationship. (3) Assign one of $n+1$ ``slots'' to every term \prm{such}\pin{so} that \prm{inside}\pin{in} each node of $T'$, every element has a distinct slot. (4) Extract a well-decorated tree from $T'$ by transforming every node into the following bottom-up sequence of operations: (i) $\oplus$-assemble the input from below, (ii) introduce every element for which the current node is pivotal with a color that encodes ``open link requests'' to elements (identified by their slots) further up, (iii) satisfy the color-link requests \prm{transmitted }from below via $\pred{Add}$, (iv) remove the satisfied requests by $\pred{Recolor}$.
\end{proof}
We note that the converse of \Cref{thm:tree-width-implies-clique-width} does not hold: Despite its finite cliquewidth, the treewidth of instance $\instance_<$ from \Cref{ex:cw-transitivity}
is infinite, as its $\rpred$-edges form an infinite clique.

To the informed reader, \Cref{thm:tree-width-implies-clique-width} might not come as a surprise, given that this relationship is known to hold for 
countable unlabelled undirected graphs\footnote{This follows as a direct consequence of a compactness property relating the cliquewidth of countable undirected graphs to that of their finite induced subgraphs.}~\cite{Cou04}.
It does, however, cease to hold for infinite structures with predicates of higher arity. 

\begin{example}\label{ex:badnewsternary}
Let $\rpred$ be a ternary predicate. 
 The instance $\inst_\mathrm{tern} = \{ \rpred(-1,n,n{+}1) \mid n\in \mathbb{N} \}$ has a treewidth of $2$, however, it does not have finite cliquewidth (see \cref{app:example-15-justification}). Concomitantly, the rule set $\ruleset_\mathrm{tern} = \{ \rpred(v,x,y) \to \exists z \rpred(v,y,z) \}$ is $\bts$, but not $\fcs$.  
\end{example}

\noindent
While this result may be somewhat discouraging, one can show that its effects can be greatly mitigated by the technique of reification.

\begin{definition}
Given a finite signature $\Sigma = \Sigma_{\leq2} \uplus \Sigma_{\geq3}$ divided into at-most-binary and higher-arity predicates, we define the \emph{reified version} of $\Sigma$ as the binary signature $\Sigma^{\smash{\mathrm{rf}}} = \Sigma_{\leq2} \uplus \Sigma^{\smash{\mathrm{rf}}}_{\smash{2}}$ with
$\Sigma^{\smash{\mathrm{rf}}}_2= \{\rpred_i \mid \rpred\in \Sigma_{\geq3}, 1\leq i \leq \arity{\rpred}\}$ a fresh set of binary predicates.
The~function $\fcn{reify}$ maps atoms over $\Sigma_{\leq2}$ to themselves, while any higher-arity atom $\alpha=\rpred(t_1,...\,,t_k)$ with $k\geq 3$ is mapped~to the
set $\{\rpred_i(u_\alpha,t_i) \mid 1\leq i \leq k\}$,
where $u_\alpha$ is a fresh null or variable. We lift $\fcn{reify}$ to instances, rules, and queries in the natural way. 
\defend
\end{definition}

It is best to think of the ``reification term'' $u_\alpha$ as a locally existentially quantified variable. In particular, in rule heads, $u_\alpha$ will be existentially quantified. Moreover, in datalog queries, $\fcn{reify}$ is only applied to $\Sigma_\mathrm{EDB}$-atoms, while $\Sigma_\mathrm{IDB}$-atoms are left unaltered; this ensures that the result is again a datalog query.

\begin{example}
Consider the instance $\inst_\mathrm{tern}$ from \Cref{ex:badnewsternary}. We observe that
$\fcn{reify}(\inst_\mathrm{tern}) =\\ \{ \rpred_1(u_n,-1), \rpred_2(u_n,n), \rpred_3(u_n,n{+}1) \mid n\in \mathbb{N} \}$ has a treewidth of $3$ and a cliquewidth of $6$, witnessed by 
the (non-well-founded) expression $\fcn{Add}_{\rpred_1,5,1}(\const{*}_1 \oplus E)$, where 
$E$ is implicitly defined via
$ E = \fcn{Recolor}_{2\to 3} (\fcn{Recolor}_{4\to 5} (\fcn{Recolor}_{3\to 6}(\fcn{Add}_{\rpred_3,4,3}(\fcn{Add}_{\rpred_2,4,2}(\const{*}_2 \oplus (\const{*}_4 \oplus E))))))$.
\end{example}
Let us list in all brevity some pleasant and fairly straightforward properties of reification:
\begin{itemize}[leftmargin=5ex]
\item[(i)] 
If $\tw{\instance}$ is finite, then so are  
 $\tw{\fcn{reify}(\instance)}$ and $\cw{\fcn{reify}(\instance)}$.
\item[(ii)] 
$\sat{\fcn{reify}(\instance),\fcn{reify}(\ruleset)} \equiv \fcn{reify}(\sat{\instance,\ruleset})$.
\item[(iii)] 
If $\ruleset$ is $\bts$, then $\fcn{reify}(\ruleset)$ is $\bts$ and $\fcs$.
\item[(iv)] 
$(\database,\ruleset) \models (\mathfrak{q},\Xi)$ \iffi $(\fcn{reify}(\database), \fcn{reify}(\ruleset)) \!\models \!(\fcn{reify}(\mathfrak{q}),\fcn{reify}(\Xi))$.
\end{itemize}

\begin{example}
Revisiting \Cref{ex:badnewsternary}, we can confirm that $\fcn{reify}(\ruleset_\mathrm{tern})$ comprising the rule
$\rpred_1(u,v) \wedge \rpred_2(u,x) \wedge \rpred_3(u,y) \to \exists u'z \big( \rpred_1(u',v) \wedge \rpred_2(u',y) \wedge \rpred_3(u',z)\big)$ is $\bts$ and $\fcs$.
\end{example}

The above insights regarding reification allow us to effortlessly reduce any query entailment problem over an arbitrary $\bts$ rule set to a reasoning problem
over a binary $\fcs$ one. Also, reification is a highly local transformation; it can be performed independently and atom-by-atom on $\database$, $\ruleset$, and $(\mathfrak{q},\Xi)$. Therefore, restricting ourselves to $\fcs$ rule sets 
does not deprive us of the expressiveness, versatility, and reasoning capabilities offered by arbitrary $\bts$ rule sets over arbitrary signatures, which includes the numerous classes based on guardedness: guarded and frontier-guarded rules as well as their respective variants weakly, jointly, and glut-(frontier-)guarded rules \cite{BagLecMugSal11,CaliGK13,KrotzschR11}. 
It is noteworthy that our line of argument also gives rise to an independent proof of decidability of query entailment for $\bts$:\footnote{Note that this result establishes entailment for arbitrary DaMSOQs, while previously reported results \cite{BagLecMugSal11,CaliGK13} only covered CQs. But even when restricting the attention to CQ entailment, the proofs given in these (mutually inspired) prior works do not appear entirely conclusive to us: both invoke a result by Courcelle \cite{Cou90}, which, as stated in the title of the article and confirmed by closer inspection, only deals with classes of \emph{finite} structures/graphs with a uniform treewidth bound. Hence, the case of infinite structures is not covered, despite being the prevalent one for universal models. Personal communication with Courcelle confirmed that the case of arbitrary countable structures -- although generally believed to hold true -- is not an immediate consequence of his result.}

\begin{theorem}
For every $\bts$ rule set $\ruleset$, the query entailment problem $(\database,\ruleset) \models (\mathfrak{q},\Xi)$ for databases $\database$, and DaMSOQs $(\mathfrak{q},\Xi)$ is decidable.
\end{theorem}

%
%

\begin{proof}
As argued\prm{ above}, \prm{the problem }$(\database,\ruleset) \models (\mathfrak{q},\Xi)$ \prm{can be reduced }\pin{reduces }to $(\fcn{reify}(\database),\fcn{reify}(\ruleset)) \models (\fcn{reify}(\mathfrak{q}),\fcn{reify}(\Xi))$. As $\ruleset$ is $\bts$, so is 
$\fcn{reify}(\ruleset)$. The latter being binary, we conclude that it is \prm{also }$\fcs$. Then, the claim follows from decidability of DaMSOQ entailment for $\fcs$ rule sets (\Cref{cor:fcsdecidable}). 
\end{proof}




%% file: cw-binary.tex
\section{Comparing FCS and FUS}\label{sect:cw-binary}

We have seen that the well-known $\bts$ class is subsumed by the $\fcs$ class (directly for arities ${\leq}2$, and via reification otherwise). As discussed in the introduction, another prominent class of rule sets (incomparable to $\bts$) with decidable CQ entailment is the $\fus$ class. We dedicate the remainder of the paper to mapping out the relationship between $\fcs$ and $\fus$, obtaining the following two results (established in \cref{sec:binary} and \cref{sec:higher-arity}, respectively):

\begin{theorem}\label{thm:binary}
Any \prm{finite-unification set }\pin{$\fus$ rule set }of single-headed rules over a binary signature is \prm{a finite-cliquewidth set}\pin{$\fcs$}.
\end{theorem}


\begin{theorem}\label{thm:higher-arity}
There exists a \prm{finite-unification set }\pin{$\fus$ rule set }of multi-headed rules over a binary signature that is not \prm{a finite-cliquewidth set}\pin{$\fcs$}.
\end{theorem}
%
%
As a consequence of these findings, the necessary restriction to single-headed rules prevents us from wielding the powers of reification in this setting.

%% file: binary-case.tex
\subsection{The Case of Single-Headed Rules}\label{sec:binary}

In this section, we establish \cref{thm:binary}. To this end, let a binary signature $\sig$, a finite unification set $\ruleset$ of single-headed rules over $\sig$, and a database $\database$ over $\sig$ be arbitrary but fixed for the remainder of the section. \prm{Hence, we will often conveniently}\pin{We will} abbreviate $\sat{\database,\ruleset}$ by $\sa$.
Noting that $\sa$ is a universal model, \cref{thm:binary} is an immediate consequence of the following lemma, which we are going to establish in this section.

\begin{lemma}\label{lem:binary-main}
$\sa$ has finite cliquewidth.
\end{lemma}

\msubsection{Looking past datalog}
Let $\se \subseteq \sa$ be the instance containing the \emph{existential atoms} of $\sa$, that is, the atoms derived via the non-datalog rules of $\ruleset$. We show that $\se$ forms a \emph{typed polyforest}, meaning that $\se$ can be viewed as a directed graph where
(i) edges are typed with binary predicates from $\sig$ and (ii) when disregarding the orientation of the edges, the graph forms a forest.
This implies that $\se$ has treewidth $1$. A tree decomposition of $\se$ can be extended into a finite-width tree decomposition of $\db \cup \se$ by adding  the finite set $\adom{\db}$ to every node. 
In the sequel, we use $\db \cup \se$ as a basis, to which, in a very controlled manner, we then add the ``missing'' non-existential atoms derived via datalog rules. Thereby, $\ruleset$ being $\fus$ will be of great help.

\msubsection{Rewriting datalog rules}
We transform the datalog subset of $\ruleset$ into a new rule set $\rulesetnorm$\!, which we then use to fix a set of colors and establish the finiteness of $\cw{\sa}$.
Letting $\rulesetdl$ denote the set of all datalog rules from $\ruleset$, we note the following useful equation:
$\sa = \sat{\db \cup \se,\rulesetdl}$ ($\dag$).
For any $\rpred \in \sig$, we let $\rew{\rpred}$ denote the datalog rule set 
$\set{\varphi(\vx, \vy) \to \rpred(\vy) \mid \exists{\vx} \varphi(\vx, \vy) \in \rewrs{\ruleset}{\rpred(\vy)}}$
giving rise to the overall datalog rule set
$
\rulesetnorm = {\bigcup}_{\rpred \in \sig} \rew{\rpred}.
$
We now show that the rule set $\rulesetnorm$ admits an important property:

\begin{lemma}\label{lem:trigger-for-each-rewritten-rule}
For any $\rpred(\vt) \in \sa \setminus (\db \cup \se)$, there exists a trigger $(\rho, h)$ in $\db \cup\se$ with $\rho \in \rulesetnorm$ such that applying $(\rho, h)$ adds $\rpred(\vt)$ to $\db \cup\se$.
\end{lemma}
\prm{\begin{proof}
Given that $\ruleset$ is $\fus$, it follows that for any $\arity{\rpred}$-tuple $\vt$ of terms from $\db \cup \se$, there exists a CQ $\exists \vx \thephi(\vx,\vy) \in \rewrs{\ruleset}{\rpred(\vy)}$ such that the following holds:
$$\sat{\db \cup \se,\ruleset} \models \rpred(\vt) \ \ \ \ifandonlyif \ \ \ \db \cup \se \models \exists \vx \thephi(\vx,\vt).
$$
Thus, said trigger exists for some $\rho \in \rew{\rpred}$ in $\db \cup \se$.
\end{proof}}
\pin{\begin{proof}
Given that $\ruleset$ is $\fus$, it follows that for any $\arity{\rpred}$-tuple $\vt$ of terms from $\db \cup \se$, there exists a CQ $\exists \vx \thephi(\vx,\vy) \in \rewrs{\ruleset}{\rpred(\vy)}$ such that the following holds:
$\sat{\db \cup \se,\ruleset} \models \rpred(\vt) \ \ifandonlyif \ \db \cup \se \models \exists \vx \thephi(\vx,\vt).
$
Thus, said trigger exists for some $\rho \in \rew{\rpred}$ in $\db \cup \se$.
\end{proof}}
From \prm{the above lemma }\pin{\cref{lem:trigger-for-each-rewritten-rule} }and ($\dag$), we \prm{can }conclude $\ksat{1}{\db \cup \se,\rulesetnorm} \;=\; \sa$ ($\ddag$).
This \prm{equation }tells us that we can apply \prm{all }rules of $\rulesetnorm$  \prm{within a single }\pin{in one }step to obtain \prm{the same result as }$\sa$. 
Ultimately, we will leverage this \prm{observation }to bound the cliquewidth of $\sa$.
%
In view of ($\ddag$), we may reformulate \cref{lem:binary-main} as follows:

\begin{lemma}\label{lem:binary-reformulated}
$\ksat{1}{\db \cup \se, \rulesetnorm}$ has finite cliquewidth.
\end{lemma}


\msubsection{Separating connected from disconnected rules}
\prm{In the next step}\pin{Next}, we distinguish between two types of rules in $\rulesetnorm$\!.
We \prm{define}\pin{say} a datalog rule \prm{to be}\pin{is} {\em disconnected} \iffi \prm{variables in its head}\pin{its head variables} belong to distinct connected components in its body; otherwise it is called {\em connected}. 
We let $\rulesetnormcon$ denote all connected rules of $\rulesetnorm$ and let $\rulesetnormminus$ denote all disconnected rules\prm{, that is, $\rulesetnorm = \rulesetnormcon \;\uplus\; \rulesetnormminus$}.
One can observe that upon applying connected rules, frontier variables can only be mapped to -- and hence connect -- ``nearby terms'' of $\db \cup \se$ (using a path-based distance, bounded by the size of rule bodies), which, together with our insights about the structure of $\db \cup \se$, permits the construction of a tree decomposition of finite width, giving rise to the following lemma:

\begin{lemma}\label{cor:skeleton-plus-has-bounded-cw} 
$\ksat{1}{\db \cup \se,\rulesetnormcon}$ has finite treewidth.
\end{lemma}

Note that this lemma does not generalize to all of $\rulesetnorm$, since disconnected rules from $\rulesetnormminus$ might realize ``concept products'' as in $\pred{A}(x) \wedge \pred{A}(y) \to \pred{E}(x,y)$. Clearly, such rules from $\rulesetnormminus$ are the reason why $\bts$ fails to subsume $\fus$. Let us define 
$\sep = \ksat{1}{\db \cup \se,\rulesetnormcon}.$
As $\ksat{1}{\se,\rulesetnorm} = \ksat{1}{\sep,\rulesetnormminus}$, all that is left to show is:

\begin{lemma}\label{lem:binary-reformulated-plus}
$\ksat{1}{\sep, \rulesetnormminus}$ has finite cliquewidth.
\end{lemma}

\msubsection{Searching for a suitable coloring}
Having established the finite treewidth of $\sep$ by \cref{cor:skeleton-plus-has-bounded-cw}, its finite cliquewidth can be inferred by \cref{thm:tree-width-implies-clique-width}, implying the existence of a well-decorated tree $\mathcal{T}$ with $\inst^{\mathcal{T}} = \sep$. Moreover, we know that $\sep$ already contains all terms of $\sa$. Thus, all that remains is to ``add'' the missing datalog atoms of $\sa \setminus \sep$ to $\inst^\mathcal{T}$. Certainly, the coloring function $\fcn{col}^\mathcal{T}_\varepsilon$ provided by $\mathcal{T}$ cannot be expected to be very helpful in this task. To work around this, the following \pin{technical} lemma ensures that an arbitrary coloring can be installed on top of a given instance of finite cliquewidth. 

\begin{lemma}[Recoloring Lemma]
\label{cor:recoloring}
Let $\inst$ be an 
instance 
satisfying $\cw{\inst} = n$ and let $\coloring' : \adom{\inst} \to \cols'$ be an arbitrary coloring of~$\inst$. Then $\cw{\inst,\coloring'} \leq (n+1) \cdot |\cols'|.$
\end{lemma}
\begin{proof}(Sketch) Let $\cw{\inst} = n$ be witnessed by a well-decora\-ted tree $\mathcal{T}$. Let $\cols$ with $|\cols|=n$ be the set of colors used in $\mathcal{T}$. We construct a well-decorated tree $\mathcal{T}^{\coloring'}$ representing $(\inst,\coloring')$ and using the color set $(\cols \times \cols') \;\uplus\; \cols' $, thus
witnessing $\cw{\inst,\coloring'} \leq (n+1) \cdot |\cols'|$. 
We obtain $\mathcal{T}^{\coloring'}$ from $\mathcal{T}$ by modifying each node $s$ of $\mathcal{T}$ as follows:
\begin{itemize}
    \item If $s$ is labeled with $\pred{*}_k$, we change it to $\pred{*}_{(k, \coloring'(s))}$. 
    \item If $s$ is labeled with $\pred{c}_k$ with $\const{c} \in \mathrm{Cnst}$, we change it to $\pred{c}_{(k, \coloring'(\const{c}))}$. 
    \item If $s$ is labeled with $\pred{Add}_{\rpred,k}$, we replace $s$ with a sequence of nodes, one for  each decorator in $\set{\pred{Add}_{\rpred,(k,\ell)} \mid \ell \in \cols'}$. We proceed in an analogous fashion for decorators $\pred{Add}_{\rpred,k,k'}$ and $\pred{Recolor}_{k \to k'}$.
    \item If $v$ is labeled with other decorators, we keep it as is.
\end{itemize}
Last, on top of the obtained tree, we apply a ``color projection'' to $\cols'$ by adding recoloring statements of the form $\pred{Recolor}_{(k,\ell) \to \ell}$ for all $(k,\ell) \in \cols \times \cols'$.
To complete the construction, we add the missing nodes to $\mathcal{T}^{\coloring'}$, each decorated with $\pred{Void}$.
\end{proof}


We proceed by defining types for elements in $\sep$, giving rise to the desired coloring.
For the following, note that by definition, any rule from $\rulesetnormminus$ must have two frontier variables.

\begin{definition}
Let $\rho \in \rulesetnormminus$ with $\body(\rho) = \thephi(x_1, x_2,\vy)$, with $x_1$, $x_2$ frontier variables. 
 Let us define
$
\rho_1(x) =
\exists{\vy x_2}\; \thephi(x, x_2, \vy)
\ \text{and}\ 
\rho_2(x) =
\exists{\vy x_1}\; \thephi(x_1, x, \vy).
$
Then, for a term $t \in \adom{\sep}$, we define its {\em type} $\tau(t)$ as $\big\{\rho_i(x) \mid \rho \in \rulesetnormminus ,\ \sep \models \rho_i(t), i \in \set{1,2}\big\}.$\defend
\end{definition}

\noindent
\prm{Enough groundwork has been laid }\pin{Now, we }define our new coloring function $\coloringe$. Given a term $t$ of $\sep$, we let
$
\coloringe(t) = \tau(t).
$


\begin{corollary}
There exists an $n_\exists \in \mathbb{N}$ with $\cw{\sep, \coloringe} = n_\exists$.
\end{corollary}
\begin{proof}
From \cref{cor:skeleton-plus-has-bounded-cw} and \cref{thm:tree-width-implies-clique-width}, we know that $\cw{\sep}$ is finite. Hence, there exists a natural number $n$ and a coloring $\coloring$ such that $\cw{\sep, \coloring} = n$. Moreover, the codomain of $\coloringe$ is finite, say $n'$. Thus, we get $n_\exists = (n+1)\cdot n'$ by \cref{cor:recoloring}.
\end{proof}

\msubsection{Coping with disconnected rules}
We now conclude the proof of \cref{lem:binary-reformulated-plus}, i.e., we show that $\cw{\so, \coloringe} \leq n_\exists$, thereby proving \cref{thm:binary}. To this end, consider some rule $\rho \in \rulesetnormminus$. As stated earlier, we can assume that $\rho$ is of the~form
%
$
\thephi(x_1, x_2, \vy) \to \rpred(x_1, x_2).
$
From this, we obtain the following useful correspondence:

\begin{lemma}\label{lem:col_query_correspondency}
For any $t,t' \in \adom{\sep}$ and $\rho \in \rulesetnormminus$,
$$\sep \models \exists{\vy}\thephi(t,t',\vy) \text{ \ \ \ifandonlyif \ \ } \rho_1(x) \in \tau(t) \text{ and } \rho_2(x) \in \tau(t').$$
\end{lemma}
\begin{proof}
\prm{The forward direction} \pin{($\Rightarrow$)} follows from the definition of each set. For \prm{the backward direction} \pin{$(\Leftarrow)$}, we exploit disconnectedness and split $\thephi(x_1, x_2, \vy)$ into two distinct parts. Let $\thephi_1(x_1, \vy_1)$ be the connected component of $\thephi(x_1, x_2, \vy)$ that contains $x_1$, and $\thephi_2(x_2, \vy_2)$ denote the remainder (which includes the connected component of $x_2$). By assumption, $\vy_1$ and $\vy_2$ are disjoint, whence $\exists{\vy}\thephi(t,t',\vy)$ is equivalent to $\exists{\vy_1} \thephi_1(t, \vy_1) \wedge \exists{\vy_2} \thephi_2(t', \vy_2)$. Note that $\rho_1(x) \in \tau(t)$ implies $\sep \models \exists{\vy_1} \thephi_1(t, \vy_1)$, while $\rho_2(x) \in \tau(t')$ implies $\sep \models \exists{\vy_2} \thephi_2(t', \vy_2)$. Therefore, $\sep \models \exists{\vy} \thephi(t, t', \vy)$.
\end{proof}

\begin{lemma}\label{lem:adding_per_rule}
Let $\rho \in \rulesetnormminus$ be of the form 
$\thephi(x_1, x_2, \vy) \to \rpred(x_1, x_2)$. Then,
$$\hspace{-12ex}
\bigcup_%
{\hspace{12ex}
\rho_1\!(x)\in\,\ell,\ \rho_2\!(x)\in\,\ell'
}\hspace{-12ex}
\fcn{Add}_{\rpred,\ell,\ell'}(\sep, \coloringe) = \sep \cup \big\{\rpred(t,t') \mid \sep\models\exists{\vy}\thephi(t,t',\vy) \big\}.$$
\end{lemma}
\begin{proof}
All $\sep$ atoms are contained in both sides of the equation.
Considering any $\rpred(t,t') \not\in \sep$ we find it contained in the left-hand side \iffi $\rho_1(x) \in \coloringe(t) = \tau(t)$ and $\rho_2(x) \in \coloringe(t') = \tau(t')$ \iffi  $\sep \models \exists{\vy}\thephi(t,t'\vy)$ (by \cref{lem:col_query_correspondency}) \iffi $\rpred(t,t')$ is contained in the right-hand side of the equation.
\end{proof}

\noindent
Of course, the right hand side of the equation in \cref{lem:adding_per_rule} coincides with $\ksat{1}{\sep,\{\rho\}}$. 
\prm{We observe -- given that the coloring $\coloringe$ remains unaltered -- that the finitely many distinct applications of $\fcn{Add}_{\rpred,\ell,\ell'}$ in \cref{lem:adding_per_rule} are independent and can be sequentialized without changing the result. Moreover, no such application increases the cliquewidth of the colored instance it is applied to. Since these arguments can be lifted to the application of \emph{all} rules from $\rulesetnormminus$, we obtain}
\pin{We observe -- given that the coloring $\coloringe$ remains unaltered -- that the finitely many distinct applications of $\fcn{Add}_{\rpred,\ell,\ell'}$ in \cref{lem:adding_per_rule} are independent and can be chained without changing the result. Further, no such application increases the cliquewidth of the instance it is applied to. Since these arguments can be lifted to the application of \emph{all} rules from $\rulesetnormminus$, we get}
$$
\fcn{cw} \Big(
\sep \cup 
\hspace{-15ex}
\bigcup_%
{\hspace{15ex}
\thephi(x_1, x_2, \vy) {\to} \rpred(x_1, x_2) \, \in \, \rulesetnormminus} 
\hspace{-15ex}
\big\{\rpred(t,t') \mathrel{\big|} \sep\models\exists{\vy}\thephi(t,t',\vy) \big\}\Big) \leq n_\exists.
$$ 
Observe that the considered instance is equal to $\so$. Hence, we have established 
\cref{lem:binary-reformulated-plus}, 
concluding \cref{lem:binary-reformulated}, 
which finishes the proof of \cref{lem:binary-main}, 
thus entailing the desired \cref{thm:binary}.

%% file: higher-arities.tex
\subsection{The Case of Multi-Headed Rules}\label{sec:higher-arity}

We will now prove \cref{thm:higher-arity}, which implies that for multi-headed rules $\fus \not\subseteq \fcs$. To this end, we exhibit a $\fus$ rule set that yields universal models of infinite clique\-width. 
Let $\sigi = \set{\hpred, \vpred
}$ 
with $\hpred, \vpred$ binary predicates
. Let $\rsgrid$ denote the following rule set over $\sigi$:
$$
\begin{array}{@{}l@{\ \ \ \ \ \ }r@{\ }l@{\ \ \ \ \ \ \ \ }}
\text{(loop)} &  & \rightarrow \exists \, x \ \; \big(\,\hpred(x\,,x\,) \land \vpred (x,x\,)\big) \\[0.5ex]
\text{(grow)} &  & \rightarrow \exists yy' \big(\,\hpred(x\,,y\,) \land \vpred (x,y')\big) \\[0.5ex]
\text{(grid)} & \hpred(x,y) \land \vpred(x,x')  & \rightarrow \exists \, y' \; \big(\,\hpred(x'\!,y') \land \vpred(y,y')\big)
\end{array}
$$
We make use of $\rsgrid$ to establish \cref{thm:higher-arity}
, the proof of which consists of two parts. First, we provide a database $\dbg$ to form a knowledge base $(\dbg, \rsg)$ for which no universal model of finite cliquewidth exists (\cref{lem:no-universal-model-cw}). Second, we show that $\rsg$ is a finite unification set (\cref{lem:rsg-is-fus}). 

\msubsection{$\rsgrid$ is not a finite-cliquewidth set}
Toward establishing this result, let $\dbg = \set{
\top(\const{a})}
$ and define the instance $\ginf$, where $\const{a}$ is a constant and $y$ as well as $x_{i,j}$ with $i,j \in \mathbb{N}$ are nulls: 
\begin{eqnarray*}
\ginf & = & \big\{
\hpred(\const{a},x_{1,0}),\vpred(\const{a},x_{0,1}),\hpred(y,y),\vpred(y,y)\big\} \\
 & & \cup\,\big\{\hpred(x_{i,j},x_{i+1,j}), \vpred(x_{i,j},x_{i,j+1}) \mid  (i,j) \in (\mathbb{N} \times \mathbb{N}) \setminus \{(0,0)\} \big\}.
\end{eqnarray*}

\begin{figure}
    \centering
    \includegraphics[width=0.7\linewidth]{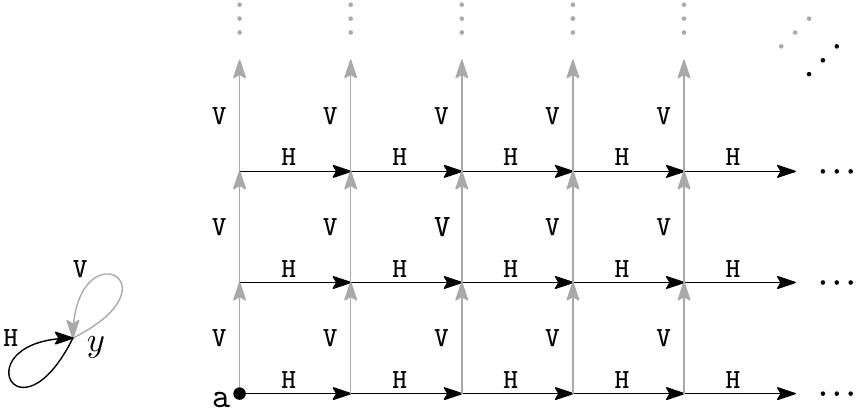}
    \vspace{-1ex}
    \caption{ The instance $\ginf$.}
    \label{fig:ginf}
\end{figure}

\cref{fig:ginf} depicts $\ginf$ graphically.
The following lemma summarizes consecutively established properties of $\ginf$ and its relationship with $\kbaseg$.

\begin{lemma}
With $\ginf$ and $\kbaseg$ as given above, we obtain:
\begin{itemize}
\item $\ginf$ has infinite cliquewidth,
\item $\ginf$ is a universal model of $\kbaseg$,
\item the only homomorphism $h: \ginf \to \ginf$ is the identity,
\item any universal model of $\kbaseg$ contains an induced sub-instance isomorphic to $\ginf$,
\item $\kbaseg$ has no universal model of finite cliquewidth.
\end{itemize}
\end{lemma}

From the last point (established via the insight that taking induced subinstances never increases clique\-width), the announced result immediately follows.

\begin{lemma}\label{lem:no-universal-model-cw}
$\rsg$ is not $\fcs$.
\end{lemma}

\msubsection{$\rsgrid$ is a finite-unification set}
Unfortunately, $\rsg$ does not fall into any of the known syntactic $\fus$ subclasses and showing that a provided rule set is $\fus$ is a notoriously difficult task in general.\footnote{This should not be too surprising however, as being $\fus$ is an undecidable property~\cite{BagLecMugSal11}.}
At least, thanks to the rule $(\mathrm{loop})$, every BCQ that is free of constants is always entailed and can thus be trivially re-written. 
For queries involving constants, we provide a hand-tailored query rewriting algorithm, along the lines of prior work exploring the multifariousness of $\fus$ \cite{OstMarCarRud22}. The required argument is quite elaborate and we just sketch the main ideas here due to space restrictions.

We make use of special queries, referred to as \emph{marked queries}: queries with some of their terms ``marked''  with the purpose of indicating terms that need to be mapped to database constants in the course of rewriting. The notion of query satisfaction is lifted to such marked CQs, which enables us to identify the subclass of \emph{properly marked queries} as those who actually have a match into some instance $\sg$, where $\db$ is an arbitrary database.

With these notions at hand, we can now define a principled rewriting procedure consisting of the exhaustive application of three different types of transformation rules. We show that this given set of transformations is in fact sound and complete, i.e., it produces correct first-order rewritings upon termination. Last, we provide a  termination argument by showing that the operations reduce certain features of the rewritten query. Thus, we arrive at the second announced result, completing the overall proof.

\begin{lemma}\label{lem:rsg-is-fus}
$\rsg$ is $\fus$.
\end{lemma}

\subsection{Fus and Expressive Queries}

\newcommand{\TM}{{\scaleobj{0.8}{\mathit{TM}}}}

Let us put the negative result of \cref{sec:higher-arity} into perspective.
Thanks to \cref{cor:fcsdecidable}, we know that $\fcs$ ensures decidability of arbitrary DaMSOQ entailment, a quite powerful class of queries. By contrast, $\fus$ is a notion tailored to CQs and unions thereof. As it so happens, searching for a method to establish decidability of DaMSOQ entailment for arbitrary $\fus$ rule sets turns out to be futile. This even holds for a fixed database and a fixed rule set. 
\begin{lemma}
DaMSOQ entailment is undecidable for $(\dbg,\rsg)$.
\end{lemma}
The corresponding proof is rather standard:
One can, given a deterministic Turing machine $\mathit{TM}$, create a DaMSOQ $(\mathfrak{q}_\TM,\Xi_\TM)$ that is even expressible in monadic datalog and, when evaluated over $\ginf$, uses the infinite grid to simulate a run of that Turing machine, resulting in a query match \ifandonlyif $\mathit{TM}$ halts on the empty tape.
As $\ginf$ is a universal model of $(\dbg,\rsg)$, and DaMSOQ satisfaction is preserved under homomorphisms, the entailment $(\dbg,\rsg) \models (\mathfrak{q}_\TM,\Xi_\TM)$ coincides with the termination of $\mathit{TM}$, concluding the argument.

%% file: conclusions.tex
\section{Conclusions and Future Work}

In this paper, we have introduced a generic framework, by means of which model-theoretic properties of a class of existential rule sets can be harnessed to establish that the entailment of \emph{datalog/MSO queries} -- a very expressive query formalism subsuming numerous popular query languages -- is decidable for that class.
We have put this framework to use by introducing \emph{finite-cliquewidth sets} and clarified this 
class's 
relationship with notable others, resulting in the insight that a plethora of known as well as hitherto unknown decidability results can be uniformly obtained via the decidability of DaMSOQ entailment over finite-cliquewidth sets of rules.
Our results entail various appealing directions for follow-up investigations: 
\begin{itemize}
\item 
Certainly, the class of single-headed binary $\fus$ rule sets is not the most general fragment of $\fus$ subsumed by $\fcs$. We strongly conjecture that many of the known, syntactically defined, ``concrete subclasses'' of $\fus$ are actually contained in $\fcs$, and we are confident that corresponding results can be established.
\item
Conversely, we are striving to put the notion of $\fcs$ to good use by identifying comprehensive, syntactically defined subclasses of $\fcs$, enabling decidable, highly expressive querying beyond the realms of $\bts$, $\fus$, or $\mathbf{fc}$.
As a case in point, the popular modeling feature of \emph{transitivity} of a relation has been difficult to accommodate in existing frameworks \cite{BagetBMR15,GanzingerMV99,KieronskiR21,SzwastT04}, whereas the observations presented in this paper seem to indicate that $\fcs$ can tolerate transitivity rules well and natively. 
\item
Finally, we are searching for even more general MSO-friendly width notions that give rise to classes of rule sets subsuming $\fcs$, and which also encompass $\bts$ natively, without arity restrictions or ``reification detours.'' 
\end{itemize}
Aside from these major avenues for future research, there are also interesting side roads worth exploring:
\begin{itemize}
\item 
Are there other, more general model-theoretic criteria of ``structural well-behavedness'' that, when ensured for universal models, guarantee decidability of query entailment ``just'' for (U)CQs?
Clearly, if we wanted $\fus$ to be subsumed by such a criterion, it would have to accept structures like $\ginf$ (i.e., infinite grids) and we would have to relinquish decidability of DaMSOQ entailment. 
\item
More generally: Are there other ``decidability sweet spots'' between expressibility of query classes and structural restrictions on universal models?
\end{itemize}

%% file: appendices/fo-well-decoration.tex
\section{Definition of $\Phi_\mathrm{well}$ (\NoCaseChange{\cref{sec:cwdef}})}

Recall that a decorated tree $\mathcal{T}$ is called a \emph{well-decorated tree} \iffi
\begin{itemize}
\item for every null $s \in \{0,1\}^*$, $\mathcal{T}$ contains exactly one fact $\mathrm{Dec}(s)$ with  $\mathrm{Dec} \in \decorum$,
\item for every $\const{c}\in \mathrm{Cnst}$, $\mathcal{T}$ contains at most one fact of the form $\pred{c}_k(s)$,
\item if $\pred{Add}_{\rpred,\vk}(s) \in \mathcal{T}$ or $\pred{Recolor}_{k\to k'}(s) \in \mathcal{T}$,
       then $\pred{Void}(s0)\not\in \mathcal{T}$ and $\pred{Void}(s1) \in \mathcal{T}$,
\item if $\oplus(s) \in \mathcal{T}$, then $\pred{Void}(s0),\pred{Void}(s1) \not\in \mathcal{T}$,
\item if $\pred{Void}(s) \in \mathcal{T}$ or $\pred{c}_k(s) \in \mathcal{T}$, then  $\pred{Void}(s0),\pred{Void}(s1) \in \mathcal{T}$.

\end{itemize}

Here, fixing $\cols$, $\Sigma$, and $\mathrm{Cnst}$,  we will define a first-order sentence $\Phi_\mathrm{well}$ such that for any $(\cols,\mathrm{Cnst},\Sigma)$-decorated tree $\mathcal{T}$, $\mathcal{T}$ is well-decorated \ifandonlyif $\mathcal{T} \models \Phi_\mathrm{well}$. 

Let us denote $\decorum$ with $\Gamma$ and we write $\varphi^i_{\pred{Void}}(x)$ to abbreviate $\exists z (\pred{Succ}_i(x,z) \wedge \pred{Void}(z))$. 
Consider the following (where $\pred{c} \in \mathrm{Cnst}$, $\pred{d} \in \mathrm{Cnst} \cup\{\pred{*}\}$):
\begin{align*}
    \Phi^{\mathrm{funct}} &= \forall{x} \bigvee_{\mathrm{Dec} \in \Gamma}\Big( \mathrm{Dec}(x) \wedge \scaleobj{0.8}{\bigwedge_{\scaleobj{1.25}{\mathrm{Dec}' \in \Gamma \setminus \{\mathrm{Dec}\}}}} \neg\mathrm{Dec}'(x) \Big)\\ 
    \Phi^{\mathrm{unique}}_{\pred{c}_k} &=  \forall{xy}\; \pred{c}_k(x) \Rightarrow (\pred{c}_k(y) \Rightarrow x = y) \wedge \bigwedge_{k'\in \cols\setminus\{k\}}\neg\pred{c}_{k'}(y)\\
\end{align*}
\begin{align*}
    \Phi_{\pred{d}_k} &=  & \forall{x}\; \pred{d}_k(x) & \Rightarrow \ \ \varphi^0_{\pred{Void}}(x) \wedge \varphi^1_{\pred{Void}}(x)\\
    \Phi_{\pred{Add}_{\rpred,\vk}} &= & \forall{x}\; \pred{Add}_{\rpred,\vk}(x) & \Rightarrow \neg \varphi^0_{\pred{Void}}(x) \wedge \varphi^1_{\pred{Void}}(x)\\
    \Phi_{\pred{Recolor}_{k\to k'}} &= & \forall{x}\; \pred{Recolor}_{k\to k'}(x) & \Rightarrow \neg \varphi^0_{\pred{Void}}(x) \wedge \varphi^1_{\pred{Void}}(x)\\
    \Phi_{\oplus} &= & \forall{x}\; \oplus(x) & \Rightarrow \neg \varphi^0_{\pred{Void}}(x) \wedge \neg\varphi^1_{\pred{Void}}(x)\\
    \Phi_{\pred{Void}} &= &\forall{x}\; \pred{Void}(x) & \Rightarrow \ \ \varphi^0_{\pred{Void}}(x) \wedge \varphi^1_{\pred{Void}}(x)\\
\end{align*}
Then we let 
$$
\Phi_\mathrm{well} = 
\Phi^{\mathrm{funct}} \wedge 
\hspace{-3ex}\bigwedge_{\hspace{3ex}\pred{c}_k\in \Gamma,\,\pred{c} \in \mathrm{Cnst}}\hspace{-3ex} \Phi^{\mathrm{unique}}_{\pred{c}_k} \wedge 
\bigwedge_{\mathrm{Dec}\in \Gamma} \Phi_{\mathrm{Dec}}.
$$

%% file: appendices/appendix.tex
\section{Proofs for \NoCaseChange{\cref{{sec:fcw-decidability}}}}\label{appendix:btsvsfcs}

\begin{definition}
We define 
$$
\varphi^\mathrm{dom}(x) = 
\hspace{-5ex}\bigvee_{\hspace{5ex}{\const{c} \in \mathrm{Const} \cup\{\const{*}\}}\atop {k \in \cols}}\hspace{-5ex}\const{c}_k(x)
\quad\text{ and }\quad
\varphi^\mathrm{Dom}(X) = \forall x. \Big( x{\in} X \to \varphi^\mathrm{dom}(x) \Big),
$$
which are lifted to variable sequences by
$$
\varphi^\mathrm{dom}(\vx) = \hspace{-4ex} \bigwedge_{\hspace{4ex} x\text{ from }\vx}\hspace{-3ex}\varphi^\mathrm{dom}(x)
\quad\text{ and }\quad
\varphi^\mathrm{Dom}(\VX) = \hspace{-4ex} \bigwedge_{\hspace{4ex} X\text{ from }\VX}\hspace{-3ex}\varphi^\mathrm{Dom}(X).
$$
Furthermore, let
$$
\thephi_\fcn{up}(x,y) = \szero(y,x) \vee \sone(y,x)
$$
and
$$
\varphi^\mathrm{colIn}_{k}(x,y) = \varphi^\mathrm{dom}(x) \wedge 
\forall_{\ell \in \cols} X_\ell. 
\Big(\big(\hspace{-3ex} \bigwedge_{\hspace{3ex}\psi \in M_{xy}}\hspace{-3ex}\psi\, \big) \to X_k(y)\Big),
$$
where $X_\ell$ are MSO set variables and $M_{xy}$ contains the following formulae (for all respective elements of $\decorum$):
\begin{align*}
 \const{c}_k(x) & \to X_{k}(x)\\ 
\forall x'\!y'\!. X_k(x') \wedge \thephi_\fcn{up}(x'\!,y') \wedge \pred{Add}_{\rpred,k',k''}(y') & \to X_{k}(y')\\ 
\forall x'\!y'\!. X_k(x') \wedge \thephi_\fcn{up}(x'\!,y') \wedge \pred{Add}_{\rpred,k'}(y') & \to X_{k}(y')\\ 
\forall x'\!y'\!. X_k(x') \wedge \thephi_\fcn{up}(x'\!,y') \wedge \pred{Recolor}_{k\to k'}(y') & \to X_{k'}(y')\\ 
\forall x'\!y'\!. X_k(x') \wedge \thephi_\fcn{up}(x'\!,y') \wedge \pred{Recolor}_{k'\!\to k''}(y') & \to X_{k}(y')\hspace{1.5ex}\text{for }k{\neq} k'\\ 
\forall x'\!y'\!. X_k(x') \wedge \thephi_\fcn{up}(x'\!,y') \wedge \oplus(y') & \to X_{k}(y')
\end{align*}
%
%
For each predicate $\rpred\in \Sigma$ of arity $n$, we define
\begin{align*}
\varphi_{\rpred_{n}}(x_1,\ldots,x_n) & = 
\exists z.\hspace{-4ex}
\bigvee_{\hspace{2ex}
\scaleobj{0.9}{(k_1,\ldots,k_n) \in \cols^n}
}\hspace{-4ex} 
\scaleobj{0.8}{\Big(}
\,\pred{Add}_{\rpred,k_1,\ldots,k_n}(z) \wedge\! \scaleobj{0.8}{\bigwedge_{\scaleobj{1.25}{1 \leq i \leq n}}} \varphi^\mathrm{colIn}_{k_i}(x_i,z) 
\scaleobj{0.8}{\Big)}
\end{align*}

\vspace{-3ex}
$\left.\right.$\defend
\end{definition}

\begin{definition}
Let $\mathcal{T}$ be a $(\cols,\mathrm{Cnst},\Sigma)$-decorated tree.
For $\const{c}\in \mathrm{Cnst}$, let $\const{c}^\mathcal{T}$ denote the unique $s \in \{0,1\}^*$ with $\pred{c}_k(s)\in \mathcal{T}$. Moreover, we define the injection $\fcn{orig}: \adom{\inst^\mathcal{T}} \to \{0,1\}^*$ 
by letting $\fcn{orig}(e)=e^\mathcal{T}$ if $e\in \mathrm{Cnst}$, and $\fcn{orig}(e)=e$ otherwise.\defend
\end{definition}

\begin{lemma}\label{lem:building_blocks_correct}
Given an $(\cols,\mathrm{Cnst},\Sigma)$-well-decorated tree $\mathcal{T}$, we have: 
\begin{enumerate}
\item $\mathcal{T},\{x \mapsto s\} \models \varphi^\mathrm{dom}(x)$ \iffi $s = \fcn{orig}(e)$ for some $e \in \adom{\inst^\mathcal{T}}$,
\item $\mathcal{T},\{X \mapsto S\} \models \varphi^\mathrm{Dom}(X)$ \iffi $S = \{\fcn{orig}(e) \mid e \in E\}$ for some $E \subseteq \adom{\inst^\mathcal{T}}$,
\item $\mathcal{T},\{x \mapsto s,y \mapsto s'\} \models \varphi^\fcn{up}(x,y)$ \iffi $s \in \{s'0,s'1\}$,
\item $\mathcal{T},\{x \mapsto s,y \mapsto s'\} \models \varphi^\mathrm{colIn}_{k}(x,y)$ \iffi $s = \fcn{orig}(e)$ for some $e \in \fcn{ent}^\mathcal{T}(s')$ and $\fcn{col}^\mathcal{T}_{s'}(e) = k$,
\item $\mathcal{T},\{\vx \mapsto \vs\} \models \varphi_{\rpred_{n}}(x)$ \iffi \,\! $\vs = \fcn{orig}(\ve)$ for some $\rpred(\ve) \in \inst^\mathcal{T}$.
\end{enumerate}
\end{lemma}

\begin{proof} Most cases are immediate with earlier cases being used to show latter cases.
For $(4)$, note that $M_{xy}$ serves as a monadic datalog program (for which $x$ acts like a constant) that, following the path from $s$ upward, assigns to each node $s^*$ the appropriate coloring predicate $X_{\fcn{col}^\mathcal{T}_{s^*}(\fcn{orig}^{-1}(s))}$ (to be shown by induction over the distance of $s^*$ from $s$). The subformula $$\forall_{\ell \in \cols} X_\ell.\Big(\big(\hspace{-3ex} \bigwedge_{\hspace{3ex}\psi \in M_{xy}}\hspace{-3ex}\psi\, \big) \to X_k(y)\Big)$$ then checks if the smallest fixpoint of $M_{xy}$ assigns $X_k$ to $s'$, which means $\fcn{col}^\mathcal{T}_{s'}(\fcn{orig}^{-1}(s))=k$.  
\end{proof}

\begin{definition}\label{def:normalized-trans}
An MSO formula $\xi$ over $\Sigma$ and $\mathrm{Cnst}$ is called \emph{normalized} \iffi (i) the only connectives it uses are existential quantifiers (over set and individual variables), conjunction, and negation, and (ii) the constants $\const{c}\in \mathrm{Cnst}$ only occur in equality statements of the form $x=\const{c}$.
We then recursively define the function $\fcn{trans}$, mapping normalized MSO formulae over $\Sigma$ and $\mathrm{Cnst}$ to MSO formulae over $\decorum \cup \{\szero,\sone\}$ as follows (where $\xi(\vx,\VX)$ denotes an MSO formula with free individual variables $\vx$ and free set variables $\VX$):
\newcommand{\tight}{\hspace{-1pt}}
\begin{align*}
x=y & \ \mapsto\  \varphi^\mathrm{dom}\tight(x) \land \varphi^\mathrm{dom}\tight(y) \wedge x=y \\ 
x=\const{c} & \ \mapsto\  \pred{c}_k(x) \\ 
\rpred(x_{1}, \ldots,x_{n}) & \ \mapsto\  \varphi_{\rpred_{n}}(x_{1}, \ldots,x_{n}) \\
\neg \xi(\vx,\VX) & \ \mapsto\ \varphi^\mathrm{dom}\tight(\vx) \wedge \varphi^\mathrm{Dom}\tight(\VX) \wedge \neg \fcn{trans}\big(\xi(\vx,\VX)\big) \\
\xi_1\tight(\vx_1,\VX_1) \wedge \xi_2\tight(\vx_2,\VX_2) & \ \mapsto\ \fcn{trans}\big(\xi_1\tight(\vx_1,\VX_1)\big) \wedge \fcn{trans}\big(\xi_2\tight(\vx_2,\VX_2)\big) \\
\exists x.\xi(\vx,\VX) & \ \mapsto\ \exists x.\Big(\varphi^\mathrm{dom}\tight(x) \wedge \fcn{trans}\big(\xi(\vx,\VX)\big)\Big)\\ 
\exists X.\xi(\vx,\VX) & \ \mapsto\ \exists X.\Big(\varphi^\mathrm{Dom}\tight(X) \wedge \fcn{trans}\big(\xi(\vx,\VX)\big)\Big)\\[-7.7ex] 
\end{align*}
~\defend
\end{definition}

\begin{observation}
Every MSO formula can be converted into an equivalent normalized one.
The function $\fcn{trans}$ preserves free individual variables and free set variables.
\end{observation}

\begin{proof} By making use of standard equivalences (e.g. $\thephi \lor \psi \equiv \neg \thephi \land \neg \psi$,  $\thephi \rightarrow \psi \equiv \neg (\thephi \land \neg \psi)$, $\forall x \thephi \equiv \neg \exists x \neg \thephi$, and $\forall X \thephi \equiv \neg \exists X \neg \thephi$), every MSO formula can be converted into an equivalent normalized MSO formula. By a routine induction on the structure of a given normalized, MSO formula over $\sig$ and $\mathrm{Cnst}$, one can show that free individual and set variables are preserved through applications of $\fcn{trans}$.
\end{proof}

\begin{lemma}\label{lem:translation_correct}
Given a finite set of colors $\cols$, a finite set of constants $\mathrm{Cnst}$, and a signature $\Sigma$, let  
\begin{itemize}
    \item $\xi = \xi(\vx,\VX)$ be a normalized MSO formula over $\Sigma$ and $\mathrm{Cnst}$ with free individual variables $\vx$ and free set variables $\VX$,
    \item $\mathcal{T}$ be a $(\cols,\mathrm{Cnst},\Sigma)$-well-decorated tree,
    \item $\mathrm{Assig}$ be the set of assignments $\mu$ that map variables from $\vx$ to terms from $\adom{\mathcal{T}}$ and set variables from $\VX$ to subsets of $\adom{\mathcal{T}}$,
    \item $\mathrm{Assig}'$ be the set of assignments $\mu$ that map variables from $\vx$ to terms from $\adom{\inst^\mathcal{T}}$ and set variables from $\VX$ to subsets of $\adom{\inst^\mathcal{T}}$,
    \item for some $\mu \in \mathrm{Assig}'$, denote with $\mu_\fcn{orig}$ the assignment from $\mathrm{Assig}$ that satisfies $\mu_\fcn{orig}(x) = \fcn{orig}(\mu(x))$ for every $x$ from $\vx$ and $\mu_\fcn{orig}(X) = \{\fcn{orig}(t) \mid t \in \mu(X)\}$ for every $X$ from $\VX$.  
\end{itemize}
Then,
$$
\{ \mu \in \mathrm{Assig} \mid \mathcal{T},\mu \models \fcn{trans}(\xi)\} 
= 
\{ \mu_\fcn{orig}' \mid \mu' \in \mathrm{Assig}', \inst^\mathcal{T},\mu' \models \xi\}. 
$$
\end{lemma}


\begin{proof} We prove the result by induction on the structure of $\xi$, which we take to be a normalized MSO formula over $\sig$ and $\cnst$.

\textit{Base case.} We have three cases to consider: either (i) $\xi$ is of the form $x = y$, (ii) $\xi$ is of the form $x = \const{c}$, or (iii) $\xi$ is of the form $\rpred(x_{1}, \ldots, x_{n})$. 
 We prove cases (i) and (iii) as (ii) is similar, and in both cases we first show that the left side of the equation is a subset of the right, and then that the right side of the equation is a subset of the left.

(i) Suppose that $\mathcal{T},\mu \models \fcn{trans}(x = y)$. By \cref{def:normalized-trans} above, $\fcn{trans}(x = y) = \varphi^\mathrm{dom}(x) \land \varphi^\mathrm{dom}(y) \wedge x=y$, implying that for $\mu(x) = \mu(y) = s \in \{0,1\}^{*}$. Moreover, since $\mathcal{T}, \mu \models \varphi^\mathrm{dom}(x) \land \varphi^\mathrm{dom}(y)$, there exists an $e \in \fcn{ent}^\mathcal{T}\!(\varepsilon)$ (that is, $e$ is a term in $\inst^{\mathcal{T}}$) such that $\orig(e) = s$. Let us define $\mu'(z) = \orig^{-1}(\mu(z))$ for $z \in \{x,y\}$. Then, $\mu'(x) = \orig^{-1}(\mu(x)) = e = \orig^{-1}(\mu(y)) = \mu'(y)$, showing that $\inst^{\mathcal{T}}, \mu' \models x = y$. Therefore, $\mu_{\orig}' = \mu$ is an element of the right side of the equation.

To show the opposite direction, let $\mu' \in \mathrm{Assig}'$ and suppose that $\inst^\mathcal{T},\mu' \models \xi$. Then, there exists a term $e$ in $\inst^{\mathcal{T}}$ (that is, an $e \in \fcn{ent}^\mathcal{T}\!(\varepsilon)$) such that $\mu'(x) = \mu'(y) = e$. As the function $\orig$ is injective, we have that $\mu_{\orig}'(x) = \mu_{\orig}'(y) = s \in \{0,1\}^{*}$, and also, since $e \in \fcn{ent}^\mathcal{T}\!(\varepsilon)$, it follows that $\mathcal{T}, \mu_{\orig}' \models \varphi^\mathrm{dom}(x) \land \varphi^\mathrm{dom}(y)$. Thus, $\mu_{\orig}'$ is an element of the left side of the equation, establishing that the two sets are indeed equal.

(iii) Suppose that $\mathcal{T}, \mu \models \fcn{trans}(\rpred(x_{1}, \ldots, x_{n}))$. By \cref{def:normalized-trans} above, $\fcn{trans}(\rpred(x_{1}, \ldots, x_{n})) = \thephi_{\rpred_{n}}(x_{1}, \ldots, x_{n})$, that is to say,
$$
\mathcal{T}, \mu \models \exists z. \bigvee_{(k_1,\ldots,k_n) \in \cols^{n}} \Big(\pred{Add}_{\rpred,k_1,\ldots,k_n}(z) \wedge \bigwedge_{1 \leq i \leq n} \varphi^\mathrm{colIn}_{k_i}(x_i,z) \Big)
$$
Then, for $\mu[s/z]$ (which is the same as $\mu$, but maps $z$ to $s$) and some $(k_1,\ldots,k_n) \in \cols^{n}$, we know that
$$
\mathcal{T}, \mu[s/z] \models \pred{Add}_{\rpred,k_1,\ldots,k_n}(z) \wedge \bigwedge_{1 \leq i \leq n} \varphi^\mathrm{colIn}_{k_i}(x_i,z).
$$
Let $s_{i} = \mu[s/z](x_{i})$ for $i \in \{1,\ldots,n\}$. By \cref{lem:building_blocks_correct}, there exist $e_{1}, \ldots, e_{n} \in \ent^{\mathcal{T}}(s)$ such that $e_{i} = \orig(s_{i})$ and $\col^{\mathcal{T}}_{s}(e_{i}) = k_{i}$, for $i \in \{1, \ldots,n\}$, and by \cref{def:entities-tree-to-instance}, we know that $\rpred(e_{1}, \ldots, e_{n}) \in \inst^{\mathcal{T}}$. Let $\mu'(y) = \orig^{-1}(\mu(y))$ for $y \in \{x_{1},\ldots,x_{n}\}$. Then, it follows that $\inst^{\mathcal{T}}, \mu' \models \rpred(x_{1},\ldots, x_{n})$, showing that $\mu = \mu_{\orig}'$ is an element of the right side of the equation.

For the opposite direction, suppose that $\mu' \in \mathrm{Assig}'$ and $\inst^{\mathcal{T}}, \mu' \models \rpred(x_{1}, \ldots, x_{n})$. Therefore, there exist terms $e_{1}, \ldots, e_{n}$ in $\inst^{\mathcal{T}}$ 
 such that $\mu'(x_{i}) = e_{i}$ for $i \in \{1, \ldots, n\}$, and $\rpred(e_{1}, \ldots, e_{n}) \in \inst^{\mathcal{T}}$. Let $\orig(e_{i}) = s_{i}$, for $i \in \{1, \ldots, n\}$. Since $\rpred(e_{1}, \ldots, e_{n}) \in \inst^{\mathcal{T}}$, by \cref{def:entities-tree-to-instance} there must exist an $s$ in $\mathcal{T}$ with $e_{1},\ldots, e_{n} \in \ent^{\mathcal{T}}(s)$ such that $\pred{Add}_{\rpred,k_{1}, \ldots, k_{n}}(s) \in \mathcal{T}$ and $\col_{s}^{\mathcal{T}}(e_{i}) = k_{i}$, for $i \in\{1, \ldots,n\}$. Therefore, $\mathcal{T}, \mu_{\orig}' \models \thephi_{\rpred_{n}}(x_{1}, \ldots,x_{n})$, showing that $\mu_{\orig}'$ is an element of the left side of the equation.
 
 \textit{Inductive step.} We have four cases to consider: either (i) $\xi$ is of the form $\neg \xi' (\vx,\VX)$, (ii) $\xi$ is of the form $\xi_{1}(\vx_{1},\VX_{1}) \land \xi_{2}(\vx_{2},\VX_{2})$, (iii) $\xi$ is of the form $\exists x. \xi'(\vx,\VX)$, or (iv) $\xi$ is of the form $\exists X. \xi'(\vx,\VX)$. We show cases (ii) and (iv) as the remaining cases are proven in a similar fashion. As before, we first show that the left side of the equation is a subset of the right, and then the opposite.
 
 (ii) Suppose that $\mathcal{T},\mu \models \mathrm{trans}(\xi)$. By our assumption, $\mathrm{trans}(\xi) = \mathrm{trans}(\xi_{1}(\vx_{1},\VX_{1}) \land \xi_{2}(\vx_{2},\VX_{2}))$, which by \cref{def:normalized-trans}, is equal to $\mathrm{trans}(\xi_{1}(\vx_{1},\VX_{1})) \land \mathrm{trans}(\xi_{2}(\vx_{2},\VX_{2}))$. It follows that $\mathcal{T}, (\mu \restriction i) \models \mathrm{trans}(\xi_{i}(\vx_{i},\VX_{i}))$ for $i \in \{1,2\}$ (with $(\mu \restriction i) = \mu \restriction \{\vx_{i},\VX_{i}\}$), from which it follows that there exist $\mu^{1}$ and $\mu^{2}$ such that $\inst^{\mathcal{T}}, \mu^{i} \models \xi_{i}(\vx_{i},\VX_{i})$ and $\mu^{i}_{\orig} = (\mu \restriction i)$ for $i \in \{1,2\}$. Let us define
 \[
  \mu'(z) =
  \begin{cases}
  (\mu^{i})(z) & \text{if $z$ is in $\vx_{1},\vx_{2}$} \\
  (\mu^{i})(Z) & \text{if $Z$ is in $\VX_{1},\VX_{2}$.}
  \end{cases}
\]
Observe that $\inst^{\mathcal{T}}, \mu' \models \xi_{1}(\vx_{1},\VX_{1}) \land \xi_{2}(\vx_{2},\VX_{2})$ and that $\mu_{\orig}' = \mu$, thus showing that $\mu$ is an element of the right side of the equation.

For the opposite direction, let $\mu' \in \mathrm{Assig}'$ and assume that $\inst^{\mathcal{T}}, \mu' \models \xi_{1}(\vx_{1},\VX_{1}) \land \xi_{2}(\vx_{2},\VX_{2})$. Let $(\mu' \restriction i) = \mu' \restriction \{\vx_{i},\VX_{i}\}$ and observe that $\inst^{\mathcal{T}}, (\mu' \restriction i) \models \xi_{i}(\vx_{i},\VX_{i})$ for $i \in \{1,2\}$. Therefore, by IH, we have $\mathcal{T}, (\mu' \restriction i)_{\orig} \models \mathrm{trans}(\xi_{i}(\vx_{i},\VX_{i}))$ for $i \in \{1,2\}$, which entails that $\mathcal{T}, \mu_{\orig}' \models \mathrm{trans}(\xi_{1}(\vx_{1},\VX_{1})) \land \mathrm{trans}(\xi_{2}(\vx_{2},\VX_{2}))$. Hence, $\mu_{\orig}'$ is an element of the left side of the equation.

(iv) Let us suppose that $\mathcal{T},\mu \models \mathrm{trans}(\exists X. \xi(\vx,\VX))$, that is to say,
$$
\mathcal{T}, \mu \models \exists X.\Big(\varphi^\mathrm{Dom}(X) \wedge \fcn{trans}\big(\xi(\vx,\VX)\big)\Big).
$$
It follows that $\mathcal{T}, \mu[X/S] \models \varphi^\mathrm{Dom}(X) \wedge \fcn{trans}\big(\xi(\vx,\VX)\big)$ for some $S \subseteq \{0,1\}^{*}$ with $\mu[X/S]$ the same as the assignment $\mu$, but where the set variable $X$ is mapped to $S$. By IH, we know that $\inst^{\mathcal{T}}, \mu'[X/S] \models \xi(\vx,\VX)$, where $\mu'[X/S](z) = \orig^{-1}(\mu[X/S](z))$ for $z$ in $\vx$ and $\mu'[X/S](Z) = \orig^{-1}(\mu[X/S](Z))$ for $Z$ in $\VX$. Note that $\mu'[X/S]$ is defined because $\mathcal{T}, \mu[X/S] \models \varphi^\mathrm{Dom}(X)$, i.e., there exists a set $E \subseteq \ent^{\mathcal{T}}(\varepsilon)$ such that $S = \{\orig(e) \ | \ e \in E\}$. Hence, $\inst^{\mathcal{T}}, \mu' \models \exists X. \xi(\vx,\VX)$, showing that $\mu_{\orig}' = \mu$ is an element of the right side of the equation.

For the opposite direction, let $\mu' \in \mathrm{Assig}'$ and assume that $\inst^{\mathcal{T}}, \mu' \models \exists X. \xi(\vx,\VX)$. This entails that $\inst^{\mathcal{T}}, \mu'[X/E] \models \xi(\vx,\VX)$ for some set $E$ of terms in $\inst^{\mathcal{T}}$ (i.e., for some $E \subseteq \ent^{\mathcal{T}}(\varepsilon)$) with $\mu'$ being the same as $\mu'$, but where the set variable $X$ is mapped to $E$. By IH, we have that $\mathcal{T}, \mu'[X/E]_{\orig} \models \xi(\vx,\VX)$, where $\mu'[X/E]_{\orig}$ is equal to the assignment $\mu_{\orig}'[X/S]$ with $S = \{\orig(e) \ | \ e \in E\}$. We have that
$$
\mathcal{T}, \mu_{\orig}' \models \exists X.\Big(\varphi^\mathrm{Dom}(X) \wedge \fcn{trans}\big(\xi(\vx,\VX)\big)\Big)
$$
is a consequence, thus showing that $\mu_{\orig}'$ is an element of the left side of the equation.
\end{proof}

\begin{corollary}\label{cor:MSOsentence_cw_decidable}
Given a finite set of colors $\cols$, a finite set of constants $\mathrm{Cnst}$, and a signature $\Sigma$, if
\begin{itemize}
    \item $\mathcal{T}$ is a $(\cols,\mathrm{Cnst},\Sigma)$-well-decorated tree, and 
    \item $\Xi$ is a normalized MSO sentence over $\Sigma$ and $\mathrm{Cnst}$,
\end{itemize}
then
$$
\mathcal{T} \models \fcn{trans}(\Xi) \text{\ \ \ \iffi \ \ } 
\inst^\mathcal{T} \models \Xi. 
$$
\end{corollary}

\begin{proof} The corollary follows directly from \Cref{lem:translation_correct} above, as it serves as a special case where the normalized MSO formula over $\Sigma$ and $\mathrm{Cnst}$ contains no free variables.  
\end{proof}

\begin{customthm}{\ref{thm:cliquewidth-k-MSO-decidable}}
Let $n\in \mathbb{N}$. Determining if a given MSO formula $\Xi$ 
has a model $\inst$ with $\cw{\inst} \leq n$ is decidable.
\end{customthm}

\begin{proof}
Let $\cols=\{1,\ldots,n\}$, $\mathrm{Cnst}$ be the set of all constants mentioned in $\Xi$, and $\Xi' = \fcn{trans}(\Xi)$. By \Cref{cor:MSOsentence_cw_decidable}, we know that for the sentence $\Xi$ (which uses the signature $\decorum \cup \{\szero,\sone\}$) and any $(\cols,\mathrm{Cnst},\Sigma)$-well-decorated tree $\mathcal{T}$, $\instance^\mathcal{T} \models \Xi$ \iffi $\mathcal{T} \models \Xi'$. Hence, in order to check if $\Xi$ is satisfiable in some instance of cliquewidth less than or equal to $n$, we can compute $\Xi'$ and check if it is satisfiable in an $(\cols,\mathrm{Cnst},\Sigma)$-well-decorated tree $\mathcal{T}$. Checking if $\Xi'$ is satisfiable in some well-decorated tree $\mathcal{T}$ is equivalent to checking if $\Xi'\wedge \Phi_\mathrm{well}$ is satisfiable in an arbitrary $(\cols,\mathrm{Cnst},\Sigma)$-decorated tree $\mathcal{T}'$. Let us define $\Xi''$ to be $\Xi'\wedge \Phi_\mathrm{well}$ where (i) all predicates from $\decorum$ are replaced by MSO set variables, which are (ii) quantified over existentially. Observe that $\Xi''$ is an MSO formula over the signature $\{\szero,\sone\}$, which will be valid on $\ibtree$ \iffi some decoration exists satisfying $\Xi'\wedge \Phi_\mathrm{well}$. Thus, we have reduced our problem to checking the validity of an MSO sentence on $\ibtree$, which is decidable by Rabin's Tree Theorem \cite{Rabin69}. 
\end{proof}

\begin{customthm}{\ref{thm:tree-width-implies-clique-width}}
Let $\inst$ be a countable instance over a binary signature.
If $\inst$ has finite treewidth, then $\inst$ has finite cliquewidth.
\end{customthm}
\input{appendices/tree-width-clique-proof-simplified}

%% file: appendices/tree-width-clique-proof-simplified.tex
\begin{proof}
Let $T = (V,E)$ be a tree-decomposition of $\inst$ witnessing that $\inst$ has a treewidth of $k$. 
Then, countability of $\inst$ implies countability (and thus countable degree) of $(V,E)$.

Therefore, $(V,E)$ can be oriented (by repeated introduction of $E$-neighboring copies of nodes) and transformed into a full binary tree $T'=(V',E',\termset)$ with $V'= \{0,1\}^*$, $E'= \{(s,s0), (s,s1) \mid s\in\{0,1\}^*\}$, and $\termset:V'\to 2^{\adom{\inst}}$, which inherits all tree decomposition properties from $T$. The moderate reformulation using the $\termset$ function is necessary due to the introduction of node copies containing identical sets of terms.

Let $\fcn{slot}: \adom{\inst} \to \{0,\ldots,k\}$ be a function that satisfies $\fcn{slot}(t_1)\neq \fcn{slot}(t_2)$ for any $t_1\neq t_2$ with $\{t_1,t_2\} \subseteq \termset(s)$ for some $s\in V'$; note that the existence of such a function is guaranteed by the treewidth bound of $T'$ inherited from $T$. Moreover, for each domain element $t\in \adom{\inst}$ let $\fcn{pivot}(t)$ denote the node $s \in V'$ with $t \in \termset(s)$ that is closest to the root of $T'$.
For each atom $\rpred(t_1,t_2) \in \inst$, one of the following three statements must hold for $T'$:

\begin{itemize}
    \item $\fcn{pivot}(t_1) = \fcn{pivot}(t_2)$, denoted $t_1 \approx t_2$,
    \item $\fcn{pivot}(t_1)$ is an ancestor of $\fcn{pivot}(t_2)$, denoted $t_1 \prec t_2$,
    \item $\fcn{pivot}(t_2)$ is an ancestor of $\fcn{pivot}(t_1)$, denoted $t_2 \prec t_1$.
\end{itemize}

Let now the finite set $\cols$ of colors consist of pairs $(j,S)$ with $0\leq j\leq k$ and $S$ being a subset of
$$ 
\{ \rpred(\downarrow) \mid \rpred \in \Sigma_1 \} \cup \{ \rpred(\downarrow,i), \rpred(i,\downarrow), \rpred(\downarrow,\downarrow) \mid \rpred \in \Sigma_2, 0\leq i\leq k\}.
$$

Let us provide a description of the process transforming~$T'$ into a $(\cols,\mathrm{Cnst},\Sigma)$–well-decorated infinite binary tree representing $\inst$. For a node $s$ of $T'$, we define $\fcn{tf}(s)$ as follows:

\begin{itemize}
\item Take a fresh node $v$, decorate it with $\oplus$, and set its two children to be $\fcn{tf}(s0)$ and $\fcn{tf}(s1)$ with $s0$ and $s1$ the two children of $s$ in $T'$. Set $v$ to be the {\em current root} of the transformation.
\item Next, go through all $t \in \adom{\inst}$ with $\fcn{pivot}(t)=s$, which we assume are in ascending order by $\fcn{slot}(t)$ (there are at most $k+1$ of such terms) and do the following:
    \begin{enumerate}
    \item Take a fresh node $v'$, decorate it with $\oplus$, and set its children to be the current root and a fresh node decorated with $t_{(j,S)}$ if $t \in \mathrm{Cnst}$, or $\const{*}_{(j,S)}$ otherwise, where $j = \fcn{slot}(t)$ and $S$ contains:
        \begin{itemize}
        \item[$\circ$] all $\rpred(\downarrow)$ with $\rpred(t) \in \inst$ and all $\rpred(\downarrow,\downarrow)$ with $\rpred(t,t) \in \inst$, 
        \item[$\circ$] all $\rpred(\downarrow,i)$ with $\rpred(t,t') \in \inst$ for some $t' \in \adom{\inst}$ with $\fcn{slot}(t')=i$ and where $t'\prec t$ or ($t'\approx t$ and $i>j$), and
        \item[$\circ$] all $\rpred(i,\downarrow)$ with $\rpred(t',t) \in \inst$ for some $t' \in \adom{\inst}$ with $\fcn{slot}(t')=i$ and where $t'\prec t$ or ($t'\approx t$ and $i>j$).
        \end{itemize}
    \item Build a chain $P$ of nodes decorated with $\pred{Add}$ and $\pred{Recolor}$ predicates to reflect the following set of operations:
        \begin{itemize}
        \item[$\circ$] for every color $(j,S)$, for which $S$ contains some $\rpred(\downarrow)$, perform $\fcn{Add}_{\rpred,(j,S)}$ and then $\fcn{Recolor}_{(j,S)\to(j,S \setminus \{\rpred(\downarrow)\})}$,
        \item[$\circ$] for every color $(j,S)$ such that $S$ contains some $\rpred(\downarrow,\downarrow)$, perform $\fcn{Add}_{\rpred,(j,S),(j,S)}$ and then $\fcn{Recolor}_{(j,S)\to(j,S \setminus \{\rpred(\downarrow,\downarrow)\})}$,
        \item[$\circ$] for any two colors $(j,S)$ and $(i,S')$, for which $S'$ contains some $\rpred(\downarrow,j)$, perform $\fcn{Add}_{\rpred,(i,S'),(j,S)}$ and immediately afterward perform $\fcn{Recolor}_{(i,S')\to(i,S' \setminus \{\rpred(\downarrow,j)\})}$,
        \item[$\circ$] for any two colors $(j,S)$ and $(i,S')$, for which $S'$ contains some $\rpred(j,\downarrow)$, perform $\fcn{Add}_{\rpred,(j,S),(i,S')}$ and immediately afterward perform $\fcn{Recolor}_{(i,S')\to(i,S' \setminus \{\rpred(\downarrow,j)\})}$.
        \end{itemize}
    \item Set the current root to be the highest node of $P$ and set $v'$ to be the child of the lowest node of $P$.
    \end{enumerate}
\end{itemize}

To obtain $\fcn{tf}(s)$, take the resulting decorated tree and extend the resulting tree to a full binary one by adding missing elements and decorating them with $\pred{Void}$.

Finally, take the root $\varepsilon$ of $T'$, obtain $\fcn{tf}(\varepsilon)$, add on top nodes decorated with $\pred{Add}_\rpred$ for all nullary predicates $\rpred \in \Sigma$ with $\rpred \in \inst$ (in arbitrary order), again 
add missing nodes, and decorate them with $\pred{Void}$.
The resulting $(\cols,\mathrm{Cnst},\Sigma)$-well-decorated tree represents $\inst$.
\end{proof}

%% file: appendices/reify.tex
\section{Justification of \Cref{ex:badnewsternary} (\NoCaseChange{\cref{sec:cw-vs-tw}})}\label{app:example-15-justification}

To provide some insight into why $\inst_\mathrm{tern}$ from \Cref{ex:badnewsternary} does not have finite cliquewidth, first observe that the instance has the structure of the natural number line, but with an additional term (viz. $-1$) that stands in a ternary relation with each natrual number and its successor. We will argue that no coloring $\coloring$ exists such that $(\inst_{\mathrm{tern}},\coloring)$ is isomorphic to a colored instance represented by some well-decorated tree $\mathcal{T}$. For a contradiction, let us suppose the contrary. Then, we may assume w.l.o.g. that some node $s$ in $\mathcal{T}$ is decorated with $-1_k(s)$ (if $-1$ is assumed to be a null, introduced by a node $\const{*}_k(s)$, then an almost identical argument can be given). Now, we know that at some node $s'$ along the path from $s$ to the root $\varepsilon$ of $\mathcal{T}$ a disjoint union will occur introducing an infinite amount of pairs of consecutive natural numbers. In order for the ternary edges to be added between $-1$ and each pair of consecutive natural numbers, an infinite number of $\pred{Add}_{\rpred,\vk}$ operations are required. However, since only a finite number of operations can be performed after the introduction of $-1$ (as only a finite number of nodes exist between $s$ and the root $\varepsilon$ of $\mathcal{T}$), all such ternary edges cannot be added, showing that for all colorings $\coloring$, $(\inst_{\mathrm{tern}},\coloring)$ is not isomorphic to a colored instance represented by some well-decorated tree. Concomitantly, the rule set $\ruleset_\mathrm{tern} = \{ \rpred(v,x,y) \to \exists z \rpred(v,y,z) \}$ is $\bts$, but not $\fcs$.

\section{Properties of Reification (\NoCaseChange{\cref{sec:cw-vs-tw}})}

\newcommand{\reify}[1]{\fcn{reify}(#1)}
\newcommand{\reifym}[1]{\fcn{reify}^{-1}(#1)}
\newcommand{\homm}{h}
\newcommand{\reifynew}[1]{\fcn{new}(\reify{#1})}
\newcommand{\newreify}[1]{\reifynew{#1}}
\newcommand{\pn}[1]{\reify{\ksat{#1}{\inst, \ruleset}} \equiv \ksat{#1}{\reify{\inst}, \reify{\ruleset}}}
\newcommand{\pnl}[1]{\reify{\ksat{#1}{\inst, \ruleset}}}
\newcommand{\pnr}[1]{\ksat{#1}{\reify{\inst}, \reify{\ruleset}}}

\input{appendices/reification/reify-i}


\input{appendices/reification/reify-iv}

\input{appendices/reification/reify-ii}

\input{appendices/reification/reify-iii}

%% file: appendices/reification/reify-i.tex
\subsection{$\tw{\instance} \in \mathbb{N}$ implies $\tw{\fcn{reify}(\instance)}\in \mathbb{N}$ and $\cw{\fcn{reify}(\instance)}\in \mathbb{N}$.}

Let $\Sigma$ be a finite signature, $\inst$ an instance over $\Sigma$, $T = (V, E)$ a tree decomposition of $\inst$ of width $n$, and denote by $\inst_{\ge 3} \subseteq \inst$ the sub-instance of $\inst$ consisting of all atoms over $\Sigma_{\ge 3}$ from $\inst$. Define $\reify{T}$ to be a pair $(V', E')$ consisting of:
\begin{itemize}
    \item Let $V'$ contain the sets $X'$ obtained in the following way: For $X \in V$ of the original tree decomposition, let \[ X' = X \cup \{ u_{\alpha} \mid \alpha \in \inst_{\ge 3}[X] \} \] (where $\inst_{\ge 3}[X]$ denotes the induced sub-instance of $\inst_{\ge 3}$ with active domain $X$). This gives rise to a bijection $V \to V', X\mapsto X'$, which we will use implicitly in the following.
    \item $E' = \{ (X_1', X_2') \in V' \times V' \mid (X_1, X_2)\in E \}$
\end{itemize}

\begin{observation}\label{obs:reify-tree-ftw}
    Let $\Sigma$ be a finite signature, $\inst$ an instance over $\Sigma$ of finite treewidth, and $T = (V, E)$ a tree decomposition of $\inst$ of width $n$. Then $\reify{T}$ is a tree decomposition of $\reify{\inst}$ of finite width.
\end{observation}
\begin{proof}
    Let $\inst$ be an instance of treewidth $n$ and $T=(V,E)$ a tree decomposition with width $n$. $\reify{T} = (V', E')$ is a tree decomposition of $\reify{\inst}$:
    \begin{itemize}
        \item We will establish that $\bigcup_{X'\in V'}X' = \adom{\reify{\inst}}$. By definition each $X'\in V'$ is a subset of $\adom{\reify{\inst}}$, which shows the $\subseteq$ inclusion. On the other hand, let $t\in\adom{\reify{\inst}}$. If $t\in \adom{\inst}$, then there exists an $X\in V$ such that $t\in X$ and hence by definition of $\reify{T}$, the node $X'$ in $\reify{T}$ corresponding to $X$ contains $t$. So assume $t\not\in\adom{\inst}$. Then, there exists an atom $\alpha\in \inst$ such that $t = u_{\alpha}$. Furthermore, there exists an $X\in V$ such that every term of $\alpha$ is contained in $X$. By construction this means that $u_{\alpha}\in X'$, where $X'$ is the node in $\reify{T}$ corresponding to node $X$. We obtain thus the inclusion $\supseteq$ and hence have shown the claim.
        \item Let $\rpred(t_1, \ldots, t_n)\in \reify{\inst}$. In order to show that this implies $t_{1}, \ldots, t_{n} \in X'$ for some $X'\in V'$, we discuss two cases: Either 
        $\rpred\in \Sigma_{\leq2}$ or $\rpred\in \Sigma^{\mathrm{rf}}\setminus \Sigma_{\leq2}$ (which means that 
        $n = 2$). In the former case, there exists an $X\in V$ such that $t_{1}, \ldots, t_{n}\in X$, and hence (by definition) $t_{1}, \ldots, t_{n}\in X'$, the corresponding node of $X$ in $\reify{T}$. In the latter case, there exists an atom $\alpha\in\inst$ and a null $u_{\alpha}$ such that 
         $\rpred(t_1, \ldots, t_n)$ is of the form $\rpred(t_{1}, u_{\alpha})$. By definition, the node $X'$ in $\reify{T}$ corresponding to $X$ contains not just $t_{1}$, but $u_{\alpha}$ as well. Hence, the claim follows.
        \item Last, we show that for every $t\in \adom{\reify{\inst}}$ the subgraph of $\reify{T}$ induced by all nodes $X'$ containing $t$ is connected. Again, we have two cases: Either $t\in \adom{\inst}$ or $t\not\in\adom{\inst}$ with $t\in \adom{\reify{\inst}}$. The former case is immediate by the definition of $\reify{T}$, so we focus on the latter case, and assume that $t = u_{\alpha}$ for some  
        $\alpha \in \inst$. By the definition of $\reify{T}$, $u_{\alpha}$ is in exactly those $X'\in V'$ for which the corresponding $X\in V$ contains all of $t_{1}, \ldots, t_{n}$. By $T$ being a tree decomposition, the set of all those $X$ induces a connected subgraph of $T$. By the definition of $\reify{T}$, the nodes of $\reify{T}$ are connected in exactly the same way in $E'$ as their corresponding nodes in $T$; hence, we immediately obtain that all $X'\in V'$ containing $t$ induce a connected subgraph of $\reify{T}$.
    \end{itemize}
    This establishes, that $\reify{T}$ is a tree decomposition of $\reify{\inst}$.
    
    Last, it is clear that $\reify{T}$ has a width of at most $n + K$ where $K$ is the maximal number of atoms over $\Sigma_{\ge 3}$ using elements from a set of cardinality $n+1$, as this constitutes the maximal amount of ``reification-nulls'' introduced in a node $X'$ in $\reify{T}$ in comparison to the corresponding node $X$ in $T$.
\end{proof}
\begin{corollary}\label{cor:reify-i}
    Let $\Sigma$ be a finite signature. If $\inst$ over $\Sigma$ has finite treewidth, then $\tw{\reify{\inst}}$ and $\cw{\reify{\inst}}$ are finite.
\end{corollary}
\begin{proof}
    Since $\inst$ has finite treewidth, there is a tree decomposition $T = (V, E)$ of $\inst$. Then, $\reify{T}$ is a tree decomposition of $\reify{I}$ of finite width. This implies finite treewidth and, by \cref{thm:tree-width-implies-clique-width}, finite cliquewidth of $\reify{I}$.
\end{proof}

%% file: appendices/reification/reify-iv.tex
\subsection{$\sat{\hspace{-0.3mm}\fcn{reify}(\instance),\!\fcn{reify}(\ruleset)\hspace{-0.3mm}} \equiv \fcn{reify}(\sat{\instance\!,\hspace{-0.3mm}\ruleset})$}

\newcommand{\homfrom}{h_{\leftarrow}}
\newcommand{\homto}{h_\to}
\newcommand{\homtop}{h'_\to}



\begin{definition}\label{def:reify-hom-to}
Let $h$ be a homomorphism from some instance $\inst$ to other instance $\inst'$. We let $\reify{h}$ denote the following homomorphism from $\reify{\inst}$ to $\reify{\inst'}$:
$$ \reify{h}(x) = \begin{cases}
\homm(x)\;&\text{if $x \in \termset(\inst)$},\\
u_{h(\alpha)}\;&\text{if $x$ = $u_{\alpha}$},
 \text{where $\alpha$ is an atom of $\inst$}.\ \ \ 
\end{cases}
$$

\vspace{-4ex}
~\defend
\end{definition}

\begin{fact}\label{fact:reification-homm}
If $\homm$ is a homomorphism from $\inst$ to $\inst'$, then $\reify{\homm}$ is a homomorphism from $\reify{\inst}$ to $\reify{\inst'}$.
\end{fact}

\begin{definition}\label{def:reify-hom-from}
Let $\inst$ and $\inst'$ be instances and let $h$ be a homomorphism from $\reify{\inst}$ to $\reify{\inst'}$. We let $\reifym{h}$ denote the function $t \mapsto \homm(t)$ from $\inst$ to $\inst'$, i.e. the restriction of $h$ to elements of $\adom{\inst}$.
\defend
\end{definition}

\begin{fact}\label{fact:reification-homm-two}
If $\homm$ is a homomorphism from $\reify{\inst}$ to $\reify{\inst'}$, then $\reifym{\homm}$ is a homomorphism from $\inst$ to $\inst'$.
\end{fact}

\begin{observation}\label{obs:reification-bijection}
Given two instances $\inst$ and $\inst'$ and two sets of homomorphisms $H = \set{h : \inst \to \inst'}$ and $H' = \set{h : \reify{\inst} \to \reify{\inst'}}$, the function $\reify{\cdot}$ is a bijection from $H$ to $H'$.
\end{observation}
\begin{proof}
Follows from \cref{def:reify-hom-to} and \cref{def:reify-hom-from}.
\end{proof}

We now prove the following statement by induction on $n$: For every instance $\inst$ and every rule set $\ruleset$, 
$$
\pn{n}.
$$

\textit{Base case.} We note that the base case for $n = 0$ is trivial since
$$
\pnl{0} = \reify{\inst} = \pnr{0}.
$$
Hence, we show the case for $n = 1$, that is, we show that
$$
\pnl{1} \equiv \pnr{1}.
$$
We define two homomorphisms: (i) $\homto$ from $\pnl{1}$ to $\pnr{1}$ and (ii) $\homfrom$ from $\pnr{1}$ to $\pnl{1}$. For any term $t$ in $\reify{\inst}$ we set $\homto(t) = t = \homfrom(t)$. By \cref{obs:reification-bijection}, we know that any trigger $\tau = (\rho, \homm)$ in $(\inst, \ruleset)$ is mirrored in $(\reify{\inst}, \reify{\ruleset})$ by $\tau' = (\reify{\rho}, \reify{\homm})$. Moreover, we know that both $\tau$ and $\tau'$ are satisfied in $\pnl{1}$ and $\pnr{1}$  by the homomorphisms $\homm'_l : \head(\rho) \to \ksat{1}{\inst, \rs}$ and $\homm_r : \head(\reify{\rho}) \to \pnr{1}$, respectively. Let $\homm_l = \reify{\homm'_l}$ be a homomorphism from $\head(\reify{\rho})$ to $\pnl{1}$.

Let us assume that an application of $\tau$ resulted in the creation of a new atom $\pred{R}_i(t,u)$ in $\pnl{1}$ (after reification). Let $\pred{R}_i(x,v_u)$ be the respective atom of the head of $\reify{\rho}$. Then, we set $\homto(u) = \homm_r(v_u)$, and if $x$ is existentially quantified, we set $\homto(t) = \homm_r(x)$. 

Let us assume that an application of $\tau'$ resulted in the creation of a new atom $\pred{R}_i(t,u)$ in $\pnr{1}$. Let $\pred{R}_i(x,v_u)$ be the respective atom of the head of $\reify{\rho}$. Then, we set $\homfrom(u) = \homm_l(v_u)$, and if $x$ is existentially quantified, we set $\homfrom(t) = \homm_l(x)$.

\textit{Inductive step.} Now, given a natural number $n$, we will prove the claim for $n+ 1$, assuming that it holds for every natural number $i \leq n$.
\begin{align*}
\pnl{n + 1} =&\; \reify{\ksat{1}{\ksat{n}{\inst, \ruleset}, \ruleset}} \\
            \equiv&\; \ksat{1}{\reify{\ksat{n}{\inst, \ruleset}}, \reify{\ruleset}}\\
            \equiv&\; \ksat{1}{\ksat{n}{\reify{\inst}, \reify{\ruleset}}, \reify{\ruleset}}\\
            =&\; \pnr{n + 1}
\end{align*}
Therefore, we may conclude the following: 
\begin{corollary}\label{cor:reify-iv}
For every instance $\inst$ and every rule set $\ruleset$,
$$
\pn{\infty}.
$$
\end{corollary}
\begin{proof}
\begin{align*}
\pnl{\infty} =&\; \reify{\bigcup_{i\in\mathbb{N}} \ksat{i}{\inst, \ruleset}}\\
            =&\; \bigcup_{i\in\mathbb{N}} \pnl{i}\\
            \equiv&\; \bigcup_{i\in\mathbb{N}} \pnr{i}\\
            =&\; \pnr{\infty}
\end{align*}
\end{proof}

%% file: appendices/reification/reify-ii.tex
\subsection{If $\ruleset$ is $\bts$, then $\fcn{reify}(\ruleset)$ is $\bts$ and $\fcs$.}

\begin{lemma}\label{lem:reify-ii}
    If $\ruleset$ is $\bts$, then $\reify{\ruleset}$ is $\bts$ and $\fcs$.
\end{lemma}

\begin{proof}
    Let $\ruleset$ be a $\bts$ rule set. We will show that $\reify{\ruleset}$ is also $\bts$. To this end, we let $\db$ be an arbitrary database over $\Sigma^{\mathrm{rf}}$ and show in the following that $(\db,\reify{\ruleset})$ has a universal model of finite treewidth. 


    From $\db$, we obtain the database $\db' = \fcn{dereify}(\db)$ over the original signature $\Sigma$, defined to consist of
    \begin{itemize}
        \item every atom $\rpred(\vt)$ from $\db$ with $\arity{\rpred}\leq 2$ and
        \item each atom $\rpred(t_{1}, \ldots, t_{\arity{\rpred}})$ with ${\arity{\rpred}}\geq 3$ for which there exists some $t\in \adom{\db}$ such that $\{ \rpred_{i}(t, t_{i}) \mid 1 \le i \le \arity{\rpred}\} \subseteq \db$.
    \end{itemize}
    Intuitively, $\fcn{dereify}$ constitutes the ``best effort'' of reconstructing the original database from some ``imperfectly reified'' one.    


    Since $\ruleset$ is $\bts$ there exists a universal model $\unimod$ of $(\db', \ruleset)$ that has finite treewidth $n$. Let $T = (V,E)$ be a tree decomposition of $\unimod$ of width $n$.

    Next, we will show that $\reify{\unimod}\cup \db$ is a model of $(\db, \reify{\ruleset})$.
    As $\db \subseteq \reify{\unimod}\cup \db$, it remains to be shown that $\reify{\ruleset}$ is satisfied in $\reify{\unimod}\cup \db$. 
    Let $\rho \in \reify{\ruleset}$ be a rule of the form 
    \[ \rho = \forall\vx\vy\phi(\vx, \vy) \rightarrow \exists\vz \psi(\vy, \vz)\] 
    and let $\rho_{\mathrm{urf}}\in\ruleset$ be the rule for which $\rho = \reify{\rho_{\mathrm{urf}}}$. We will denote by $\vu$ the variables from $\vx$ that are introduced by the reification.
    Last, let $h$ be a homomorphism mapping $\body(\rho)$ to $\reify{\unimod}\cup \db$. We will use $v$ and $v_{i}$ to denote variables from $\vx\vy$. 
    
    Note that if for $\rpred\in\Sigma_{\ge 3}$ and $1\le i \le \arity{\rpred}$ an atom of the form $\rpred_{i}(u, v{})$ from $\body(\rho)$ is mapped via $h$ to $\reify{\unimod}\cup \db$, then 
    one of the two holds:
    \begin{itemize}
        \item If $h(u)$ is a constant, then $h(v)$ is a constant. We may conclude this from the following argument: Since $\reify{\unimod}$ has appearances of atoms of the form $\ppred_{k}(u_{\beta}, t)$ for $u_{\beta}$ being a null, $h$ can only map $\rpred_{i}(u, v{})$ to $\db$. But then, since $\db$ is a database and contains only constants as terms, $h(v{})$ has to be a constant as well.
        \item If $h(u)$ is not a constant, then no restrictions are imposed on $h(v{})$. More precisely, $h$ would map $\rpred_{i}(u, v{})$ into $\reify{\unimod}$.
    \end{itemize}
    
    We differentiate several cases in which $h$ maps the atoms of $\body(\rho)$ to $\reify{\unimod}\cup \db$:
    \begin{enumerate}
        \item\label{item:reify-ii-1} Suppose $h$ maps the atom $\alpha$ from $\body(\rho)$ into $\reify{\unimod}\setminus \db$. If  $\alpha$ is an atom over $\Sigma_{\le 2}$, then $h$ will map it not only to $\reify{\unimod}$, but also to $\unimod$ by the definition of reification. Therefore, we suppose that $\alpha$ is an atom over $\Sigma_{\ge 3}$. Then, it is part of the conjunction in $\body(\rho)$ consisting of the ``complete set'' of atoms of the form $\rpred_{i}(u, v{}_i)$ (with $1\le i \le \arity{\rpred}$) which are mapped to $\rpred_{i}(h(u), h(v{}_{i}))\in \reify{\unimod}$ such that there is an atom $\beta = \rpred(h(v{}_{1}), \ldots, h(v{}_{\arity{\rpred}}))\in \unimod$ and $h(u)=u_{\beta}$.

        \item\label{item:reify-ii-2} Suppose $h$ maps the atom $\alpha$ from $\body(\rho)$ into $\db$. Hence, $\alpha$ is of the form $\rpred(\vvv)$, where either (1) $\rpred\in\Sigma_{\le 2}$ or (2) $\rpred\in\Sigma^{\mathrm{rf}}$. Since $\db$ restricted to the signature $\Sigma_{\le 2}$ is a sub-instance of $\db'$, $\alpha$ is mapped via $h$ to $\db'$ in case (1). In case (2), there exists an atom $\beta = \rpred(h(v{}_{1}), \ldots, h(v{}_{\arity{\rpred}}))$ in $\body(\rho_{\mathrm{urf}})$ such that $\alpha = \rpred_{i}(u_{\beta}, v{}_{i}) \in \reify{\beta}$. As $\rho = \reify{\rho_{\mathrm{urf}}}$, the conjunction 
        \[ \reify{\beta} = \bigwedge_{1\le i \le \arity{\rpred}}\rpred_{i}(u_{\beta}, v{}_{i})\] 
        appears in $\body(\rho)$. Furthermore, since $h(u_{\beta})$ is a constant appearing in $\db$, and the whole of $\body(\rho)$ is mapped into $\db$, we know that each $h(v{}_{i})$ is a constant, and hence, every $\rpred_{i}(u_{\beta}, v{}_{i})$ is mapped into $\db$. Consequently, $\beta$ is mapped via $h$ to an atom of $\db'$, and in the case of $\alpha$ being an atom over $\Sigma^{\mathrm{rf}}$, it is also part of a conjunction $\bigwedge_{1\le i \le \arity{\rpred}}\rpred_{i}(u_{\beta}, v{}_{i})$, which is the same as $\reify{\beta}$, and which constitutes a ``full set'' of $\Sigma^{\mathrm{rf}}$ atoms holding that are mapped atom by atom to $\db$. Also, note that by construction, $\beta \in \db'$. 

\end{enumerate}
    
    Putting all of this together, we obtain the following: Every atom over $\Sigma_{\le 2}$ is mapped by $h$ either into $\db$ or into $\reify{\unimod}$. In the former case, each atom over $\Sigma_{\le 2}$ is also mapped to $\db'$, as we are not dealing with predicates introduced by reification. In the latter case, by the definition of reification, it follows that each atom is actually mapped to $\unimod$.
    Atoms over $\Sigma^{\mathrm{rf}}$ (meaning that they are of the form $\rpred_{i}(u, v_{i})$ for $\rpred\in\Sigma_{\ge 3}$) are dealt with in two cases: The first is when $h(u)$ is a constant. Then, they are mapped into $\db$ (meaning that $h(v{}_{i})$ is a constant), so by \cref{item:reify-ii-2}, we know that the whole conjunction $\bigwedge_{1\le i \le \arity{\rpred}}\rpred_{i}(u, v{}_{i})$ of $\rho$ is mapped to $\db$. Furthermore, by the same token, we know that $\rpred(h(v{}_{1}), \ldots, h(v{}_{\arity{\rpred}})) \in \db'$.
    For the second case (meaning that $h(u)$ is a null) note that first it is mapped to $\reify{\unimod}$ and by \cref{item:reify-ii-1} we obtain an analogous result where we know that the whole conjunction $\bigwedge_{1\le i \le \arity{\rpred}}\rpred_{i}(u, v{}_{i})$ of $\rho$ is mapped to $\reify{\unimod}$. Furthermore, by the same token, we know that $\rpred(h(v{}_{1}), \ldots, h(v{}_{\arity{\rpred}})) \in \unimod$.
    This implies that the body of $\rho_{\mathrm{urf}}$ is mapped into $\unimod$ (note that $\db'\subseteq \unimod$ as $\unimod$ is a model of $\db'$) by $h$. 
    Since $\unimod$ is a model of $(\db', \ruleset)$, $\head(\rho_{\mathrm{urf}})$ is mapped into $\unimod$ by a homomorphism $h_{\mathrm{head}}$ extending $h$ mapping $\vz$ to $\vt$, componentwise, for a tuple of terms $\vt$ from $\unimod$.
    Then, by reifying again, we obtain by \cref{fact:reification-homm} that $\reify{h_{\mathrm{head}}}$ is a homomorphism that maps $\reify{\head(\rho_{\mathrm{urf}})}$ to $\reify{\unimod}$. Note that, by the definition of reification for rules, $\reify{\head(\rho_{\mathrm{urf}})} = \head(\rho)$. 
    Consequently, $\rho$ is satisfied in $\reify{\unimod}\cup\db$.
    
    Additionally, as $\db$ is finite and since $\reify{T}$ is a tree decomposition of $\reify{\unimod}$ of finite width (see \cref{obs:reify-tree-ftw}), we can conclude that $\reify{\unimod}\cup \db$ has finite treewidth (a witnessing tree decomposition can be obtained from $\reify{T}$ by adding every term from $\db$ to every node).
    
    To prove universality of $\reify{\unimod}\cup \db$, we consider an arbitrary instance $\inst$ over $\Sigma^\mathrm{rf}$ satisfying $(\db, \reify{\ruleset})$ and show the existence of a homomorphism $h^\mathrm{univ}$ from $\reify{\unimod}\cup \db$ into $\inst$. 
    Let $\inst' = \fcn{dereify}(\inst)$, that is, $\inst'$ is obtained from $\inst$ in the same way $\db'$ was obtained from $\db$. 
    
    We next show Claim~${\dag}$: $\inst'$ satisfies $(\db', \ruleset)$. Because $\db\subseteq \inst$, we immediately obtain $\db'\subseteq \inst'$. It remains to show that $\ruleset$ is satisfied in $\inst'$. For this suppose that $\rho \in\ruleset$ of the form \[ \forall\vx\vy \phi(\vx, \vy) \rightarrow \exists\vz\psi(\vy, \vz) \] 
    and $h$ is a homomorphism mapping $\body(\rho)$ to $\inst'$. Similar to before, $v{}$ or $v{}_{i}$ will signify variables from $\vx\vy$. Furthermore, let 
    \[ \reify{\rho} = \forall\vx\vu\vy\,\reify{\phi}(\vx, \vu, \vy) \rightarrow \exists\vz\vu'\,\reify{\psi}(\vy, \vz, \vu'), \] 
    where we explicitly mention the nulls introduced by the reification of $\rho$, namely $\vu$ and $\vu'$. 
    For each atom of $\body(\rho)$, there are two possibilities: Either the atom is over $\Sigma_{\le 2}$, or it is over $\Sigma_{\ge 3}$. 
    $h$ (and any extension of it) already maps atoms from the first case directly into $\inst$. 
    For the second case, let $\alpha$ be such an atom from $\body(\rho)$. 
    We know by the definition of $\inst'$ that any atom of the form $\rpred(t_{1}, \ldots, t_{\arity{\rpred}})\in\inst'$ corresponds to $\{\rpred_{1}(t{}, t_{1}), \ldots, \rpred_{\arity{\rpred}}(t{}, t_{\arity{\rpred}})\}\subseteq \inst$ for some $t{}\in\adom{\inst}$. 
    In this case, we extend $h$ to a homomorphism $h^{\alpha}$ by mapping $u_{\alpha}$ to $t{}$. Since $u_{\alpha}$ appears only in the conjunction $\bigwedge_{1\le i \le \arity{\rpred}}\rpred_{i}(u_{\alpha}, v{}_{i})$ in $\body(\reify{\rho})$, $h^{\alpha}$ is a partial homomorphism from $\body(\reify{\rho})$ into $\inst$.
    In this way we extend $h$ successively to a homomorphism $h^{\ast}$ mapping all of $\body(\reify{\rho})$ into $\inst$. As $\inst$ satisfies $(\db, \reify{\ruleset})$, it also satisfies $\reify{\rho}$. Consequently, there is an extension $h^{\ast}_{\mathrm{head}}$ of $h^{\ast}$ mapping $\head(\reify{\rho})$ to $\inst$.
    Concerning the atoms appearing in $\head(\reify{\rho})$ we again have two cases: Either, an atom $\beta$ is over $\Sigma_{\le 2}$ or it is over $\Sigma_{2}^{\mathrm{rf}}$. In the former case, $h^{\ast}_{\mathrm{head}}$ maps $\beta$ into $\inst'$. 
    Note that in this case $\beta$ is an atom from $\head(\rho)$. 
    In the latter case, observe that $\beta$ is a conjunct of a conjunction of the form $\bigwedge_{1\le i \le \arity{\ppred}}\ppred_{i}(u_{\gamma}, z^{\ast}_{i})$ in $\head(\reify{\rho})$ for an atom $\gamma = \ppred(z^{\ast}_{1}, \ldots, z^{\ast}_{\arity{\ppred}})\in \head(\rho)$. 
    This conjunction is mapped via $h^{\ast}_{\mathrm{head}}$ into $\inst$ and by construction, $\gamma$ is mapped by the same homomorphism $h^{\ast}_{\mathrm{head}}$ into $\inst'$.
    It follows that $\head(\rho)$ is mapped by $h^{\ast}_{\mathrm{head}}$ to $\inst'$. 
    As $h$ was an arbitrary homomorphism mapping $\body(\rho)$ into $\inst'$, we have shown that $\inst'$ satisfies $\rho$. And since $\rho$ is an arbitrary rule from $\ruleset$, $\inst'$ satisfies $\ruleset$, establishing Claim~${\dag}$.

    Because of Claim~${\dag}$ and $\unimod$ being a universal model of $(\db',\ruleset)$, there exists a homomorphism $h'\colon \unimod \to \inst'$, which we extend to a homomorphism $h^\mathrm{univ}\colon \reify{\unimod}\to \inst$ in the following way: Observe that for every atom using $\rpred\in\Sigma_{\ge 3}$ of the form 
    \[ \alpha = \rpred(t_{1}, \ldots, t_{\arity{\rpred}})\in \unimod \] 
    we have 
    \[ \rpred_{1}(u_{\alpha}, t_{1}), \ldots, \rpred_{\arity{\rpred}}(u_{\alpha}, t_{\arity{\rpred}})\in\reify{\unimod} \, . \] 
    Moreover, since $h'$ is a homomorphism, it holds that 
    \[ h'(\alpha) = \rpred(h'(t_{1}), \ldots, h'(t_{\arity{\rpred}}))\in \inst'\, , \] 
    hence there exists a $t{}\in\adom{\inst}$ such that 
    \[ \rpred_{1}(t{}, t_{1}), \ldots, \rpred_{\arity{\rpred}}(t{}, t_{\arity{\rpred}}) \in \inst \, . \] 
    Taking an arbitrary such $t{}$, we let $u_{\alpha}^{h'} = t{}$. Then, $h^\mathrm{univ}$ defined by 
    \[ h^\mathrm{univ}(x) = \begin{cases} h'(x) & \text{if } x\in\adom{\unimod} \\ u_{\alpha}^{h'} & \text{if } x = u_{\alpha} \text{ for some } \alpha \in\unimod \end{cases} \] 
    is a homomorphism.
    We conclude by arguing that $h^\mathrm{univ}$ also serves as a homomorphism from $\reify{\unimod}\cup \db$ into $\inst$. Note that every term in $\adom{\db}$ is already in $\db'$ by construction (albeit, the term is possibly isolated, only occurring in the context of $\top$) of which $\unimod$ is a model (and thus $h$ has the same domain and codomain as $h^\mathrm{univ}$). From this insight, together with the fact that $\db\subseteq \inst$ (due to $\inst$ being a model of $(\db,\reify{\ruleset})$), and the fact that the terms from $\db'$ are constants, the claimed correspondence holds.

    In summary, $\reify{\unimod}\cup\db$ is a universal model of $(\db, \reify{\ruleset})$ with finite treewidth. As $\db$ was assumed to be an arbitrary database over $\Sigma^{\mathrm{rf}}$, we obtain that $\reify{\ruleset}$ is $\bts$, and by \cref{thm:tree-width-implies-clique-width}, this implies that $\reify{\ruleset}$ is also $\fcs$, completing the proof.
\end{proof}

%% file: appendices/reification/reify-iii.tex
\subsection{\large $(\database,\ruleset) \models (\mathfrak{q},\Xi)$ \iffi $(\fcn{reify}(\database), \fcn{reify}(\ruleset)) \models \fcn{reify}((\mathfrak{q},\Xi))$}

\newcommand{\damsoq}{(\mathfrak{q},\Xi)}
\newcommand{\rdamsoq}{(\reify{\mathfrak{q}},\reify{\Xi})}
\newcommand{\daq}{\mathfrak{q}}
\newcommand{\rsdaq}{\rs_{\daq}}
\newcommand{\rdaq}{\reify{\mathfrak{q}}}
\newcommand{\msoq}{\Xi}
\newcommand{\msoqv}{\msoq(\vX, \vx)}
\newcommand{\rmsoq}{\reify{\Xi}}
\newcommand{\msoqq}{\Xi(\vX, \vx)}
\newcommand{\rmsoqq}{\reify{\Xi(\vX, \vx)}}
\newcommand{\kbs}{\db, \rs}
\newcommand{\kbr}{\reify{\db}, \reify{\rs}}
\newcommand{\ch}{\sat{\db, \rs}}
\newcommand{\reipm}[1]{#1^{+}}
\newcommand{\instp}{\reipm{\inst}}

Let us begin by stating the following observation and fact:

\begin{fact}\label{fact:reify-iii-1}
For a given datalog query $\daq$, a database $\db$, and a rule set $\rs$, the following are equivalent:
\begin{itemize}
    \item $(\kbs) \models \daq$.
    \item $\sat{\kbs} \models \daq$.
    \item $\sat{\sat{\kbs}, \rsdaq} \models \goal$.
\end{itemize}
\end{fact}

\begin{fact}\label{fact:reify-iii-2}
For a given datalog query $\daq$, a database $\db$, and a rule set $\rs$, the following are equivalent:
\begin{itemize}
    \item  $\sat{\kbr} \models \rdaq.$ 
    \item  $\sat{\reify\db, \reify\rs \cup \reify{\rs_{\daq}}} \models \goal.$
\end{itemize}
\end{fact}

\newcommand{\siff}[1]{{\stackrel{\textrm{\cref{#1}}\ \ }{\hfill\iff}} &}
\newcommand{\ssiff}[1]{{\stackrel{#1\ \ }{\hfill\iff}} &}

\begin{lemma}\label{lem:reify-iii-lem-1}
For a given datalog query $\daq$, a database $\db$, and a rule set $\rs$, the following holds:
$$(\kbs) \models \daq \iff (\kbr) \models \rdaq$$
\end{lemma}
\begin{proof}
\begin{align*}
                        & \hspace{-8ex} (\kbs) \models \daq \\
                        \iff & \ch \models \daq \\
\siff{fact:reify-iii-1}    \sat{\ch, \rsdaq} \models \goal\\
\ssiff{\goal \in \sig_0}  \reify{\sat{\ch, \rsdaq}} \models \goal\\
\siff{cor:reify-iv}        \sat{\reify{\ch}, \reify{\rsdaq}} \models \goal\\
\siff{cor:reify-iv}        \sat{\sat{\kbr}, \reify{\rsdaq}} \models \goal\\
\siff{fact:reify-iii-2}    \sat{\kbr} \models \rdaq\\
\siff{fact:reify-iii-1}    (\kbr) \models \rdaq
\end{align*}

\vspace{-4ex}
\end{proof}

\begin{definition}\label{def:reify-iii-reipm}
For an instance $\inst$, let us define the following: $$\reipm{\inst} = \set{\top(u_\alpha) \;|\; \alpha \in \inst} \cup \inst.$$

\vspace{-4ex}~\defend
\end{definition}

\begin{fact}\label{fact:reify-iii-3}
For every instance $\inst$, $\adom{\reipm{\inst}} = \adom{\reify{\inst}}$.
\end{fact}

\begin{observation}\label{obs:reify-iii-hom-closed}
For every DaMSOQ $\damsoq$ and every instance $\inst$ we have:
$$\inst \models \damsoq \iff \reipm{\inst} \models \damsoq$$
\end{observation}
\begin{proof}
Recall that every DaMSOQ is closed under homomorphisms.\\
$(\Leftarrow)$ Note that there exists a homomorphism from $\reipm{\inst}$ to $\inst$.\\
$(\Rightarrow)$ Note that $\inst \subset \reipm{\inst}$.
\end{proof}

\begin{lemma}\label{lem:reify-iii-mso-instp}
If $\msoq$ be an MSO formula and $\inst$ be an instance, then
$$\reipm{\inst} \models \msoq \iff \reify{\inst} \models \reify{\msoq}.$$
\end{lemma}
\begin{proof}

We prove the lemma by induction on the structure of $\msoq=\msoqv$, with the following thesis: For every 
assignment $\mu$ that maps variables from $\vx$ to terms from $\adom{\reipm{\inst}}$ and set variables from $\VX$ to subsets of $\adom{\reipm{\inst}}$ holds
$$\reipm{\inst}, \mu \models \msoqv \iff \reify{\inst}, \mu \models \reify{\msoqv}.$$

For the base case, let us assume that $\msoqv$ is atomic. Then, there are three cases to consider:
\begin{itemize}
    \item $\msoqv = \pred{R}(\vx)$ with $\arity{\pred{R}} > 2$. Then, letting $\vx = x_1, \ldots, x_n$, we have $\reify{\msoqv}=\reify{\pred{R}(\vx)} = \exists{y} \bigwedge_{1 \leq i \leq n} \pred{R}_i(y, x_i)$. 
    Now, we note that $\alpha = \pred{R}(\mu(x_1), \ldots, \mu(x_n)) \in \reipm{\inst}$ holds by the definition of $\reify{\cdot}$ \iffi $\reify{\alpha} = \bigwedge_{1 \leq i \leq n} \pred{R}_i(u_\alpha, \mu(x_i)) \in \reify{\inst}$.
    \item $\msoqv = \pred{R}(\vx)$, where $\arity{\pred{R}} \leq 2$. This case trivially follows since $\reify{\pred{R}(\vx)} = \pred{R}(\vx)$ and $\reify{\inst}$ coincides with $\reipm{\inst}$ on atoms of low arity.
    \item $\msoqv = X(x)$, where $X$ is a second-order variable. This case follows as $\reify{X(x)} = X(x)$.
\end{itemize}

The cases for the Boolean connectives are trivial. Therefore, we prove one of the quantifier cases and note that all other cases are shown in a similar fashion.
Assume $\msoqv = \forall{y}\; \msoq'(\vX, \vx, y)$. 
            By induction hypothesis, we know that for every 
            $t' \in \adom{\reipm{\inst}}$, 
            we have:
            \begin{align*}
            & \instp,\mu\cup\{y \mapsto t'\}\models \msoq'(\vX, \vx, y) \iff\\[-0.8ex] 
            & \ \ \ \ \reify{\inst},\mu\cup\{y \mapsto t'\} \models \reify{\msoq'(\vX, \vx, y)}.
            \end{align*}
            From this and \cref{fact:reify-iii-3}  we may conclude that the thesis holds for $\msoqv = \forall{y}\; \msoq'(\vX, \vx, y)$.
%
\end{proof}

Finally, combining \cref{cor:reify-iv}, \cref{lem:reify-iii-lem-1}, \cref{obs:reify-iii-hom-closed}, and \cref{lem:reify-iii-mso-instp}, we obtain the following corollary:

\begin{corollary}\label{cor:reify-iii}
For every DaMSOQ $\damsoq$, every database $\db$, and every rule set $\rs$, we have:
$$(\kbs) \models \damsoq \!\iff\! (\kbr) \models \rdamsoq.$$
\end{corollary}
\begin{proof}
From \cref{lem:reify-iii-lem-1} we know that:
$$(\kbs) \models \daq \iff (\kbr) \models \rdaq.$$
Moreover, by \cref{obs:reify-iii-hom-closed} and \cref{lem:reify-iii-mso-instp}, we obtain:
\begin{align*}
    &\hspace{-8ex} (\kbs) \models \msoq\\        
\ssiff{\text{$\msoq$ is hom. closed}} \sat{\db, \rs} \models \msoq\\
\siff{obs:reify-iii-hom-closed} \reipm{(\sat{\db, \rs})} \models \msoq\\
\siff{lem:reify-iii-mso-instp}  \reify{\sat{\db, \rs}} \models \rmsoq\\
\siff{cor:reify-iv}             (\kbr) \models \rmsoq
\end{align*}

\vspace{-4ex}
\end{proof}

%% file: appendices/appendix51.tex
\section{Proofs for \NoCaseChange{\cref{sec:binary}}}\label{appendix:binary}




\begin{observation}{\label{obs:sate-is-forest}}
$\se$ is a typed polyforest of bounded degree.
\end{observation}

\begin{proof}
Since $\ruleset$ is a set of single-headed rules over a binary signature, and since we only consider non-datalog rules (as $\se$ only contains atoms derived by the non-datalog rules from $\ruleset$), we know that each rule $\rho$ applied during the Skolem chase (see~\cref{sec:preliminaries}) is of the form $\phi(\vec{x},y) \rightarrow head(\rho)$ with $head(\rho) \in \{\exists z \rpred(y,z), \exists z \rpred(z,y),\\ \exists z \exists z' \rpred(z,z')\}$. That is to say, for each rule $\rho$, $\fr(\rho) = 1$ or $\fr(\rho) = 0$. Observe that if a rule of frontier $1$ is applied during the Skolem chase, then it gives rise to a single edge (typed with the predicate $R$ occurring in the head of the rule) in 
 $\se $ that can either be directed toward (if $\head(\rho) = \exists z \rpred(y,z)$) or away from (if $\head(\rho) = \exists \vec{z} \rpred(z,y)$) the frontier term it is attached to. If a rule of frontier $0$ is applied during the Skolem chase, then it creates a disconnected typed edge that corresponds to, and gives rise to, a new typed polytree in 
 $\se$. Hence, we can see that 
 $\se$ is a typed polyforest. Moreover, as we are using the Skolem chase and due to the fact that $\ruleset$ is finite, there can only be a finite number of typed edges added to each vertex in the typed polyforest, ensuring that it is of bounded degree.
\end{proof}

\newcommand{\sed}{\database \cup \se}
\begin{customlem}{\ref{cor:skeleton-plus-has-bounded-cw}}
$\ksat{1}{\sed, \rulesetnormcon}$ has finite treewidth.
\end{customlem}
\begin{proof}
Given an instance $\inst$ and two its terms $t, t'$ define the {\em distance} $\ldist{\inst}{t,t'}$ from $t$ to $t'$ in $\inst$ as the minimal length of any path from $t$ to $t'$ treating binary atoms of $\inst$ as undirected edges.

For a null $t$ of $\se$ and a natural number $n$ let {\em neighbourhood} $\neigh{n}{t}$ of $t$ of distance $n$ be a following set of terms $\set{t \in \termset(\sed) \;|\; \ldist{\sed}{t, t'} \leq n}$. Also let $\neigh{n}{\db} = \set{t \;|\; \exists{\const{c}\in\adom{\db}}\\ \ldist{\sed}{t,\const{c}} < n}$. Observe that every $t^n$ (and $\neigh{n}{\db}$) is finite as $\sed$ have bounded degree - a direct consequence of Skolem chase usage.

Let $T_n = (V_n, E_n)$ be a following tree decomposition:
\begin{itemize}
    \item $V_n = \set{\neigh{n}{t} \;|\; t \in \nullset(\se) } \cup \set{\neigh{n}{\db}}$.
    \item $E_n = \set{(\neigh{n}{t}, \neigh{n}{t'}) \;|\;  \wedge\exists{\rpred\in\sig}\; \se \models \rpred(t,t') \vee \rpred(t', t)} \;\cup\;\\ \set{(\neigh{n}{\db}, \neigh{n}{t}) \;|\;  \exists{\const{c} \in \conset(\db)}\; \exists{\rpred\in\sig} \se \models \rpred(t,\const{c}) \vee \rpred(\const{c}, t)}$.
\end{itemize}

Now we will argue that for every $n > 1$ tree $T_n$ is a tree decomposition of $\sed$ as:
\begin{itemize}
    \item $\bigcup_{v\in V_n}v$ is equal to $\termset(\sed)$.
    \item Each atom $\rpred(t,t')$ of $\sed$ is indeed contained in one of bags of $V_n$, namely $\neigh{n}{t}$, as $n > 1$.
    \item From the definition of neighbourhood, one can note that for each term $t$ in $\sed$, the subgraph of $T_n$ induced by the nodes $v$ of $V$ such that $t \in v$ is connected. 
\end{itemize}


Let $d \in \mathbb{N}$ be such that for every rule $\rho \in \rulesetnormcon$, with two frontier variables $x,y$, the following holds $d \geq \ldist{\body(\rho)}{x,y}$. Note that such a $d$ exists because $x$ and $y$ share the same connected component of $\body(\rho)$. Finally, observe that for every atom $\rpred(t,t')$ of $\ksat{1}{\sed, \rulesetnormcon} \setminus (\sed)$ we have that $t,t' \in \neigh{d}{t}$ or $t,t' \in \neigh{d}{\db}$. From this we conclude that $T_d$ is a tree decomposition of $\ksat{1}{\sed, \rulesetnormcon}$ of bounded width.

\end{proof}

\subsection*{\large Recoloring lemma}


\begin{customlem}{\ref{cor:recoloring}}
Let $\inst$ be an 
instance 
satisfying $\cw{\inst} = n$ and let $\coloring' : \adom{\inst} \to \cols'$ be an arbitrary coloring of~$\inst$. Then $\cw{\inst,\coloring'} \leq (n+1) \cdot |\cols'|.$
\end{customlem}

First, consider the following:

\begin{lemma}
\label{obs:cross-coloring}
 Let $\inst$ be a countable instance over some signature $\sig$. Moreover, let $\coloring : \adom{\inst} \to \cols$ and $\coloring' : \adom{\inst} \to \cols'$ be two colorings and let $\coloring {\times} \coloring'$ denote the coloring
 of $\inst$ defined by $\coloring {\times} \coloring'(t) = (\coloring(t), \coloring'(t))$. If $\cw{\inst, \coloring} \leq |\cols|$, then $\cw{\inst, \coloring {\times} \coloring'} \leq |\cols| \cdot |\cols'|$.
\end{lemma}


\begin{proof}
  Let $(\inst, \coloring)$ be a colored instance and $\coloring'\colon\adom{\inst}\to \cols'$ another coloring for $\inst$. Furthermore, let $\cw{\inst, \lambda} = k$. Then there exists for a finite set of colors $\cols$ a well $(\cols, \mathrm{Cnst}, \Sigma)$-decorated tree $\mathcal{T}$ such that $\inst^{\mathcal{T}}=\inst$ and $\fcn{col}^\mathcal{T} = \coloring$.
  We obtain a $(\cols\times\cols', \mathrm{Cnst}, \Sigma)$-deco\-rated tree $\mathcal{T}'$ inductively over the nodes of $\mathcal{T}$ and $\mathcal{T}'$ in the following way:
  
  \textit{Transformations.} Let $s, s'\in \{ 0,1 \}^{\ast}$ and $N_{p} \subseteq (\{ 0,1 \}^{\ast})^{2}$. Furthermore, let $\sigma \colon \{ 1, \ldots, |\cols'| \} \to \cols'$ be a bijection enumerating every color in $\cols'$.
    \begin{enumerate}
        \item\label{item: Transformation1} If $\oplus(s)\in \mathcal{T}$, then $\oplus(s') \in \mathcal{T}^{s}$ and add to $N_{p}$ the pairs $(s0, s'0)$ and $(s1, s'1)$.
        \item\label{item: Transformation2} If for $\const{c}\in \mathrm{Cnst} \cup \{ \const{*} \}$ and $\ell\in\cols$, $\pred{c}_{\ell}(s)\in \mathcal{T}$, then $\pred{c}_{(\ell, \lambda'(s))}(s')\in \mathcal{T}^{s}$ and add to $N_{p}$ the pairs $(s0, s'0)$ and $(s1, s'1)$.
        \item\label{item: Transformation3} If $\pred{Void}(s)\in \mathcal{T}_{s}$, then add $\pred{Void}(s')$ to $\mathcal{T}'$ and add to $N_{p}$ the pairs $(s0, s'0)$ and $(s1, s'1)$.
        \item\label{item: Transformation4} If for $\rpred\in\Sigma_1$ and $\ell\in\cols$ $\pred{Add}_{\rpred, \ell}(s)\in\mathcal{T}_{s}$, then add 
        \[ \{ \pred{Add}_{\rpred, (\ell, \ell')}(0^{i-1}) \, | \,  1\le j \le |\cols'|, \, \sigma(j) = \ell' \} \]to $\mathcal{T}^{s}$ and add to $N_{p}$ for $1\le i\le k'$ the pairs $(s0, s'0^{i})$ and $(s1, s'0^{i-1}1)$.
        \item\label{item: Transformation5} If for $\rpred\in\Sigma_2$ and $\ell_1, \ell_2\in\cols$ $\pred{Add}_{\rpred, \ell_1, \ell_2}(s)\in\mathcal{T}_s$, then add 
        \begin{align*}
            \{ \pred{Add}_{\rpred, (\ell_1, \ell_1'), (\ell_2, \ell_2')}&(0^{i-1 + (j-1)\cdot |\cols'|}) \, | \\
            & 1\le i, j \le |\cols'|, \, \sigma(i) = \ell_1' , \, \sigma(j) = \ell_2'  \}
        \end{align*}
        to $\mathcal{T}^{s}$ and add to $N_{p}$ for $1\le i, j\le k'$ the pairs $(s0, s'0^{i + (j-1)\cdot k'})$ and $(s1, s'0^{i-1 + (j-1)\cdot k'}1)$.
        \item\label{item: Transformation6} If for $\ell_1, \ell_2\in\cols$ $\pred{Recolor}_{\ell_1 \to \ell_2}(s)\in\mathcal{T}_{s}$, then add 
        \begin{align*}
            \{ \pred{Recolor}_{(\ell_1, \ell_1') \to (\ell_2, \ell_2')}&(0^{i-1 + (j-1)\cdot |\cols'|}) \, |  \\
            & 1\le i, j \le |\cols'|, \, \sigma(i) = \ell_1' , \, \sigma(j) = \ell_2' \}
        \end{align*}
        to $\mathcal{T}^{s}$ and add to $N_{p}$ for $1\le i, j\le k'$ the pairs $(s0, s'0^{i + (j-1)\cdot k'})$ and $(s1, s'0^{i-1 + (j-1)\cdot k'}1)$.
    \end{enumerate}
  
  \textit{Base case.} Initialize $N_{p} = \{ (\varepsilon, \varepsilon) \}$. Then apply the appropriate transformation for the tuple $(\varepsilon, \varepsilon)$.
    
  \textit{Induction Step.} Suppose $N_p$ has been obtained after a finite amount of transformations. With $\prec$ we denote the usual prefix order on $\{ 0,1 \}^{\ast}$. We define a partial order on $N_p$ in the following way: \[ (s_1, s_1') < (s_2, s_2') \quad \text{\ifandonlyif} \quad s_1 \prec s_2 \text{ or } s_1 = s_2 \land s_1' \prec s_2' \, . \] Take the tuple $(s,s')\in N_p$ such that $(s,s')$ is a $|s'|$-minimal and $<$-maximal element in $N_p$. Then apply the appropriate transformation to $(s,s')$.
  
  Finally, obtain $\mathcal{T}' = \mathcal{T}_{\mathrm{bin}}\cup \bigcup_{s\in\ \{0,1\}^{\ast}}\mathcal{T}^{s}$.
  
  We now show some properties of the set $N_p$.
  \begin{enumerate}
    \item\label{item: N_property1} If $(s, s')\in N_p$ with $\const{c}\in \mathrm{Cnst} \cup \{ \const{*} \}$ and $\pred{c}_{\ell}(s)\in\mathcal{T}$ for some $\ell\in \cols$, then $\pred{c}_{(\ell, \lambda'(s))}(s')\in \mathcal{T}'$ and $s'$ is the unique $s''\in \{ 0, 1\}^{\ast}$ such that $(s,s'')\in N_p$.
        \begin{subproof}
            Let $s\in \{ 0,1 \}^{\ast}$ and for some $\ell\in \cols$ and $\const{c}\in \mathrm{Cnst}$ it holds that $\const{c}_{\ell}(s)\in \mathcal{T}$. Furthermore, suppose that $s', s''\in \{ 0,1 \}^{\ast}$ such that $(s, s'), (s,s'')\in N_p$ with $\pred{c}_{(\ell, \lambda'(s))}(s'), \pred{c}_{(\ell, \lambda'(s))}(s'')\in \mathcal{T}^{s}$. This can only happen, if there was previously a transformation applied for $r, s''' \in \{ 0,1 \}^{\ast}$ to $(r, s''')\in N_p$ with $\pred{Succ}_0(r, s)\in \mathcal{T}$ and $s''' \prec s', s''$ such that either
            \begin{enumerate}
                \item\label{item: N_property1_a} for $\rpred\in\Sigma_1$ and $k\in\cols$ $\pred{Add}_{\rpred, k}(r)\in\mathcal{T}$, or
                \item\label{item: N_property1_b} for $\rpred\in\Sigma_2$ and $k_1, k_2\in\cols$ $\pred{Add}_{\rpred, k_1, k_2}(r)\in\mathcal{T}$, or
                \item\label{item: N_property1_c} for $k_1, k_2\in\cols$ $\pred{Recolor}_{k_1 \to k_2}(r)\in\mathcal{T}$.
            \end{enumerate}
            But then w.l.o.g. (in case \ref{item: N_property1_a}) $s' = s'''0^{i}$ and $s'' = s'''0^{j}$ with $1\le i < j \le |\cols'|$ or (in cases \ref{item: N_property1_b} and \ref{item: N_property1_c}) $s' = s'''0^{i}$ and $s'' = s'''0^{j}$ with $1\le i < j \le |\cols'|^{2}$. But then, by the Inductions Step, Transformation \ref{item: Transformation2} is only applied to the tuple $(s, s'''0^{|\cols'|})$ (in case \ref{item: N_property1_a}) or $(s, s'''0^{|\cols'|^{2}})$ (in cases \ref{item: N_property1_b} and \ref{item: N_property1_c}) as this is $<$-maximal for the introduced $N_p$-pairs. Hence $s'''0^{|\cols'|^{2}} = s' = s''$, showing uniqueness. Furthermore, by Transformation \ref{item: Transformation2} we obtain $\pred{c}_{(\ell, \lambda'(s))}(s'), \pred{c}_{(\ell, \lambda'(s))}(s'')\in \mathcal{T}^{s}$.
            \end{subproof}
    \item\label{item: N_property2} If $(s, s')\in N_p$ with $\pred{Void}(s)\in\mathcal{T}$, then $\pred{Void}(s')\in \mathcal{T}'$ and for all $s''\in\{ 0,1 \}^{\ast}$ with $s' \prec s''$ we have $\pred{Void}(s'')$.
        \begin{subproof}
            If $(s, s')\in N_p$, application of Transformation \ref{item: Transformation3} yields $\pred{Void}(s')\in \mathcal{T}'$. As then $(s0, s'0), (s1, s'1)$ are added to $N_p$, applications of Transformation \ref{item: Transformation3} to both will yield $\pred{Void}(s'0)\in \mathcal{T}'$ and $\pred{Void}(s'1)\in \mathcal{T}'$. Inductively we obtain for all $s''\in\{ 0,1 \}^{\ast}$ with $s' \prec s''$ we have $\pred{Void}(s'')$.
        \end{subproof}
  \end{enumerate}
  
  $\mathcal{T}'$ is a well $(\cols\times\cols', \mathrm{Cnst}, \Sigma)$-decorated tree.
    \begin{itemize}
        \item For every $s\in \{ 0,1 \}^{\ast}$, $\mathrm{Dec}(s)\in\mathcal{T}'(s)$ for exactly one $\mathrm{Dec} \in \mathrm{Dec}(\cols', \mathrm{Cnst}, \Sigma)$ by constructions.
        \item Since, by construction, if for $\const{c}\in \mathrm{Cnst}$, $(\ell_1, \ell_2)\in \cols\times\cols'$ and $s'\in\{ 0,1 \}^{\ast}$ it holds that $\pred{c}_{(\ell_1, \ell_2)}(s')\in\mathcal{T}'$, then there is a $s\in\{ 0,1 \}^{\ast}$ such that $\pred{c}_{\ell_1}(s')\in\mathcal{T}$ and $(s, s')\in N_p$. By Property \ref{item: N_property1} $s'$ is unique for $s$. On the other hand, Transformation \ref{item: Transformation2} is only executed, if there is a pair $(s, s')\in N_p$ with $\pred{c}_{\ell_1}(s')\in\mathcal{T}$. In conclusion, for every $s'\in \{ 0,1 \}^{\ast}$, $\const{c}\in \mathrm{Cnst}$, and $\ell\in\cols\times\cols'$ there is at most one fact of the form $\pred{c}_{\ell}(s')\in\mathcal{T}'$.
        \item If $\pred{Void}(s')\in\mathcal{T}'$ or $\pred{c}_{(\ell_1, \ell_2)}(s')\in \mathcal{T}'$, then $\pred{Void}(s'0), \pred{Void}(s'1)$. In the case of $\pred{Void}(s')\in \mathcal{T}'$ this follows immediately from Property \ref{item: N_property2}. If on the other hand $\pred{c}_{(\ell_1, \ell_2)}(s')\in \mathcal{T}'$ we observe that there exists by construction exactly one $s\in\{0,1 \}^{\ast}$ such that $\pred{c}_{\ell_1}(s)\in\mathcal{T}$ and $(s, s')\in N_p$. Note that $\pred{c}_{(\ell_1, \ell_2)}(s')$ was then produced by application of Transformation \ref{item: Transformation2} and $(s0, s'0), (s1, s1')$ will have been added to $N_p$. Application of Transformation \ref{item: Transformation3} then has yielded $\pred{Void}(s'0), \pred{Void}(s'1)\in \mathcal{T}'$ as $\pred{Void}(s0), \pred{Void}(s1)\in\mathcal{T}$.
        \item If $\oplus(s')\in\mathcal{T}$ then $\pred{Void}(s'0), \pred{Void}(s'1)\not\in \mathcal{T}'$ follows by construction.
        \item If for $\ell, \ell'\in\cols\times\cols'$ we have $\pred{Recolor}_{\ell\to \ell'}(s') \in \mathcal{T}'$ or $\pred{Add}_{\rpred,\ell,\ell'}(s') \in \mathcal{T}'$ or $\pred{Add}_{\rpred,\ell}(s') \in \mathcal{T}'$ then by definition of Transformations \ref{item: Transformation4} - \ref{item: Transformation6} $\pred{Void}(s'0)\not\in \mathcal{T}'$ and $\pred{Void}(s'1) \in \mathcal{T}'$ hold by construction.
    \end{itemize}
  
  We now show that $\mathcal{I}^{\mathcal{T}'}$ and $\mathcal{I}^{\mathcal{T}}$ are isomorphic and $\fcn{col}_{\epsilon}^{\mathcal{T}'} = \coloring \times \coloring'$.
  Obtaining the isomorphism works as follows: By Property \ref{item: N_property1} for every $(s,s')\in N_p$ such that for some $\const{c}\in \mathrm{Cnst}\cup\{ \const{*} \}$ and $\ell\in \cols$ it holds that $\pred{c}_{\ell}(s)\in\mathcal{T}$, $s'$ is unique. Furthermore, by construction we obtain \[ \mathrm{Cnst}\cap ent^{\mathcal{T}}(\varepsilon) = \mathrm{Cnst}\cap ent^{\mathcal{T}'}(\varepsilon) \, . \] Hence we define $h_{\mathrm{Cnst}}\colon \mathrm{Cnst}\cap \fcn{ent}^{\mathcal{T}}(\varepsilon) \to \mathrm{Cnst} \cap \fcn{ent}^{\mathcal{T}'}(\varepsilon), \const{c}\mapsto \const{c}$, which is a bijection by the discussion beforehand.
  We extend $h_{\mathrm{Cnst}}$ to $h$ by mapping each $s\in\{ 0,1 \}^{\ast}$ with $\pred{*}_{k}(s)\in\mathcal{T}$ to the unique $s'\in\{ 0,1 \}^{\ast}$ such that $(s, s')\in N_p$. This yields a bijection $h\colon \fcn{ent}^{\mathcal{T}}(\varepsilon) \to \fcn{ent}^{\mathcal{T}'}(\varepsilon)$. We show that $h$ is an isomorphism $\inst^{\mathcal{T}} \to \inst^{\mathcal{T}'}$. We need only to show that atoms are preserved in both directions.
  We show this in the case of binary atoms, the rest will be analogously dealt with. So let $\rpred(t_0, t_1)\in\inst^{\mathcal{T}}$ be an atom. Then there exist a $s\in\{ 0,1 \}^{\ast}$ such that $R(t_0, t_1)\in Atoms_{s}^{\mathcal{T}}$, and $s_0, s_1\in\{ 0, 1\}^{\ast}$ such that the nodes $ss_0 = t_0$ and $ss_1 = t_1$. Then there are $s', s_0', s_1'\in\{ 0,1 \}^{\ast}$, with the later two unique, such that
  \begin{itemize}
      \item $(s, s'), (ss_0, s's_0'), (ss_1, s's_1') \in N_p$,
      \item $\pred{Add}_{\rpred, k_0, k_1}(s)\in\mathcal{T}$,
      \item $\fcn{col}_s^{\mathcal{T}}(t_0) = k_0$ and $\fcn{col}_s^{\mathcal{T}}(t_1) = k_1$,
      \item $\fcn{col}_{ss_0}^{\mathcal{T}}(t_0) = \ell_0$ and $\fcn{col}_{ss_1}^{\mathcal{T}}(t_1) = \ell_1$,
      \item $\fcn{col}_{s's_0'}^{\mathcal{T}'}(h(t_0)) = (\ell_0, \coloring'(ss_0))$ and $\fcn{col}_{s's_1'}^{\mathcal{T}}(h(t_1)) = (\ell_1, \coloring'(ss_1))$,
      \item $\pred{Add}_{\rpred, (k_0, \coloring'(ss_0)) , (k_1, \coloring'(ss_1))}(s')\in\mathcal{T}'$, and
      \item $\fcn{col}_s^{\mathcal{T}'}(h(t_0)) = (k_0, \coloring'(ss_0))$ and $\fcn{col}_s^{\mathcal{T}'}(h(t_1)) = (k_1, \coloring'(ss_1))$.
  \end{itemize}
  Hence $\rpred(h(t_0), h(t_1))\in\inst^{\mathcal{T}'}$. The backwards-direction works similarly. For uniqueness note that by construction for every pair $(s,s')\in N_p$ the $s$ will be unique for the $s'$.
  
  Last, $\fcn{col}_{\epsilon}^{\mathcal{T}'} = \coloring \times \coloring'$ is plain as every entity will be introduced with the color given by $\coloring'$ in the second component in $\mathcal{T}'$.
\end{proof}
To get \cref{cor:recoloring} from \cref{obs:cross-coloring} we note that it is enough to project colors of $(\inst, \coloring \times \coloring')$ onto to their second components with $\fcn{Recolor}_{(\ell,\ell') \to \ell'}$. Note that the resulting coloring function will have an extended codomain (comparing to $\coloring \times \coloring'$) to $\cols \times \cols' \cup \cols'$.

%
%




%% file: appendices/multihead-full.tex
\section{Extended Version of \NoCaseChange{\cref{sec:higher-arity}}}\label{appendix:higher-arity}

We will now prove \cref{thm:higher-arity} implying that for multi-headed rules, $\fus$ is no longer subsumed by $\fcs$. To this end, we exhibit a $\fus$ rule set which may unavoidably lead to universal models of infinite cliquewidth. 

We let $\sigi = \set{\hpred, \vpred}$ with $\hpred$ and $\vpred$ being binary predicates. Let us define $\rsgrid$ to be the following rule set over $\sigi$:
$$
\begin{array}{@{}l@{\ \ \ \ \ \ }r@{\ }l@{\ \ \ \ \ \ \ \ }}
\text{(loop)} & \top & \rightarrow \exists x \; (\hpred(x,x) \land \vpred (x,x)) \\[0.5ex]
\text{(grow)} & \top & \rightarrow \exists y \exists y'\; (\hpred(x,y) \land \vpred (x,y')) \\[0.5ex]
\text{(grid)} & \hpred(x,y) \land \vpred(x,x')  & \rightarrow \exists y' \; (\hpred(x',y') \land \vpred(y,y'))
\end{array}
$$

We make use of $\rsgrid$ to establish \cref{thm:higher-arity}, whose proof is split into two parts. First, we provide a database $\dbg$ to form a knowledge base $(\dbg, \rsg)$ for which no universal model of finite cliquewidth exists (\cref{lem:no-universal-model-cw}). Second, we show that $\rsg$ is a finite unification set (\cref{lem:rsg-is-fus}). 

\subsection{$\rsgrid$ is not a finite-cliquewidth set}

We devote this section to establishing the subsequent lemma, which is sufficient to conclude \cref{lem:no-universal-model-cw}.

\begin{lemma}\label{lem:no-universal-model-cw-intermediate}
Let $\const{a}$ be a constant and let $\dbg = \set{\top(\const{a})}$ be a database. Then, no universal model of finite cliquewidth exists for $\kbaseg$.
\end{lemma}

With the aim of establishing the lemma above, we define the instance $\ginf$ below, where $\const{a}$ is a constant and $x_{i,j}$ and $y$ are nulls for $i,j \in \mathbb{N}$. To provide the reader with further intuition, a graphical depiction of $\ginf$ is given in \cref{fig:ginf}.
%
%
\begin{eqnarray*}
\ginf & = & \{\top(\const{a}),\hpred(\const{a},x_{1,0}),\vpred(\const{a},x_{0,1}),\hpred(y,y),\vpred(y,y)\} \\
 & & \cup \big\{\hpred(x_{i,j},x_{i+1,j}) \ | \  (i,j) \in (\mathbb{N} \times \mathbb{N}) \setminus \{(0,0)\} \big\} \\
 & & \cup \big\{\vpred(x_{i,j},x_{i,j+1}) \ | \  (i,j) \in (\mathbb{N} \times \mathbb{N}) \setminus \{(0,0)\} \big\} 
\end{eqnarray*}


It is straightforward to show that $\ginf$ serves as a model of $\kbaseg$, and a routine exercise to prove that $\ginf$ is universal. Hence, we have the following:

\begin{observation}\label{obs:ginf-is-universal-model}
$\ginf$ is a universal model of $\kbaseg$.
\end{observation}





It is known that the cliquewidth of an $n \times n$ grid tends to infinity as $n$ tends to infinity~\cite[\thm~4]{GolumbicRot99}, entailing that any infinite grid will have an infinite cliquewidth since the cliquewidth of any instance bounds the cliquewidth of any induced sub-instance from above. By making use of this fact, along with the subsequent two, we can prove \cref{fact:grid-has-unbounded-cw} below.

\begin{fact}
        Let $\sig$ be a binary signature, $\epred$ be a binary predicate, 
        $\inst$ be an instance over $\sig$, and $\inst' = \set{\epred(t,t') \ | \ \rpred(t,t') \in \inst, \rpred \in \sig}$. Then, $\cw{\inst} \geq \cw{\inst'}$.
\end{fact}

\begin{fact}
        Let $\sig$ be a signature consisting of a single binary predicate $\epred$, let $\inst$ be an instance over $\sig$, and let $\inst' = \set{\epred(t',t) \ | \ \epred(t,t') \in \inst} \cup \inst$. Then, $\cw{\inst} \geq \cw{\inst'}$.
\end{fact}

\begin{observation}\label{fact:grid-has-unbounded-cw}
$\ginf$ has an infinite cliquewidth.
\end{observation}
\begin{proof} The observation follows from the previous two facts, the fact that infinite grids have infinite cliquewidth, and the fact that adding a set of the form $\{\hpred(y, y), \vpred(y, y) \}$ with $\const{a}$ a constant and $y$ a null to any instance does not lower its cliquewidth.
\end{proof}

We now leverage the above result to demonstrate that {\em every} universal model of $\kbaseg$ has an infinite cliquewidth. Observe that if we consider any homomorphism from $\ginf$ to itself, then (i) the null $y$ must be mapped to itself as no other terms are associated with loops, and (ii) the constant $a$ must be mapped to itself, thus anchoring the image of $\ginf$ under the homomorphism at the origin of the grid shown in \cref{fig:ginf}, thus ensuring that each term of the grid will map to itself accordingly. It follows that any homomorphism from $\ginf$ to itself will be the identity function, giving rise to the following observation.

\begin{observation}\label{obs:identity}
If $h: \ginf \to \ginf$ is a homomorphism then $h$ is the identity.
\end{observation}

We now make use of the above observation to show that every universal model of $\kbaseg$ contains a sub-instance that is isomorphic to $\ginf$, thus ensuring that every universal model of $\kbaseg$ is of infinite cliquewidth (by \cref{fact:grid-has-unbounded-cw} and the fact that the cliquewidth of an instance is greater than or equal to the cliquewidth of any induced sub-instance). From this lemma, we may conclude  \cref{lem:no-universal-model-cw-intermediate} and thus \cref{lem:no-universal-model-cw}.

\begin{lemma}
Every universal model of $\kbaseg$ contains an induced sub-instance isomorphic to $\ginf$.
\end{lemma}

\begin{proof}
Let $\unimod$ be a universal model of $\kbaseg$. As $\ginf$ is a universal model (by \cref{obs:ginf-is-universal-model} above), we know that two homomorphisms $h\colon \ginf \to \unimod$ and $h'\colon \unimod \to \ginf$ exist between $\unimod$ and $\ginf$. Let us consider the homomorphism $h' \circ h$. By \cref{obs:identity}, the homomorphism $h'\circ h$ is the identity on $\ginf$, implying that $h$ is injective. Furthermore, $h$ is reflective, and hence, constitutes an embedding. To show this, suppose that $\rpred\in\sigi$ and let $\vt$ be a tuple of length $\arity{\rpred}$ from $\ginf$. If $\rpred(h(\vt))$ holds in $\unimod$ , then $\rpred((h' \circ h)(\vt))$ holds in $\ginf$, but since $h'\circ h$ is the identity on $\ginf$, we immediately obtain that $\rpred(\vt)$  holds in $\ginf$. Finally, note that the image of $\ginf$ under $h$ serves as the induced sub-instance of $\unimod$ that is isomorphic to $\ginf$, completing the proof.
\end{proof}

\subsection{$\rsgrid$ is a finite-unification set}

As we have now shown that $\rsgrid$ is not a finite-cliquewidth set, we focus our attention on demonstrating \cref{lem:rsg-is-fus}, which is restated below, and from which we may conclude the final main result (\cref{thm:higher-arity}).

\begin{customlem}{\ref{lem:rsg-is-fus}}
$\rsg$ is $\fus$.
\end{customlem}

As the rule $(loop)$ is a member of $\rsg$, every BCQ that is free of constants can be trivially re-written, implying the following:

\begin{observation}\label{obs:bcq-no-const-trivial}
Every BCQ that is free of constants is FO-rewritable with respect to $\rsg$.
\end{observation}


Showing that a provided rule set is $\fus$ is a notoriously difficult task.\footnote{This should not be too surprising however, as being $\fus$ is an undecidable property~\cite{BagLecMugSal11}.} To this end, we hand-tailor a rewriting algorithm to be used for the proof of~\cref{lem:rsg-is-fus}. However, one can refer to existing rewriting algorithms, e.g., XRewrite~\cite{GotOrsPie14}, for additional intuition. Finally, it is important to note that the proof closely follows the ideas provided in \cite{OstMarCarRud22}.

\subsection{Necessary definitions}



We first introduce special queries, referred to as \emph{marked queries}, which are queries supplemented with a set of markings, the significance of which, is to keep track of what terms of the CQs are mapped to constants in a given database throughout the course of rewriting. A marked CQ is defined to be a pair $(q, M)$ where $q$ is a CQ and the set of markings $M$ is a subset of terms (i.e., variables and constants) of $q$. Given an instance $\inst$ we say that $\inst \models (q,M)$ $\iffi$ there exists a homomorphism $h$ mapping $q$ to $\inst$ such that for every term $t$ of $q$ we have $t \in M$ $\iffi$ $h(t)$ is a constant in $\inst$.

We now make a useful observation, which follows by considering the structures that arise when applying the rules from $\rsgrid$ to any given database. Afterward, we show that marked queries take on a special form when interpreted on the instance $\sg$, where $\db$ is an arbitrary database.

\begin{observation}\label{obs:buildup-for-properly-marked-chase}
Let $\db$ be a database. Then,
\begin{itemize}
    \item if $\hpred(t,t')$ or $\vpred(t,t')$ is an atom of $\sg$ and $t'\in \adom{\db}$, then $t\in \adom{\db}$,
    \item if there is a cycle $C \subseteq \sg$ that contains at least one constant, then $C \subseteq \db$,
    \item if $\rpred(t_1,t'), \rpred(t_2,t') \in \sg$ with $\rpred\in\{\hpred, \vpred\}$ and $t_1\in  \adom{\db}$, then $t_2\in  \adom{\db}$,
    \item if $\rpred(\const{a},t), \rpred(\const{b},t) \in \sg$ with $\rpred\in\{\hpred, \vpred\}$ and $\const{a}\neq \const{b}$ are constants, then $t$ is a constant.
\end{itemize}
\end{observation}

Keeping the above observation in mind, we define the notion of a {\em properly marked CQ}.

\begin{definition}\label{def:properly-marked}
A marked CQ $(q, M)$ is {\em properly marked} when:
\begin{itemize}
    \item if $\const{a}$ is a constant in $q$, then $\const{a} \in M$,
    \item if $\hpred(t,t')$ or $\vpred(t,t')$ is an atom in $q$ and $t' \in M$, then $t \in M$,
    \item if there is a cycle $C$ in $q$ that contains at least one constant, then all terms of $C$ are contained in $M$,
    \item if $\rpred(t_1,t')$ and $\rpred(t_2,t')$ are two atoms in $q$ with $\rpred\in\{\hpred, \vpred\}$ and $t_1 \in M$, then $t_2 \in M$,
    \item if $\rpred(\const{a},t)$ and $\rpred(\const{b},t)$ are two atoms in $q$ with $\rpred\in\{\hpred, \vpred\}$ and $\const{a}\neq \const{b}$ are constants, then $t \in M$.\defend
\end{itemize}
\end{definition}

\begin{observation}\label{obs:query-extendable-properly-marked}
If $\sg \models q$, then there exists a set of markings $M$ such that $(q,M)$ is properly marked and $\sg \models (q,M)$.
\end{observation}

\begin{proof} Suppose that $\sg \models q$. Then, we know that a homomorphism $h$ exists mapping terms from $q$ into $\sg$. Let $M$ be the set of terms from $q$ such that $h(t)$ is mapped to a constant in $\sg$. One can readily check if $(q,M)$ is properly marked, and if not, add the terms from $q$ to $M$ accordingly until a properly marked CQ $(q,M')$ is constructed such that $\sg \models (q,M)$.
\end{proof}

From \cref{obs:buildup-for-properly-marked-chase} and \cref{def:properly-marked} we immediately conclude the following:

\begin{corollary}
Let $(q, M)$ be a marked query, and $\db$ be any database. If $\sg \models \mq$, then $\mq$ is properly marked.
\end{corollary}

Let us refer to a marked CQ $(q, M)$ as {\em alive} given that it is properly marked and $M$ does not contain all terms of $q$. In the case where every term of a properly marked CQ is marked, we refer to the CQ as {\em dead}.

To prove~\cref{lem:rsg-is-fus} it suffices to prove the following lemma, which we direct our attention to in the subsequent section. 

\begin{lemma}\label{lem:rsg-is-fus-reformulated}
For any properly marked CQ $(q, \prmark)$ there exists a set of dead queries $\fcn{rew}(q,\prmark)$ such that for any database $\db$ the following holds: 
\begin{center}
$\sg \models (q,\prmark)$ \ $\iffi$ \ $\exists{Q' \in \fcn{rew}(q, \prmark)}, \db \models Q'.$
\end{center}
\end{lemma}

\subsection{Rewriting Alive Queries}

We begin this section by proving the following lemma, which will later prove useful in rewriting queries relative to $\rsgrid$.

\begin{lemma}\label{lem:exists-maximal}
If a marked CQ $(q,\prmark)$ is alive, then there exists a variable $x$ in $q$ such that no atom of the form $\rpred(x,t)$ occurs in $q$ with $\rpred\in \{\hpred,\vpred\}$.
\end{lemma}
\begin{proof} Let the marked CQ $(q,\prmark)$ be alive. Then, by the definition above, we know that there exists a set of terms in $q$ that are not a members of $M$. Since $(q,\prmark)$ is properly marked, by \cref{def:properly-marked}, we know that all such terms are not constants, and hence, all such terms must be variables. Let us refer to all such terms as \emph{unmarked variables}. If we consider the sub-instance of the CQ induced by all unmarked variables, we can see that (i) it forms a directed acyclic graph as all cycles of the CQ must be in $M$ by \cref{def:properly-marked}, which (ii) is non-empty as $(q,\prmark)$ is alive. Hence, there must exists a variable $x$ in $q$ such that no atom of the form $\rpred(x,t)$ occurs in $q$ with $\rpred\in \{\hpred,\vpred\}$, which corresponds to a sink in the directed acyclic graph described above.
\end{proof}

We will refer to unmarked variables -- such as those described in the above proof -- which possess no out-going edges, as {\em maximal}. These maximal variables will be used to initiate our rewriting process. 

\begin{lemma}\label{lem:maximal-cases}
Let $x$ be a maximal variable of an alive CQ $(q, \prmark)$. Then, one of the following holds:
\begin{enumerate}
    \item \label{lem:maximal-cases-one} $x$ occurs in exactly one atom  $\rpred(t,x)$, with $\rpred\in \{\hpred,\vpred\}$ for some term $t$ of $q$,
    \item \label{lem:maximal-cases-two} $x$ occurs in exactly two atoms $\hpred(t_{\hpred},x)$ and $\vpred(t_{\vpred},x)$ for some terms $t_{\hpred},t_{\vpred}$ of $q$,
    \item \label{lem:maximal-cases-three} there exist at least two distinct terms $t\neq t'$ such that $\rpred(t,x)$ and $\rpred(t',x)$ are atoms of $q$ with $\rpred\in \{\hpred,\vpred\}$.
\end{enumerate}
\end{lemma}


\medskip
\noindent
{\bf Rewriting Operations.}
We now make use of~\cref{lem:maximal-cases} to define three operations that compute rewritings. When exhaustively applied to a query, we obtain the rewritings required by \cref{lem:rsg-is-fus-reformulated}. Given an alive CQ $(q, \prmark)$ we define the following three operations:

\begin{definition}[$\cutop$]\label{def:cut}
Let $x$ be a maximal variable of $(q, \prmark)$ satisfying case \ref{lem:maximal-cases-one} of \cref{lem:maximal-cases}, i.e., $x$ occurs in exactly one atom $\rpred(t,x)$ of $q$ with $\rpred\in \{\hpred, \vpred\}$. Then, letting $\query = \query' \cup \{\rpred(t,x)\}$, we define
\[
  \ \ \ \ \cutop((q,\prmark), \set{x}) =
  \begin{cases}
     (q' \cup \{\top(\const{c})\}, \prmark)  & \text{if $t$ is a constant $\const{c}$} \\[-1ex]
      & \text{not occurring in $\query'$;}\ \ \ \ \ \ \\
     (q', \prmark) & \text{otherwise.}\defend
  \end{cases}
\]
\end{definition}

\begin{definition}[$\reduceop$]\label{def:reduce}
Let $x, t_{\hpred}$, and $t_{\vpred}$ be as in case \ref{lem:maximal-cases-two} of \cref{lem:maximal-cases}, meaning that there exist two atoms of the form $\hpred(t_{\hpred}, x)$ and $\vpred(t_{\vpred}, x)$ in $q$. We define $q'$ to be the same as $q$, but where $\hpred(t_{\hpred}, x)$ and $\vpred(t_{\vpred}, x)$ are replaced by $\hpred(x', t_{\vpred})$ and $\vpred(x', t_{\hpred})$, and where $x'$ is fresh. We now 
 define $\reduceop((q,\prmark), \set{x}) = \set{ (q', \prmark), (q', \prmark \cup \set{x'})}$.\defend
\end{definition}

\begin{definition}[$\mergeop$]\label{def:merge}
Let $x, t$, and $t'$ be as in case \ref{lem:maximal-cases-three} of \cref{lem:maximal-cases}. We let $q'$ be obtained from $q$ by replacing all occurrences of $t$ in $q$ with $t'$ and define $\mergeop((q, \prmark), \set{x, t, t'}) = (q' \cup \{\top(t)\}, \prmark)$.\defend
\end{definition}

Let us briefly describe how rewritings are computed by means of the above operations. 
 Given an alive CQ $(q,M)$, we know that one of the three rewriting operations will be applicable (by \cref{lem:maximal-cases}), which we note, may return a \emph{set} of rewritings if the $\reduceop$ operation is applied. We ultimately obtain the set $\fcn{rew}(q, \prmark)$ of rewritings of $(q,\prmark)$ by sequentially applying the three rewriting operations to all alive queries within the set of rewritings, until all marked queries in the set are dead, at which point, the rewriting process terminates. As shown below, we can always apply a rewriting operation to a CQ if it is alive (i.e., rewriting is complete), and also, applying such an operation preserves satisfaction on $\sg$ (i.e., rewriting is sound). 

\begin{observation}[Completeness]\label{obs:high-completeness}
If a CQ is alive then at least one of above operations can be applied to it.
\end{observation}
\begin{proof}
Follows from~\cref{def:cut,def:reduce,def:merge}, and~\cref{lem:maximal-cases}.
\end{proof}

Let us now prove the soundness of the above operations:

\begin{lemma}[Soundness --- $\cutop$]\label{lem:high-soundness-cut}
Let $\db$ be a database and let $\mquery$ be an alive CQ. If $x$ is as in \cref{def:cut}, then
$$\sg \models \mquery \ \text{\iffi} \ \sg \models \cutop(\mquery, \set{x}).$$
\end{lemma}

\begin{proof} Let $\mquery = (\query,\prmark)$ be an alive CQ and $\cutop(\mquery, \set{x}) = (q',M)$. The forward direction is trivial. For the backward direction, assume that $\sg \models \cutop(\mquery, \set{x})$. By our assumption, there exists a homormorphism $h$ from $\query'$ into $\sg$ such that for all terms $t$ in $\query'$, $t \in M$ \iffi $h(t)$ is a constant in $\sg$. Since the $\cutop$ operation was applied, we know that there exists exactly one atom of the from $\rpred(t,x)$ in $\query$ with $\rpred\in \{\hpred,\vpred\}$ and $x$ a maximal variable. We have two cases to consider: either (i) the term $t$ exists in $\query'$, or (ii) the term $t$ does not exist in $\query'$.

(i) Suppose that the term $t$ exists in $\query'$. Then, $h$ maps $t$ to some element $t'$ of $\sg$. Based on the functionality of the rules in $\rsgrid$, one can see that for every term $t''$ in $\sg$ and $\ppred \in \{\hpred,\vpred\}$, there exists a null $y$ such that $\ppred(t'',y) \in \sg$. Therefore, there exists a null $y$ such that $\rpred(t',y) \in \sg$. Let us define the homomorphism $h'$ such that $h'(s) = h(s)$ for all terms $s \neq x$ in $\query$ and $h'(x) = y$. It follows that $h'$ is a homomorphism from $\query$ to $\sg$, and thus $\sg \models \query$.

To complete the proof of this case, it suffices to show that for all terms $t$ in $\query$ that $t \in M$ \iffi $h'(t)$ is a constant in $\sg$. If we consider only the terms $s \neq x$ in $\query$, then the above equivalence holds as $h'(s) = h(s)$ for all such terms and we know that for all terms $s$ in $\query'$ that $s \in M$ \iffi $h(s)$ is a constant in $\sg$. The above equivalence also holds for the maximal variable $x$ since $x \not\in M$ and $h'(x) = y$ maps to a null (and not a constant) in $\sg$. This completes the proof of this case.

(ii) Suppose that the term $t$ does not exist in $\query'$. We know that either (a) the term $t$ is a constant $\const{c}$, or (b) the term $t$ is a variable $y$. We prove both cases below.

(a) Assume that $t$ is a constant $\const{c}$. If $\const{c}$ does not occur in the database $\db$, then $\sg \not\models \cutop(\mquery, \set{x})$, contrary to our assumption. Hence, we may assume that $\const{c}$ occurs in $\db$. By the (grow) rule in $\rsgrid$, we know that for the constant $\const{c}$ in $\db$, there exists a null $z$ such that $\rpred(\const{c},z) \in \sg$. Let us extend the homomorphism $h$ to a homomorphism $h'$ such that $h'(s) = h(s)$ for $s \neq x$ and $h'(x) = z$. It is straightforward to show that $h'$ is a homomorphism from $\query$ to $\sg$, and showing that for all terms $t'$ in $\query$ we have $t' \in M$ \iffi $h'(t')$ is a constant in $\sg$ is argued similarly as in case (i) above.

(b) Assume that $t$ is a variable $y$. Based on the (loop) rule in $\rsgrid$, we know that there exists a null $z$ such that $\rpred(z,z) \in \sg$. Let us now extend the homomorphism $h$ to a homomorphism $h'$ such that $h'(s) = h(s)$ for $s \not\in\{y,x\}$ and $h'(y) = h'(x) = z$. It is straightforward to show that $\sg \models \mquery$ by making use of $h'$.
\end{proof}

\begin{lemma}[Soundness --- $\reduceop$]\label{lem:high-soundness-reduce}
Let $\db$ be a database and let $Q$ be an alive CQ. If $x$ is as in \cref{def:reduce} and $\reduceop(\mquery, \set{x}) = \set{\mquery_{1}, \mquery_{2}}$, then
$$\sg \models \mquery \ \text{\iffi} \ \exists i \in \{1,2\}, \sg \models \mquery_{i}.$$
\end{lemma}

\begin{proof} Let $Q = (\query,M)$ and $\mquery_{i} = (\query',M_{i})$ with $i \in \{1,2\}$. 

For the forward direction, suppose that $\sg \models \mquery$. It follows that there exists a homomorphism $h$ such that $h$ maps $\query$ to $\sg$ and for all terms $t$ in $\query$, $t \in M$ \iffi $h(t)$ is a constant in $\sg$. By our assumption that $\reduceop$ was applied, we know that there exist two atoms of the form $\hpred(t_{\hpred},x)$ and $\vpred(t_{\vpred},x)$ in $\query$ with $x$ maximal and $x$ occurring in no other atoms of $\query$. We have two cases to consider: either (i) both $t_{\hpred}$ and $t_{\vpred}$ are constants, or (ii) at least one term is not a constant. We argue each case accordingly.

(i) Observe that if both $t_{\hpred}$ and $t_{\vpred}$ are constants, then since $\mquery$ is properly marked, we know that $x \in M$ by \cref{def:properly-marked}. However, this contradicts our assumption that $x$ is maximal, and so, this case cannot occur.

(ii) In the alternative case, by the functionality of the Skolem chase (see~\cref{sec:preliminaries}), the fact that our rule set is $\rsgrid$, and since one term maps to a null in $\sg$ via $h$, we know that the (grid) rule must have been applied for $\hpred(h(t_{\hpred}),h(x))$ and $\vpred(h(t_{\vpred}),h(x))$ to occur in $\sg$. It follows that for some term $t$, $\hpred(t,h(t_{\vpred}))$ and $\vpred(t,h(t_{\hpred}))$ must occur in $\sg$. Let us define the homomorphism $h'$ such that $h'(s) = h(s)$ for each term $s \neq x$ in $\query$ and $h'(x') = t$. Then, we have that $\sg \models \query'$. Note that if $t$ is a null, then $\sg \models Q_{1}$ and if $t$ is a constant, then $\sg \models Q_{2}$, thus showing the desired conclusion.

Let us now show the backward direction and assume that there exists an $i \in \{1,2\}$ such that $\sg \models \mquery_{i}$. We have that $\sg \models \query'$, meaning that there exists a homomorphism $h$ mapping $\query'$ into $\sg$. Due to the fact that $\reduceop$ was applied, it follows that there exist atoms $\hpred(x',t_{\vpred})$ and $\vpred(x',t_{\hpred})$ in $\query'$ that are mapped into $\sg$, that is to say, $\hpred(h(x'),h(t_{\vpred})),\vpred(h(x'),h(t_{\hpred})) \in \sg$. By the (grid) rule and the functionality of the Skolem chase, there exists a null $y$ such that $\hpred(h(t_{\hpred}),y),\vpred(h(t_{\vpred}),y) \in \sg$. Let us define the homomorphism $h'$ such that $h'(s) = h(s)$ for all terms $s \neq x'$ in $\query'$ and $h'(x) = y$. Then, $h'$ is a homomorphism from $\query$ to $\sg$ and it is straightforward to show that for all terms $t'$ in $\query$, $t' \in M$ \iffi $h'(t')$ maps to a constant in $\sg$.
\end{proof}

\begin{lemma}[Soundness --- $\mergeop$]\label{lem:high-soundness-merge}
Let $\db$ be a database and let $\mquery$ be an alive CQ. If $x, t$, and $t'$ are as in \cref{def:merge}, then
$$\sg \models \mquery \ \text{\iffi} \ \sg \models \mergeop(\mquery, \set{x, t, t'}).$$
\end{lemma}

\begin{proof} Let $Q = (\query,M)$, $\mergeop(\mquery, \set{x, t, t'}) = (\query',M )$, $\rpred \in \{\hpred, \vpred\}$, and $\query' = \query'' \cup \{\top(t)\}$.

For the forward direction, assume that $\sg \models \mquery$. It follows that there exists a homomorphism $h$ such that $h$ maps $\query$ into $\sg$ and for all terms $s$ in $\query$, $s \in M$ \iffi $h(s)$ is a constant in $\sg$. Also, since $\mergeop$ was applied, we know that there are at least two terms $t \neq t'$ such that $\rpred(t,x)$ and $\rpred(t',x)$ are in $\query$ and that $t'$ was substituted for $t$ in $\query$ to obtain $\query'$. Hence, $\rpred(h(t),h(x)), \rpred(h(t'),h(x)) \in \sg$, however, since we are using rules from $\rsgrid$ and by the functionality of the Skolem chase (see~\cref{sec:preliminaries}), it must be the case that either (i) $h(t)$, $h(x)$, and $h(t')$ all occur in $\db$, or (ii) $h(t) = h(x) = h(t') = y$ with $y$ the null generated by the (loop) rule. In case (i), $h(x)$ will be a constant as all terms of $\db$ are constants, meaning that $x$ will not be maximal, in contradiction to our assumption. Let us consider case (ii) then. 
 Note that in case (ii), the entire connected component of $\query$ containing $t$, $x$, and $t'$ is mapped to $y$ via $h$, meaning that the entire connected component of $\query'$ containing $x$ and $t'$ will be mapped to $y$ via $h$ as well. If we restrict our attention to the remaining connected components of $\query'$, then we can see that $h$ maps all such components homomorphically into $\sg$, showing that $h$ serves a homomorphism from $\query'$ to $\sg$. Last, we must show that for all terms $s'$ in $\query'$, $s' \in M$ \iffi $h(s')$ is a constant in $\sg$, however, this follows from the fact that $\query$ and $\query'$ have the same set of terms.

For the backward direction, suppose that $\sg \models \mergeop(\mquery, \set{x, t, t'})$. Then, there exists a homomorphism $h$ such that $h$ maps $\query'$ into $\sg$ and for all terms $s$ in $\query'$, $s \in M$ \iffi $h(s)$ is a constant in $\sg$. In addition, as $\mergeop$ was applied, there are at least two terms $t \neq t'$ such that $\rpred(t,x)$ and $\rpred(t',x)$ are in $\query$ with $t'$ substituted for $t$ in $\query$ to obtain $\query'$. Let us define a homomorphism $h'$ such that $h'(s) = s$ for all terms $s \neq t$ and $h'(t) = t'$. Then, $h'$ is a homomorphism from $\query$ to $\query'$, and since $h$ is a homomorphism from $\query'$ to $\sg$, we have that $h \circ h'$ is a homomorphism from $\query'$ to $\sg$. To finish the proof of the backward direction, we need to show that for all terms $s'$ in $\query$, $s' \in M$ \iffi $(h \circ h')(s')$ is a constant in $\sg$. 

Let $s'$ be a term in $\query$. Suppose that $s' \in M$. Then, if $s' \neq t$, we have that $h(s')$ is a constant in $\sg$, and since $h(s') = (h \circ h')(s')$, it follows that $(h \circ h')(s')$ is a constant in $\sg$. If $s' = t$, then by \cref{def:properly-marked}, $t' \in M$, from which it follows that $h(t') = (h \circ h')(t)$ is a constant in $\sg$. Let us now suppose that $(h \circ h')(s')$ is a constant in $\sg$. If $s' \neq t$, then $(h \circ h')(s') = h(s')$ is a constant in $\sg$, implying that $s' \in M$; however, if $s' = t$, then $(h \circ h')(s') = (h \circ h')(t) = h(t')$ is a constant in $\sg$, implying that $t' \in M$, so by \cref{def:properly-marked}, $t \in M$ as well.
\end{proof}

\medskip
\noindent
{\bf Termination.}
Given two queries $\mquery$ and $\mquery'$, we write $\mquery \stepop \mquery'$ \iffi $\mquery'$ can be obtained from $\mquery$ by applying $\reduceop$, $\cutop$, or $\mergeop$. We let $\stepop^*$ denote the transitive closure of $\stepop$.


\begin{lemma}[Termination]\label{lem:termination}
For any properly marked CQ $\mquery$ the set of properly marked queries $\set{Q' \;|\; \mquery \stepop^* \mquery'}$ is finite.
\end{lemma}

\begin{proof} Observe that $\cutop$ deletes the only atom containing the considered maximal variable and $\mergeop$ reduces the number of atoms containing the considered maximal variable. Moreover, neither $\cutop$ nor $\mergeop$ introduce converging atoms of the form $\hpred(t_{1},t)$ and $\vpred(t_{2},t)$ when applied. Since the operation $\reduceop$ reduces the number of such converging atoms in a properly marked CQ as well, and since all such operations either preserve or decrease the number of terms occurring the query, it follows that the rewriting process will eventually terminate. 
\end{proof}

By \cref{obs:high-completeness} (Completeness), \cref{lem:high-soundness-cut,lem:high-soundness-reduce,lem:high-soundness-merge} (Soundness), and \cref{lem:termination} (Termination), we may conclude~\cref{lem:rsg-is-fus-reformulated}, thus completing the proof of \cref{lem:rsg-is-fus}.